\documentclass[]{article}

\usepackage{hyperref}
	\hypersetup{pdftitle=Analysis of HC selection for MCTS in simultaneous move games}
	\hypersetup{pdfauthor=V. Kovařík and V. Lisý}
\usepackage[utf8]{inputenc}
\usepackage{amsmath,amsfonts,amssymb,amsthm}
	\numberwithin{equation}{section} 
\usepackage{xcolor}
\usepackage{enumerate}
\usepackage{graphicx}
\usepackage{mathptmx}
\usepackage{textcomp}
\usepackage{xcolor}
\usepackage{subcaption}
\usepackage{framed}
\usepackage[plain]{algorithm}
	\floatname{algorithm}{Alg.}
	\usepackage[noend]{algorithmic}[1]
\usepackage[english]{babel}
\usepackage{eqparbox}
\usepackage[round]{natbib}
	\setcitestyle{aysep={}}
\usepackage{thm-restate}
\usepackage{tikz}
	\usetikzlibrary{chains,scopes}
\usetikzlibrary{calc}
\usepackage{enumerate}
\usepackage{geometry}
\newcommand{\cA}{\mathcal{A}}

\newcommand{\cH}{\mathcal{H}}

\newcommand{\cZ}{\mathcal{Z}}

\newtheorem{theorem}{\protect\theoremname}[section]
\newtheorem{definition}[theorem]{\protect\definitionname}
\newtheorem{lemma}[theorem]{\protect\lemmaname}
\newtheorem{proposition}[theorem]{\protect\propositionname}
\newtheorem{example}[theorem]{\protect\examplename}
\newtheorem{corollary}[theorem]{\protect\corollaryname}
\newtheorem{remark}[theorem]{\protect\remarkname}
\newtheorem{problem}[theorem]{\protect\problemname}

\makeatother
\allowdisplaybreaks

\providecommand{\corollaryname}{Corollary}
\providecommand{\claimname}{Claim}
\providecommand{\definitionname}{Definition}
\providecommand{\lemmaname}{Lemma}
\providecommand{\notationname}{Notation}
\providecommand{\remarkname}{Remark}
\providecommand{\problemname}{Problem}
\providecommand{\propositionname}{Proposition}
\providecommand{\examplename}{Example}
\providecommand{\theoremname}{Theorem}
\providecommand{\conjecturename}{Conjecture}

\definecolor{darkgreen}{RGB}{0,125,0}
\newcounter{vlNoteCounter}

\newcounter{vkNoteCounter}

\newcounter{refNoteCounter}

\newsavebox{\AntiL}

\savebox{\AntiL}{
\begin{tikzpicture}[node distance = 5mm and 1cm,start chain=going right, leaf/.style={on chain,join}, dot/.style={fill=black,circle,scale=0.5,on chain, join}]
{[start chain = trunk]
\node [dot] {};
{[start branch = 1 going below]
  \node [leaf]{$\frac{1}{2}$}; 
}
\node [dot]{};
{[start branch = 2 going below]
  \node [leaf]{$\frac{D-2}{2(D-1)}$}; 
}
\node [dot]{};
{[start branch = 3 going below]
  \node [leaf]{$\frac{D-3}{2(D-1)}$}; 
}
\node [dot]{};
{[start branch = 2 going below]
  \node [leaf,minimum size = 7mm]{$\dots$}; 
}
\node [dot]{};
{[start branch = 4 going below]
  \node [leaf,minimum size = 6mm]{$0$}; 
}
\node [leaf]{1};
}
\end{tikzpicture}
}

\newsavebox{\LinBound}

\savebox{\LinBound}{
\begin{tikzpicture}[node distance = 5mm and 1cm,start chain=going right, leaf/.style={on chain,join}, dot/.style={fill=black,circle,scale=0.5,on chain, join}]
{[start chain = trunk]
\node [dot] {};
{[start branch = 1 going above]
  \node [leaf]{$u_4$};
 }
{[start branch = 1 going below]
  \node [leaf]{$0$};  
}
\node [dot]{};
{[start branch = 2 going above]
  \node [leaf]{$u_3$};
 }
{[start branch = 2 going below]
  \node [leaf]{$0$}; 
}
\node [dot]{};
{[start branch = 3 going above]
  \node [leaf]{$u_2$};
 }
{[start branch = 3 going below]
  \node [leaf]{$0$}; 
}
\node [dot]{};
{[start branch = 4 going above]
  \node [leaf]{$u_1$};
 }
{[start branch = 4 going below]
  \node [leaf]{$0$}; 
}
\node [dot]{};
{[start branch = 5 going below]
  \node [leaf,minimum size = 6mm]{$0$}; 
}
\node [leaf]{1};
}
\end{tikzpicture}
}

\begin{document}

\begin{titlepage}
\title{Analysis of Hannan Consistent Selection for Monte Carlo\\ Tree Search in Simultaneous Move Games}

\author{Vojt\v{e}ch Kova\v{r}\'{i}k, Viliam Lis\'{y}
		\medskip \\
        vojta.kovarik@gmail.com, viliam.lisy@agents.fel.cvut.cz
		\bigskip \\
        Artificial Intelligence Center, Department of Computer Science\\
		Faculty of Electrical Engineering, Czech Technical University in Prague\\
		Zikova 1903/4, Prague 6, 166 36, Czech Republic \\       
}
	
\date{2017}

\maketitle

\begin{abstract}
Hannan consistency, or no external regret, is a~key concept for learning in games.
An action selection algorithm is Hannan consistent (HC) if its performance is eventually as good as selecting the~best fixed action in hindsight.
If both players in a~zero-sum normal form game use a~Hannan consistent algorithm, their average behavior converges to a~Nash equilibrium (NE) of the~game.
A similar result is known about extensive form games, but the~played strategies need to be Hannan consistent with respect to the~counterfactual values, which are often difficult to obtain.
We study zero-sum extensive form games with simultaneous moves, but
otherwise perfect information. These games generalize normal form games and they are 
a special case of extensive form games.
We study whether applying HC algorithms in each decision point of these games directly to the~observed payoffs leads to convergence to a~Nash equilibrium.
This learning process corresponds to a~class of Monte Carlo Tree Search algorithms, which are popular for playing simultaneous-move games but do not have any known performance guarantees.
We show that using HC algorithms directly on the~observed payoffs is not sufficient to guarantee the~convergence. With an~additional averaging over joint actions, the~convergence is guaranteed, but empirically slower. We further define an~additional property of HC algorithms, which is sufficient to guarantee the~convergence without the~averaging and we empirically show that commonly used HC algorithms have this property.
\end{abstract}

\end{titlepage}

\setcounter{page}{2}

\section{Introduction\label{sec: Intro}}

Research on learning in games led to multiple algorithmic advancements, such as the~variants of the~counterfactual regret minimization algorithm \citep{zinkevich2007regret}, which allowed for achieving human performance in poker \citep{DeepStack,brown2018libratus}.
Learning in games has been extensively studied in the~context of normal form games, extensive form games, as well as Markov games \citep{fudenberg1998theory,littman1994markov}.

One of the~key concepts in learning in games is Hannan consistency \citep{hannan1957approximation,hart2000simple,cesa2006prediction}, also known as no external regret. An~algorithm for repetitively selecting actions from a~fixed set of options is Hannan consistent (HC), if its performance approaches the~performance of selecting the~best fixed option all the~time.
This property was first studied in normal form games known to the~players, but it has later been shown to be achievable also if the~algorithm knows only its own choices and the~resulting payoffs~\citep{Auer2003Exp3}.

If both players use a~Hannan consistent algorithm in self-play in a~zero-sum normal form game, the~empirical frequencies of their action choices converge to a~Nash equilibrium (NE) of the~game\footnote{All of the games discussed in the paper might have multiple Nash equilibria, and regret minimization causes the empirical strategy to approach the set of these equilibria. For brevity, we will write that a strategy ``converges to a NE'' to express that there is eventually always some (but not necessarily always the same) NE close to the strategy.} \citep{Blum07}. A~similar convergence guarantee can be provided for zero-sum imperfect information extensive form games \citep{zinkevich2007regret}. However, it requires the~algorithms selecting actions in each decision point to be Hannan consistent with respect to counterfactual values. Computing these values requires either an~expensive traversal of all states in the~game or a~sampling scheme with reciprocal weighting~\citep{lanctot2009monte} leading to high variance and slower learning theoretically~\citep{gibson2012generalized} and in larger games even practically~\citep{bosansky2016aij}. Therefore, using them should be avoided in simpler classes of games, where they are not necessary.

In this paper, we study zero-sum extensive form games with simultaneous moves, but otherwise perfect information (SMGs).
These games generalize normal form games but are a~special case of extensive form games.
Players in SMGs simultaneously choose an~action to play, after which the~game either ends and the~players receive a~payoff, or the~game continues to a~next stage, which is again an~SMG.
These games are extensively studied in game playing literature (see \citealt{bosansky2016aij} for a~survey). They are also closely related to Markov games \citep{littman1994markov}, but they do not allow cycles in the~state space and the~players always receive a~reward only at the~end of the~game.
We study whether applying HC algorithms in each decision point of an~SMG directly to the~observed payoffs leads to convergence to a~Nash equilibrium as in normal form games, or whether it is necessary to use additional assumptions about the~learning process.

Convergence properties of Hannan consistent algorithms in simultaneous-move games are interesting also because using a~separate HC algorithm in each decision point directly corresponds to popular variants of Monte Carlo Tree Search (MCTS) \citep{coulom2006efficient,kocsis2006} commonly used in SMGs. In SMGs, these algorithms have been used for playing card games \citep{teytaud2011upper,lanctot2013goof}, variants of simple computer games \citep{Perick12Comparison}, and in the~most successful agents for General Game Playing \citep{Finnsson08}.
However, there are no known theoretical guarantees for their performance in this class of games. The main goal of the paper is to understand better the conditions under which MCTS-like algorithms perform well when applied to SMGs.

\subsection{Contributions}
We focus on two player zero-sum extensive form games with simultaneous moves but otherwise perfect information.
We denote a~class of algorithms that learn a~strategy from a~sequence of simulations of self-play in these games as SM-MCTS. These algorithms use simple \emph{selection policies} to chose next action to play in each decision point based on statistics of rewards achieved after individual actions. We study whether Hannan consistency of the~selection policies is sufficient to guarantee convergence of SM-MCTS to a~Nash equilibrium of the~game. We also investigate whether a similar relation holds between the approximate versions of Hannan consistency and Nash equilibria.

In Theorem~\ref{thm: counterexample}, we prove a~negative result, which is in our opinion the~most surprising and interesting part of the~paper, which is the~fact that Hannan consistency alone does not guarantee a~good performance of the~standard SM-MCTS. We present a~Hannan consistent selection policy that causes SM-MCTS to converge to a~solution far from the~equilibrium.

We show that this could be avoided by a slight modification to SM-MCTS (which we call SM-MCTS-A), which updates the~selection policies by the~average reward obtained after a~joined action of both players in past simulations, instead of the~result of the~current simulation.
In Theorem~\ref{thm: SM-MCTS-A convergence}, we prove that SM-MCTS-A combined with any Hannan consistent (resp. approximately HC) selection policy with guaranteed exploration converges to a~subgame-perfect Nash equilibrium (resp. an~approximate subgame-perfect NE) in this class of games.
For selection policies which are only approximately HC, we present bounds on the~eventual distance from a~Nash equilibrium.

In Theorem~\ref{thm: SM-MCTS convergence}, we show that under additional assumptions on the~selection policy, even the~standard SM-MCTS can guarantee the~convergence. We do this by defining the~property of having unbiased payoff observations (UPO), and showing that it is a~sufficient condition for the~convergence.
We then empirically confirm that the~two commonly used Hannan consistent algorithms, Exp3 and regret matching, satisfy this property, thus justifying their use in practice.
We further investigate the~empirical speed of convergence and show that the~eventual distance from the~equilibrium is typically much better than the~presented theoretical guarantees.
We empirically show that SM-MCTS generally converges as close to the~equilibrium as SM-MCTS-A, but does it faster.

Finally, we give theoretical grounds for some practical improvements, which may be used with SM-MCTS, but have not been formally justified. These include removal of exploration samples from the~resulting strategy and the~use of average strategy instead of empirical frequencies of action choices.
All presented theoretical results trivially apply also to perfect information games with sequential moves. The~counterexample in Theorem~\ref{thm: counterexample}, on the~other hand, applies to Markov and general extensive form games.

\subsection{Article outline}
In Section~\ref{sec: background} we describe simultaneous-move games, the~standard class of SM-MCTS algorithms, and our modification SM-MCTS-A. We then describe the~multi-armed bandit problem, the~definition of Hannan consistency, and explain two of the~common Hannan consistent bandit algorithms (Exp3 and regret matching).
We also explain the~relation of the~studied problem with counterfactual regret minimization, multi-agent reinforcement learning, and shortest path problems.

In Section \ref{sec: application to SM-MCTS(-A)}, we explain the relation of SM-MCTS setting with the~multi-armed bandit problem. In this section, we also define the technical notation which is used in the proofs of our main results.

In Section~\ref{sec:Counterexample}, we provide a~counterexample showing that for general Hannan consistent algorithms, SM-MCTS does not necessarily converge. While interesting in itself, this result also hints at which modifications are sufficient to get convergence.
The~positive theoretical results are presented in Section~\ref{sec: convergence}. First, we consider the~modified SM-MCTS-A and present the~asymptotic bound on its convergence rate.
We follow by defining the~unbiased payoff observations property and proving the~convergence of SM-MCTS based on HC selection policies with this property. We then present an~example which gives a~lower bound on the~quality of a~strategy to which SM-MCTS(-A) converges.

In Section~\ref{sec: exploitability}, we make a~few remarks about which strategy should be considered as the~output of SM-MCTS(-A).
In Section~\ref{sec:Experimental}, we present an~empirical investigation of convergence of SM-MCTS and SM-MCTS-A, as well as empirical evidence supporting that the commonly used HC-algorithms guarantee the~UPO property. Finally, Section~\ref{sec:Discussion} summarizes the~results and highlights open questions which might be interesting for future research.
Table~\ref{tab:notation} then recapitulates the notation and abbreviations used throughout the paper.

\section{Background\label{sec: background}}
We now introduce the~game theory fundamentals and notation used throughout the paper. We define simultaneous-move games, describe
the class of SM-MCTS algorithms and its modification SM-MCTS-A, and afterward, we discuss existing selection policies and their properties.

\subsection{Simultaneous move games}\label{sec:SM games}

A finite two-player zero-sum game with perfect information and simultaneous moves
can be described by a~tuple $G = (\mathcal{H},\mathcal{Z},\mathcal{A},\mathcal{T},u,h_{0})$,
where $\mathcal{H}$ is a~set of inner states and $\mathcal{Z}$ denotes the~terminal states\footnote{In our analysis, we will assume that the~game does not contain chance nodes, where the~next node is randomly selected by "nature" according to some probability distribution. We do this purely to reduce the~amount of technical notation required -- as long as the state-space is finite, all of the~results hold even if the~chance nodes are present. This can be proven in a~straightforward way.}.
$\mathcal{A}=\mathcal{A}_{1}\times\mathcal{A}_{2}$ is the~set of
joint actions of individual players and we denote $\mathcal{A}_{1}(h)=\{1,\dots, m^{h}\}$ and $\mathcal{A}_{2}(h)=\{1,\dots, n^{h}\}$ the~actions available to
individual players in state $h\in\mathcal{H}$. The~game begins in an~initial state $h_{0}$. The~transition function $\mathcal{T}:\mathcal{H}\times\mathcal{A}_{1}\times\mathcal{A}_{2}\mapsto\mathcal{H} \cup\mathcal{Z}$ defines the~successor state given a~current state and actions for both players. For brevity, we denote $\mathcal{T}(h,i,j)\equiv h_{ij}$.
We assume that all the~sets are finite and that the~whole setting is modeled as a~tree.\footnote{That is, if $(h,i,j)\neq(h',i',j')$ then $h_{ij}\neq h_{i'j'}$.}
The payoffs of player 1 and 2 are defined by the utility function $u = (u_1,u_2) :\mathcal{Z}\to [0,1]^2$. We assume constant-sum games (which are equivalent to zero-sum games): $\forall z\in\mathcal{Z},u_{2}(z)=1-u_{1}(z).$

A \textit{matrix game} is a~special case of the~setting above, where $G$ only has a~single stage, $h$ -- we have $\mathcal{H}=\{h\}$. In this case we can simplify the~notation and represent the~game by the~matrix $M=(v^M_{ij})_{i,j}\in[0,1]^{m\times n}$, where for $(i,j)\in\mathcal{A}_1\times\mathcal{A}_2=\{1,\dots,m\}\times\{1,\dots,n\}$, we have $u_1(h_{ij})=v^M_{ij}$. In other words,  $v^M_{ij}$ corresponds to the~payoff received by player 1 if player 1 chooses the~row $i$ and player 2 chooses the~column $j$.
A \textit{strategy} $\sigma_{p}\in\Delta(\mathcal{A}_{p})$ is a~distribution over the
actions in $\mathcal{A}_{p}$. If $\sigma_{1}$ is represented as
a row vector and $\sigma_{2}$ as a~column vector, then the~expected
value to player 1 when both players play with these strategies is
$u_{1}(\sigma_{1},\sigma_{2})=\sigma_{1}M\sigma_{2}$. Given a~profile
$\sigma=(\sigma_{1},\sigma_{2})$, define the~utilities against best
response strategies to be $u_{1}(br,\sigma_{2})=\max_{\sigma_{1}'\in\Delta(\mathcal{A}_{1})}\sigma_{1}'M\sigma_{2}$
and $u_{1}(\sigma_{1},br)=\min_{\sigma_{2}'\in\Delta(\mathcal{A}_{2})}\sigma_{1}M\sigma_{2}'$.
A strategy profile $(\sigma_{1},\sigma_{2})$ is an~$\epsilon$\textit{-Nash equilibrium} of the~matrix game $M$ if and only if 
\begin{equation}
u_{1}(br,\sigma_{2})-u_{1}(\sigma_{1},\sigma_{2})\leq\epsilon\hspace{1cm}\mbox{and}\hspace{1cm}u_{1}(\sigma_{1},\sigma_{2})-u_{1}(\sigma_{1},br)\leq\epsilon\label{eq:nfgNE}
\end{equation}

A \textit{behavioral strategy} for player $p$ is a~mapping from states
$h\in\mathcal{H}$ to a~probability distribution over the~actions
$\mathcal{A}_{p}(h)$, denoted $\sigma_{p}(h)$. Given a~profile $\sigma=(\sigma_{1},\sigma_{2})$,
define the~probability of reaching a~terminal state $z$ under $\sigma$
as $\pi^{\sigma}(z)=\pi^{\sigma}_{1}(z)\pi^{\sigma}_{2}(z)$, where each $\pi^{\sigma}_{p}(z)$ is a~product of probabilities of the~actions taken by player $p$
along the~path to $z$. Define $\Sigma_{p}$ to be the~set of behavioral
strategies for player $p$. Then for any strategy profile $\sigma=(\sigma_{1},\sigma_{2})\in\Sigma_{1}\times\Sigma_{2}$
we define the~expected utility of the~strategy profile (for player
1) as 
\begin{equation}
u(\sigma)=u(\sigma_{1},\sigma_{2})=\sum_{z}\pi^{\sigma}(z)u_{1}(z)
\end{equation}
An $\epsilon$-Nash equilibrium profile ($\sigma_{1},\sigma_{2}$)
in this case is defined analogously to (\ref{eq:nfgNE}). In other
words, none of the~players can improve their utility by more than
$\epsilon$ by deviating unilaterally. If $\sigma=(\sigma_{1},\sigma_{2})$
is an~exact Nash equilibrium (an~$\epsilon$-NE with $\epsilon=0$),
then we denote the~unique value of $G$ as $v^G=v^{h_0}=u(\sigma_{1},\sigma_{2}) = \min_{\sigma'_2} \max_{\sigma'_1} u(\sigma'_1,\sigma'_2)$ (because of the last identity, NE in this setting could also be called ``minimax optima'').

A \textit{subgame rooted in state} $h\in\mathcal{H}$ is the~game $(\mathcal{N},\mathcal{H},\mathcal{Z},\mathcal{A},\mathcal{T},u,h)$. This is the~same game as the~original one, except that it starts at the~node $h$ instead of $h_0$. We denote by $v^{h}$ the~value of this subgame. An~$\epsilon$-Nash equilibrium profile $\sigma$ is called \textit{subgame-perfect} if $\sigma$ is  $\epsilon$-Nash equilibrium of the~subgame rooted at $h$ for every $h\in\mathcal{H}$ (even at the~nodes which are unreachable under this strategy).

Two-player perfect information games with simultaneous moves are sometimes
appropriately called \textit{stacked matrix games} because at every
state $h$ there is a~joint action set $\mathcal{A}_{1}(h)\times\mathcal{A}_{2}(h)$
where each joint action $(i,j)$ either leads to a~terminal state with utility $u_1(h_{ij})$ or to a~subgame rooted in $h_{ij}$. This subgame is itself another stacked matrix game and its unique value $v^{h_{ij}}$ can be determined by backward induction (see Figure~\ref{fig:tree}). Thus finding the~optimal strategy at $h$ is the~same as finding the~optimal strategy in the~matrix game $(v^h_{ij})_{i,j}$, where $v^h_{ij}$ is either $u_1(h_{ij})$ or $v^{h_{ij}}$.

\begin{figure}
\centering \includegraphics[width=0.6\textwidth]{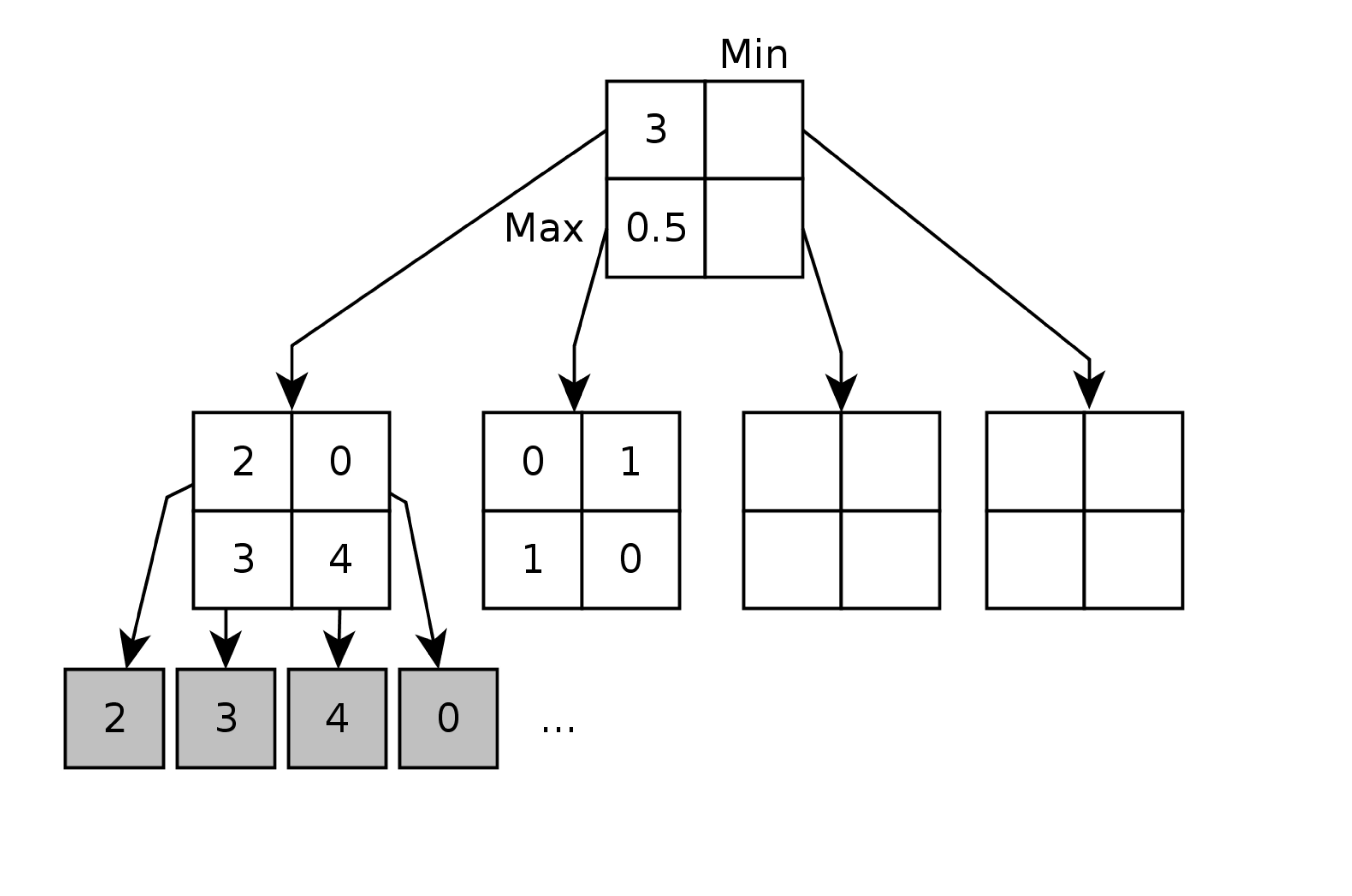}
\caption{Example game tree of a~game with perfect information and simultaneous
moves. Only the~leaves contain actual rewards  -- the~values in
the inner nodes are achieved by optimal play in the~corresponding
subtree, they are not part of the~definition of the~game.}\label{fig:tree}
\end{figure}

\subsection{Simultaneous Move Monte Carlo Tree Search} \label{section: SM-MCTS definition}

Monte Carlo Tree Search (MCTS) is a~simulation-based state space search
algorithm often used in game trees.
The main idea is to iteratively run simulations from the current state $h_0$ until the~end of the game, incrementally growing a~tree $\mathbf T \subset \cH \cup \cZ$ rooted at the~current state.
In the~basic form of the~algorithm, $\mathbf T$ only contains $h_0$ and a~single leaf is added each iteration.
Each iteration starts at $h_0$.
If MCTS encounters a node $h \in \mathbf T$ whose all children in $\cH \cup \cZ$ are already in $\mathbf T$, it uses the statistics maintained at $h$ to select one of its children that it then transitions to.
If MCTS encounters a node $h \in \mathbf T$ that has children that aren't yet in $\mathbf T$, it adds one of them to $\mathbf T$ and transitions to it.
Then we apply a~rollout policy (for example, a~random action selection) from this new leaf of $\mathbf T$ to some terminal state of the~game.
The outcome of the~simulation is then returned as a~reward to the~new leaf and all its predecessors.

In Simultaneous Move MCTS (SM-MCTS), the~main difference is that a~joint action of both players is selected and used to transition to a~following state. The~algorithm has been previously
applied, for example in the~game of Tron~\citep{Perick12Comparison}, Urban Rivals~\citep{teytaud2011upper}, and in general game-playing~\citep{Finnsson08}.
However, guarantees of convergence to a~NE remain unknown, and \citet{Shafiei09} show that the~most popular selection policy (UCB) does not converge to a~NE, even in a~simple one-stage game.
The convergence to a~NE depends critically on the~selection and update policies applied. We describe variants of two popular selection algorithms in Section \ref{sec: bandits}. 

\subsubsection*{Description of SM-MCTS algorithm}
In Algorithm~\ref{alg:smmcts}, we present a~generic template of MCTS algorithms for simultaneous-move games (SM-MCTS). We then proceed to explain how specific algorithms are derived from this template. 

\renewcommand{\algorithmiccomment}[1]{\hfill\eqparbox{COMMENT}{\texttt{// #1}}}
\newcommand{\LineIf}[2]{\STATE \algorithmicif\ {#1}\ \algorithmicthen\ {#2}}
\begin{algorithm}[b]
\centering
\begin{minipage}{0.9\linewidth}
\textbf{SM-MCTS(}$h$ -- current state of the~game\textbf{)}
\begin{algorithmic}[1]
\IF[current state is terminal]{$h \in \mathcal{Z}$}
	\RETURN $u_1(h)$ \label{alg:smmcts:terminal}
\ENDIF
\IF[current state is in memory]{$h \in \mathbf{T}$}
  \STATE $(i, j) \leftarrow$ \emph{\underline{Select}}$(h)$\label{alg:smmcts:select}
  \STATE $h' \leftarrow h_{ij}$\COMMENT{get the~next game state} \label{alg:smmcts:getnext}
  \STATE $x \leftarrow $ SM-MCTS($h'$)\label{alg:smmcts:reccall}
  \STATE \emph{\underline{Update}}$(h,i,j,x)$\label{alg:smmcts:up}
  \RETURN $x$\label{alg:smmcts:return}
\ELSE[state is not in memory]
  \STATE $\mathbf{T} \leftarrow \mathbf{T} \cup \lbrace h \rbrace$\COMMENT{add state to memory}\label{alg:smmcts:addstate}
  \STATE $x \leftarrow$ Rollout($h$)\COMMENT{finish the~game randomly} \label{alg:smmcts:rollout}
  \RETURN $x$
\ENDIF
\end{algorithmic}
\end{minipage}
\caption{Simultaneous Move Monte Carlo Tree Search}\label{alg:smmcts}
\end{algorithm}

Algorithm~\ref{alg:smmcts} describes a~single iteration of SM-MCTS. $\mathbf{T}$ represents the~incrementally built MCTS tree, in which each state is represented by one node. Every node $h\in \mathbf T$ maintains algorithm-specific statistics about the~iterations that previously used this node. In the~terminal states, the~algorithm returns the~value of the~state for the~first player (line~\ref{alg:smmcts:terminal}). If the~current state has a~node in the~current MCTS tree $\mathbf{T}$, the~statistics in the~node are used to select an~action for each player (line~\ref{alg:smmcts:select}). These actions are executed (line~\ref{alg:smmcts:getnext}) and the~algorithm is called recursively on the~resulting state (line~\ref{alg:smmcts:reccall}). The~result of this call is used to update the~statistics maintained for state $h$ (line~\ref{alg:smmcts:up}). If the~current state is not stored in tree $\mathbf{T}$, it is added to the~tree (line~\ref{alg:smmcts:addstate}) and its value is estimated using the~rollout policy (line~\ref{alg:smmcts:rollout}). The~rollout policy is usually uniform random action selection until the~game reaches a~terminal state, but it can also be based on domain-specific knowledge. Finally, the~result of the~Rollout is returned to higher levels of the~tree.

The template can be instantiated by specific implementations of the~updates of the~statistics on line~\ref{alg:smmcts:up} and the~selection based on these statistics on line \ref{alg:smmcts:select}. Selection policies can be based on many different algorithms, but the~most successful ones use algorithms for solving the multi-armed bandit problem introduced in Section~\ref{sec: bandits}.
Specifically, we firstly decide on an~algorithm $A$ and choose all of its parameters except for the~number of available actions. Then we run a~separate instance of this algorithm for each node $h \in \mathbf{T}$ and each of the~players. In each node, the~action for each player is selected based only on the history. The~update procedure then uses the~values $u_1(h)$ and and $1-u_1(h)$ to update the~statistics of player 1 and 2 at $h$ respectively (this way, both the~values are in the~interval $[0,1]$).

This work assumes that, except for the SM-MCTS back-propagation, the selection algorithms do not communicate with each other in any way.
We refrain from analyzing the (potentially more powerful) algorithms which have access to non-local variables (e.g. a~global clock, abstraction-based information sharing) because such analysis would be significantly more complicated, but also because the primary goal of this paper is to understand the limits of the simpler selection policies.

\subsubsection*{SM-MCTS-A algorithm}

\renewcommand{\algorithmiccomment}[1]{\hfill\eqparbox{COMMENT}{\texttt{// #1}}}
\begin{algorithm}[b]
\centering
\begin{minipage}{0.9\linewidth}
\textbf{SM-MCTS-A(}$h$ -- current state of the~game\textbf{)}
\begin{algorithmic}[1]
\IF[current state is terminal]{$h \in \mathcal{Z}$}
	\RETURN $(u_1(h),u_1(h))$ \label{alg:smmcts-A:terminal}
\ENDIF
\IF[current state is in memory]{$h \in \mathbf{T}$}
  \STATE $(i, j) \leftarrow$ \emph{\underline{Select}}$(h)$		\COMMENT{generate actions}
  \STATE $t^h \leftarrow t^h+1$  				\COMMENT{increase visit counter}\label{alg:smmcts-A:updatet}
  \STATE $(x,\overline{x}) \leftarrow $ SM-MCTS($h_{ij}$)	\COMMENT{get the latest and average reward}\label{alg:smmcts-A:reccall}
  \STATE \emph{\underline{Update}}$(h,i,j,\overline x)$ \COMMENT{update by the average reward}\label{alg:smmcts-A:update}
  \STATE $G^h \leftarrow G^h + x$ 			\COMMENT{update cumulative reward}\label{alg:smmcts-A:updateG}
  \RETURN $(x, G^h/t^h)$				\COMMENT{return reward \& current state average reward}\label{alg:smmcts-A:return}
\ELSE[state is not in memory]
  \STATE $\mathbf{T} \leftarrow \mathbf{T} \cup \lbrace h \rbrace$\COMMENT{add state to memory}
  \STATE $t^h$, $G^h$ $\leftarrow 0$	\COMMENT{initialize visit counter and cumulative reward}
  \STATE $x \leftarrow$ Rollout($h$)\COMMENT{finish the~game randomly}
  \RETURN $(x,x)$
\ENDIF
\end{algorithmic}
\end{minipage}
\caption{SM-MCTS-A, a variant of SM-MCTS which is guaranteed to converge to an equilibrium (under the assumptions of Theorem \ref{thm: SM-MCTS-A convergence}).}\label{alg:smmcts-A}
\end{algorithm}

We also propose an~``averaged'' variant of the~algorithm, which we denote as SM-MCTS-A (see Algorithm \ref{alg:smmcts-A}). Its main advantage over SM-MCTS is that it guarantees convergence to a~NE under much weaker conditions (Theorem \ref{thm: SM-MCTS-A convergence}).
Focusing on the differences between these algorithms helps us better understand what is missing to guarantee convergence of SM-MCTS.

SM-MCTS-A works similarly to SM-MCTS, except that during the back-propagation phase, every visited node $h\in \mathbf{T}$ sends back to its parent the~average reward obtained from $h$ so far (denoted $G^h/t^h$). The parent is then updated by this ``averaged'' reward $G^h/t^h$, rather than by the current reward $x$ (line~\ref{alg:smmcts-A:update}). In practice, this is achieved by doing a bit of extra book-keeping --- storing the cumulative reward $G^h$ and the number of visits $t^h$ and updating them during each visit of $h$ (lines~\ref{alg:smmcts-A:updatet} and \ref{alg:smmcts-A:updateG}) and back-propagating both the current and average rewards (line~\ref{alg:smmcts-A:return}).

We note that in our previous work \citep{lisy2013convergence} we prove a~result similar to Theorem \ref{thm: SM-MCTS-A convergence} here. However, the~algorithm that we used earlier is different from SM-MCTS-A algorithm described here. In particular, SM-MCTS-A uses averaged values for decision making in each node but propagates backward the~non-averaged values (unlike the~previous version, which also updates the~selection algorithm based on the~averaged values, but then it propagates backward these averaged numbers -- and on the~next level, it takes averages of averages and so on). Consequently, this new version is much closer to the~non-averaged SM-MCTS used in practice, and it has faster empirical convergence.

\subsection{Multi-armed bandit problem} \label{sec: bandits}
The multi-armed bandit (MAB) problem is one of the~basic models in online learning. It often serves as the~basic model for studying fundamental trade-offs between exploration and exploitation in an~unknown environment \citep{auer1995gambling,auer2002finite}. In practical applications, the~algorithms developed for this model have recently been used in online advertising \citep{pandey2007}, convex optimization \citep{flaxman2005online}, and, most importantly for this paper, in Monte Carlo tree search algorithms \citep{kocsis2006,browne2012survey,gelly2011monte,teytaud2011upper,coulom2007efficient}.
More details can be found in an~extensive survey of the~field by \citet{bubeck2012survey}.

\begin{definition}[Adversarial multi-armed bandit problem]\label{def: MABproblem}
Multi-armed bandit problem is specified by a~set of actions $\{1,\dots, K\}$ and a~sequence of reward vectors $x(1), x(2), \dots$, where $x(t) = ( x_i(t) )_{i=1}^K \in[0,1]^K$ for each $t\in\mathbb N$. In each time step, an~agent selects an~action $i(t)$ and receives the~payoff $x_{i(t)}(t)$.

\emph{Note that the~agent does not observe the~values $x_{i}(t)$ for $i \neq i(t)$. There are many special cases and generalizations of this setting (such as the~stochastic bandit problem), however, in this paper, we only need the~following variant of this concept:}

The (adaptive) \emph{adversarial} MAB problem $(P)$ is identical to the~setting above, except that each $x_i(t+1)$ is a random variable that might depend on $x(1),\dots,x(t)$ and $i(1),\dots,i(t)$.

A bandit algorithm is any procedure which takes as an~input the~number $K\in\mathbb N$, the~sequence of actions $i(1),...,i(t)$ played thus far and the~rewards
received as a~result and returns the~next action $i(t+1)\in\{1,...,K\}$ to be played.
\end{definition}

Bandit algorithms usually attempt to optimize their behavior with respect to some notion of regret. Intuitively, their goal is the~minimization of the~difference between playing the~strategy given by the~algorithm and playing some baseline strategy, which can use information not available to the~agent. For example, the~most common notion of regret is the~\emph{external regret}, which is the~difference between playing according to the~prescribed strategy and playing the~fixed optimal action all the~time.

\begin{definition}[External Regret]\label{def:extRegret}
The \emph{external regret} for playing a~sequence of actions $i(1)$, $i(2)$, \dots , $i(t)$ in $(P)$ is defined as
\[ R(t) = \max_{i=1,\dots,K}\sum_{s=1}^t x_i(s) - \sum_{s=1}^t x_{i(s)}(s). \]
By $r(t)$ we denote the~average external regret $r(t):=\frac 1 t R(t)$.
\end{definition}

\subsection{Hannan consistent algorithms}
A desirable goal for any bandit algorithm is the~classical notion of Hannan consistency. Having this property means that for high enough $t$, the~algorithm performs nearly as well as it would if it played the~optimal constant action since the~beginning.
\begin{definition}[Hannan consistency]\label{def:HC}
An algorithm is $\epsilon$-Hannan consistent for some $\epsilon\geq0$ if $\limsup_{t\rightarrow\infty} r(t)\leq \epsilon$ holds with probability 1, where the~``probability'' is understood with respect to the~randomization of the~algorithm. An~algorithm is Hannan consistent if it is 0-Hannan consistent.
\end{definition}

We now present regret matching and Exp3, two of the~$\epsilon$-Hannan consistent algorithms previously used in MCTS. The~proofs of Hannan consistency of variants of these two algorithms, as well as more related results, can be found in a~survey by \citet[Section 6]{cesa2006prediction}. The~fact that the~variants presented here are $\epsilon$-HC is not explicitly stated there, but it immediately follows from the~last inequality in the~proof of Theorem 6.6 in the~survey.

\subsubsection{Exponential-weight algorithm for Exploration and Exploitation \label{sec:exp3}}

\begin{algorithm}
\centering
\begin{minipage}{0.9\linewidth}
\begin{algorithmic}[1]
\REQUIRE{$K$ - number of actions; $\gamma$ - exploration parameter}
\STATE $\forall i \, : \ G_i \leftarrow 0$\COMMENT{initialize cumulative sum estimates}
\FOR[in each iteration] {$t \leftarrow 1,2,\dots$}
  \STATE $\forall i \, : \ \mu_{i} \leftarrow \frac{\exp(\tfrac{\gamma}{K} G_i)}{\sum_{j=1}^K \exp(\tfrac{\gamma}{K} G_j)}$\COMMENT{compute the~new strategy}
  \STATE $\forall i \, : \mu'_i \leftarrow (1-\gamma)\mu_i+\tfrac{\gamma}{K}$\COMMENT{add uniform exploration}
  \STATE Sample action $i(t)$ from distribution $\mu'$ and receive reward $r$
  \STATE $G_{i(t)} \leftarrow G_{i(t)} + r / \mu'_{i(t)}$\COMMENT{update the~cumulative sum estimates}
\ENDFOR
\end{algorithmic}
\end{minipage}
\caption{Exponential-weight algorithm for Exploration and Exploitation (Exp3) algorithm for regret minimization in adversarial bandit setting}
\label{alg:exp3}
\end{algorithm}

The most popular algorithm for minimizing regret in adversarial bandit setting is the~Ex\-ponential-weight algorithm for Exploration and Exploitation (Exp3) proposed by \cite{Auer2003Exp3}, further improved by \cite{stoltz2005incomplete} and then yet further by \citet[Sec. 3]{bubeck2012survey}. The~algorithm has many different variants for various modifications of the~setting and desired properties. We present a~formulation of the~algorithm based on the~original version in Algorithm~\ref{alg:exp3}.

Exp3 stores the~estimates of the~cumulative reward of each action over all iterations, even those in which the~action was not selected. In the~pseudo-code in Algorithm~\ref{alg:exp3}, we denote this value for action $i$ by $G_i$. It is initially set to $0$ on line~1. In each iteration, a~probability distribution $\mu$ is created proportionally to the~exponential of these estimates. The~distribution is combined with a~uniform distribution with probability $\gamma$ to ensure sufficient exploration of all actions (line~4). After an~action is selected and the~reward is received, the~estimate for the~performed action is updated using \emph{importance sampling} (line~6): the~reward is weighted by one over the~probability of using the~action.
As a~result, the~expected value of the~cumulative reward estimated only from the~time steps where the~agent selected the~action is the~same as the~actual cumulative reward over all the~time steps.

In practice, the~optimal choice of the~exploration parameter $\gamma$ strongly depends on the~computation time available for each decision and the~specific domain \citep{Tak14smmcts}. However, the~amount of exploration $\gamma$ directly translates to $\gamma$-Hannan consistency of the~algorithm.
We later show that, asymptotically, smaller $\gamma$ yields smaller worst-case error when the~algorithm is used in SM-MCTS.

\subsubsection{Regret matching \label{sec:rm}}
An alternative learning algorithm that allows minimizing regret in adversarial bandit setting is regret matching \citep{hart2001reinforcement}, later generalized as \emph{polynomially weighted average forecaster} \citep{cesa2006prediction}. Regret matching (RM) corresponds to selection of the~parameter $p=2$ in the~more general formulation. It is a~general procedure originally developed for playing known general-sum matrix games in \cite{hart2000simple}. The~algorithm computes, for each action in each step, the~regret for not playing another fixed action every time the~action has been played in the~past. The~action to be played in the~next round is selected randomly with probability proportional to the~positive portion of the~regret for not playing the~action.

The average strategy\footnote{The average strategy is defined as $\bar \mu(t) := \frac{1}{t} \sum_{s=1}^t \mu(s)$, where $\mu(s)$ is the strategy used at iteration $s$.} resulting from this procedure has been shown to converge to the~set of coarse correlated equilibria in general-sum games. As a~result, it converges to a~Nash equilibrium in a~zero-sum game. The~regret matching procedure in \cite{hart2000simple} requires the~exact information about all utility values in the~game, as well as the~action selected by the~opponent in each step. In \cite{hart2001reinforcement}, the~authors modify the~regret matching procedure and relax these requirements. Instead of computing the~exact values for the~regrets, the~regrets are estimated in a~similar way as the~cumulative rewards in Exp3. As a~result, the~modified regret matching procedure is applicable to the MAB problem.

\begin{algorithm}
\centering
\begin{minipage}{0.95\linewidth}
\begin{algorithmic}[1]
\REQUIRE{$K$ - number of actions; $\gamma$ - the~amount of exploration}
\STATE $\forall_i \, : \ R_i \leftarrow 0$\COMMENT{initialize regret estimates}
\FOR[in each iteration] {$t \leftarrow 1,2,\dots$}
  \STATE $\forall i \, : \ R_i^+ \leftarrow \max\{0,R_i\}$ 
  \IF[all regrets are non-positive] {$\sum_{j=1}^K R_j^+ = 0$}
    \STATE $\forall i \, : \ \mu_i \leftarrow 1/K$\COMMENT{use uniform strategy}
  \ELSE[otherwise, update the strategy based on regrets]
    \STATE $\forall i \, : \ \mu_i \leftarrow (1-\gamma)\frac{R_i^+}{\sum_{j=1}^K R_j^+} + \frac{\gamma}{K}$
  \ENDIF
  \STATE Sample action $i(t)$ from distribution $\mu$ and receive reward $r$
  \STATE $\forall i \, : \ R_i \leftarrow R_i - r$ \COMMENT{update regrets}
  \STATE $R_{i(t)} \leftarrow R_{i(t)} + r / \mu_{i(t)}$
\ENDFOR
\end{algorithmic}
\end{minipage}
\caption{Regret matching variant for regret minimization in adversarial bandit setting.}
\label{alg:banditRM}
\end{algorithm}

We present the~algorithm in Algorithm~\ref{alg:banditRM}.
The algorithm stores the~estimates of the~regrets for not playing action $i$ in all time steps in the~past in variables $R_i$.
On lines 3-7, it computes the~strategy for the~current time step.
If there is no positive regret for any action, a~uniform strategy is used (line~5).
Otherwise, the~strategy is chosen proportionally to the~positive part of the~regrets (line~7).
The uniform exploration with probability $\gamma$ is added to the~strategy as in the~case of Exp3. 
It also ensures that the~addition on line 10 is bounded.

\cite{cesa2006prediction} prove that regret matching eventually achieves zero regret in the~adversarial MAB problem, but they provide the~exact finite time bound only for the~perfect-information case, where the~agent learns rewards of all arms.

\subsection{An Alternative to SM-MCTS: Counterfactual Regret Minimization} \label{sec: CFR}
Counterfactual regret minimization (CFR) is an~iterative algorithm for computing approximate Nash equilibria in zero-sum extensive-form games with imperfect information (EFGs).
Since we use EFGs only at a~few places, we overload the~defined notation with corresponding concepts form EFGs.
In the EFG setting, the elements of $\mathcal H$ are typically called \emph{histories} rather than states. Unlike in SMGs, each non-terminal history $h\in \mathcal H$ only has a single acting player. The second difference is that in EFGs, histories are partitioned into \emph{information sets}. Instead of observing the current history $h$ directly, the acting player only sees the information set $I$ that $h$ belongs to.
The EFG framework is more general than the SMG one since each simultaneous decision in an~SMG can be modeled by two consecutive decisions in an~EFG (where the~player who acts second does not know the action chosen by the~first player).

Let $\sigma$ be the~strategy profile of the~players and denote $\pi^\sigma(h)$ the~probability of reaching $h\in \mathcal H$ from the~root of the~game under the~strategy profile $\sigma$, $\pi^\sigma(h,z)$ the~probability of reaching history $z$ given the~game has already reached history $h$. We use the~lower index at $\pi$ to denote the~players who contribute to the~probability, i.e., $\pi_p^\sigma(h)$ is player $p$'s contribution\footnote{It is the~product of the~probabilities of the~actions executed by player $p$ in history $h$.} to the~probability of reaching $h$ and $\pi_{-p}^\sigma(h)$ is the~contribution of the~opponent of $p$ and chance if it is present, i.e., $\pi^\sigma(h)=\pi_p^\sigma(h)\cdot\pi_{-p}^\sigma(h)$. We further denote by $ha$ the~history reached after playing action $a$ in history $h$.
The \emph{counterfactual value} of player $p$ playing action $a$ in an~information set $I$ under a~strategy $\sigma$ is the~expected reward obtained when player $p$ first chooses the~actions to reach $I$, plays action $a$, and then plays based on $\sigma_p$, while the~opponent and chance play according to $\sigma_{-p}$ all the~time:
\[ v_p^\sigma(I,a) = \sum_{(h,z)\in Z_I}\pi_{-p}^\sigma(h)\pi^\sigma(ha,z)u_p(z),\]
where $Z_I = \{(h,z)|z\in\mathcal{Z}, h\in I, h \text{ is prefix of } z\}$ are the~terminal histories that visit information set $I$ by a~prefix $h$.

Counterfactual regret in an~information set is the~external regret (see Definition~\ref{def:extRegret}) with respect to the~counterfactual values.
Counterfactual regret minimization algorithms minimize counterfactual regret in each information set, which provably leads to convergence of average strategies to a~Nash equilibrium of the~EFG \citep{zinkevich2007regret}.
The variant of counterfactual regret minimization most relevant for this paper is \emph{Monte Carlo Counterfactual Regret Minimization} (MCCFR) and more specifically outcome sampling. MCCFR minimizes the~counterfactual regrets by minimizing their unbiased estimates obtained by sampling. In the~case of outcome sampling, these estimates are computed based on sampling a~single terminal history, as in MCTS. Let $q(z)$ be the~probability of sampling a~history $z$. The~\emph{sampled counterfactual value} is:
\begin{equation*}
\tilde{v}^{\sigma}_p(I,a) = \left\{
\begin{array}{ll}
\frac{1}{q(z)} \pi^{\sigma}_{-p}(h) \pi^{\sigma}(ha,z) u_p(z) & \mbox{if } (h,z) \in Z_I\\
0 & \mbox{otherwise.}
\end{array} \right.
\end{equation*}
The sampling probability typically decreases exponentially with the~depth of the~tree. Therefore, $\frac{1}{q(z)}$ will often be large, which causes high variance in sampled counterfactual value. This has been shown to cause slower convergence both theoretically~\citep{gibson2012generalized} and practically~\citep{bosansky2016aij}.


\subsection{Relation to Multi-agent Reinforcement Learning}

Our work is also related to multi-agent reinforcement learning (MARL) in Markov games \citep{littman1994markov}. The~goal of reinforcement learning is to converge to the~optimal policy based on rewards obtained in simulations. Markov games are more general than SMGs in allowing immediate rewards and cycles in the~state space. However, any SMG can be viewed as a~Markov game. Therefore, the~negative results presented in this paper apply to Markov games as well.

To the~best of our knowledge, the~existing algorithms in MARL literature do not help with answering the~question of convergence of separate Hannan consistent strategies in individual decision points. They either explicitly approximate and solve the~matrix games for individual stages of the~Markov games \citep[e.g.,~][]{littman1994markov} or do not have convergence guarantees beyond repeated matrix games \citep{bowling2002multiagent}.

\subsection{Relation to Stochastic Shortest Path Problem}
Recently, some authors considered variants of Markov decision processes (MDP, see for example \citealp{puterman2014markov}), where the~rewards may change over time (stochastic shortest path problem, discussed for example in \citealp{neu2010online}) or even more generally, where the~rewards \emph{and} the~transition probabilities may change over time \citep{yu2009arbitrarily,abbasi2013online}.

SM-MCTS can be viewed as a~special case of this scenario, where the~transition probabilities change over time, but the~rewards remain the~same. Indeed, assuming the~role of one of the~players, we can view each node of the~game tree as a~state in MDP. From a~state $h$, we can visit its child nodes with a probability which depends on the~strategy of the~other player. This strategy is unknown to the~first player and will change over time. 

Our setting is more specific than the~general version of MDP -- the~state space contains no loops, as it is, in fact, a~tree. Moreover, the~rewards are only received at the~terminal states and correspond to the~value of these states. Consequently, an~algorithm which would perform well in this special case of MDPs with variable transition probabilities could also be successfully used for solving simultaneous-move games. However, to the~best of our knowledge, so far all such algorithms require additional assumptions, which do not hold in our case.

\section{Application of the MAB Problem to SM-MCTS(-A)} \label{sec: application to SM-MCTS(-A)}
To analyze SM-MCTS(-A) we need to know how it is affected by the~selection policies it uses (line \ref{alg:smmcts:select} in Algorithms \ref{alg:smmcts} and \ref{alg:smmcts-A}) and by the~properties of the~game it is applied to.
In this section, we first introduce some additional notation related to the~MAB problem and SM-MCTS(-A).
We then frame the~events at $h$ as a~separate MAB problem $(P^h)$ (resp. $(\bar P^h)$ for SM-MCTS-A) in such a~way that applying the~bandit algorithm from $h$ to $(P^h)$ yields exactly the~output observed at line $\ref{alg:smmcts:select}$.

Throughout the~paper we will use the~following notation for quantities related to MAB problems:
In any MAB problem $(P)$, the~reward assignment $x_i(t)$ is such that
\begin{equation}\label{eq:MAB}
x_i(t) = \textrm{the reward the~agent would receive in $(P)$ if they chose the~action } i \textrm{ at time } t .
\end{equation}
We define the~notions of \emph{cumulative payoff} $G(\cdot)$ and \emph{maximum cumulative payoff}  $G_{\max}(\cdot)$ and relate these quantities to the~external regret\footnote{For definition of the~external regret, see Definition \ref{def:extRegret}.}:
\begin{align}
G(t) & := \sum_{s=1}^{t}x_{i(s)}(s) \nonumber \\
G_{\max}(t) & := \max_{i\in\mathcal{A}_{1}}\sum_{s=1}^{t}x_{i}(s), \label{eq:MAB_problem_1} \\
R(t) & = G_{\max}(t)-G(t). \nonumber 
\end{align}
We also define the~corresponding average notions and relate them to the~average regret:
\begin{equation} \begin{split} \label{eq:MAB_problem_2}
g(t) 		& := G(t)/t, \\
g_{\max}(t) & := G_{\max}(t)/t , \\
r(t) 		& = g_{\max} (t)-g(t) .
\end{split} \end{equation}

Next, we introduce the~additional notation related to SM-MCTS(-A). By $i^h(t)$ and $j^h(t)$ we denote the~action chosen by player 1 (resp. 2) during the~$t$-th visit of $h$.
To track the~number of uses of each action, we set
\footnote{Note that we define $t^h_i$ as the~number of uses of $i$ up to the~$(t-1)$-th visit of $h$, increased by 1, even though the~more natural candidate would be simply the~number of uses of $i$ up to the~$t$-th visit. However, this version will simplify the~notation later and for the~purposes of computing the~empirical frequencies, the difference between the two definitions becomes negligible with increasing $t$.\label{footnote: t_i}}
\begin{align}
t^h_{i} & := 1 + \left| \left\{ 1\leq s \leq t-1 | \ i^h(s) = i \right\} \right| , \nonumber \\
t^h_{ij} & := 1 + \left| \left\{ 1\leq s \leq t-1 | \ i^h(s)=i \ \& \ j^h(s)=j \right\} \right| \label{eq:local_time}  
\end{align}
and define $t^h_{j}$ analogously to $t^h_{i}$. Note that when $(i,j)=(i^h(t),j^h(t))$, $t^h_{ij}$ is actually equal to the~number of times this joint action has been used up to (and including) the~$t$-th iteration. For $(i,j)\neq(i^h(t),j^h(t))$, this is equal to the~same number increased by 1.

We now define $(P^h)$.
When referring to the~quantities from \eqref{eq:MAB_problem_1} and \eqref{eq:MAB_problem_2} which correspond to $(P^h)$, we will add the~superscript $h$ (for example $G^h_{\max}$, $g^h$, $r^h$). To keep the~different levels of indices manageable, we will sometimes write e.g. $g^h_{ij}$ and $x^h_{ij}$ instead of $g^{h_{ij}}$ and $x^{h_{ij}}$.
To indicate whether an~average regret is related to player 1 or 2, we denote the~corresponding quantities as $r^h_1$ and $r^h_2$.

Since we want the~reward assignment corresponding to $(P^h)$ to coincide with what is happening at $h$ during SM-MCTS(-A), \eqref{eq:MAB} requires us to define $x^h_i(t)$ as
\begin{align}
x^h_i(t) := & \text{ the~reward $x$ from line \ref{alg:smmcts:reccall} in Figure~\ref{alg:smmcts} (resp. \ref{alg:smmcts-A}) that we would get} \nonumber\\
	& \text{ during the~$t$-th visit of $h$ if we switched the~first action at line \ref{alg:smmcts:select}} \nonumber\\
	& \text{ to $i$ (while keeping the~choice of player 2 as $j=j^h(t)$).} \label{equation: MAB description of SM-MCTS}
\end{align}
In particular, the reward for the selected action $i=i^h(t)$ is
\begin{align} \label{equation: equivalent MAB description of SM-MCTS}
x^h(t) =	x^h_{i^h(t)}(t) = \text{ the~reward obtained during the~$t$-th visit of $h$.}
\end{align}
For $i\neq i^h(t)$, the~unobserved reward $x_i^h(t)$ corresponds the value we would receive if we ran SM-MCTS(-A)$(h_{ij^h(t)})$ during the $t$-th visit of $h$.
Before the $t$-th visit of $h$, its child $h_{ij^h(t)}$ has been visited $(t^h_{ij^h(t)}-1)$-times. The hypothetical visit from \eqref{equation: MAB description of SM-MCTS} would therefore be the~$t^h_{ij^h(t)}\textnormal{-th}$ one, implying that
\begin{equation}\label{eq:reward_for_action}
x^h_i(t) =  x^h_{ij^h(t)}(t^h_{ij^h(t)}) .
\end{equation}

We note a property of $x^h_i(t)$ that obviously follows from either \eqref{equation: MAB description of SM-MCTS} or \eqref{eq:reward_for_action}, but might be unintuitive and is crucial for understanding the behavior of SM-MCTS(-A).
Consider an action $i\neq i^h(t)$ that has not been selected at time $t$. Then, assuming the opponent keeps playing $j^h(t)$, the random variable $x^h_i(t)$ will \emph{not change} until player 1 selects $i$ and $t^h_{ij^h(t)}$ increases. (Because nodes that do not get visited remain inactive and none of their variables change.)
This setting where rewards come from ``reward pools'' and stay around until they get ``used up'' is in a~direct contrast with the~non-adaptive MAB setting where even the non-selected rewards disappear.
However, we argue that this behavior is inherent to the~presented version of SM-MCTS(-A), and might cause its non-averaged variant to malfunction (as demonstrated in Section \ref{sec:Counterexample}).

The MAB problem $(\bar P^h)$ is defined analogously to $(P^h)$, except that we denote the rewards as $\bar x_i^h(t)$ instead of $x_i^h(t)$, and define $\bar x_i^h(t)$ as the number $\bar x$ (rather than $x$) from line~\ref{alg:smmcts-A:reccall} from Algorithm~\ref{alg:smmcts-A}.
This $\bar x$ is, by definition, equal to the average reward $g^h_{ij}(\cdot)$ from the corresponding child node.
It follows that $\bar x^h_i(t)$ coincides with $g^h_{ij^h(t)}( t^h_{ij^h(t)} )$ (more precisely, $g^h_{ij^h(t)}( t^h_{ij^h(t)} )$  is a realization of the random variable $\bar x^h_i(t)$ ).

\medskip

Lastly, we define the \emph{empirical} and \emph{average strategies} corresponding to a specific run of SM-MCTS(-A).
By $t^h(T)$ we denote the~number of visits of $h$ up to the~$T$-th iteration of SM-MCTS(-A).
By \emph{empirical frequencies} we mean the strategy profile $\hat \sigma(T) = (\hat \sigma_p^h(t^h(T)) )_{p=1,2,\ h\in\cH}$ defined as $\hat{\sigma}^h_{1}(t)(i) := t^h_{i} / t$ (resp. $t_j/t$ for player 2).
The average strategy $\bar \sigma (T)$ is defined analogously, with
\begin{equation}\label{eq:avg_str}
\bar \sigma_p^h(t)(i) := \frac{1}{t} \sum_{s=1}^t \sigma_p^h(s)(i)
\end{equation}
in place of $\hat \sigma_p^h(t)(i)$.
The~following lemma says the~two strategies can be used interchangeably. The~proof consists of an~application of the~strong law of large numbers and can be found in the~appendix.
\begin{restatable}{lemma}{empaavg} \label{lem:emp_and_avg}
In the limit, the~empirical frequencies and average strategies will almost surely be equal. That is, 
$\limsup_{t \rightarrow \infty} \max_{i \in \cA_1}\,|\hat{\sigma}_1(t)(i)-\bar{\sigma}_1(t)(i)| = 0$ holds with probability $1$.
\end{restatable}

\section{Insufficiency of Local Regret Minimization for Global Convergence} \label{sec:Counterexample}

One might hope that any HC selection policy in SM-MCTS would guarantee that the~average strategy converges to a~NE. Unfortunately, this is not the~case -- the~goal of this section is to present a~corresponding counterexample. The~behavior of the~counterexample is summarized by the~following theorem:
\begin{theorem} \label{thm: counterexample}
There exists a~simultaneous-move zero-sum game $G$ with perfect information and a~HC algorithm $A$, such that when $A$ is used as a~selection policy for SM-MCTS, then the~average strategy $\bar{\sigma}\left(t\right)$ almost surely converges outside of the~set of $\frac 1 5$-Nash equilibria.
\end{theorem}

How is such a~pathological behavior possible? Essentially, it is because the~sampling in SM-MCTS is closely related to the~observed payoffs. By synchronizing the~sampling algorithms in different nodes in a~particular way, we will introduce a~bias to the~payoff observations. This will make the~optimal strategy look worse than it actually is, leading the~players to adopt a~different strategy.

Note that the~algorithm from Theorem \ref{thm: counterexample} does have the~guaranteed exploration property defined in Section \ref{sec: convergence}, which rules out some trivial counterexamples where parts of the~game tree are never visited.

\subsection{Simplifying remarks}\label{sec: WLOG}
We present two observations regarding the~proof of Theorem \ref{thm: counterexample}.

Firstly, instead of a~HC algorithm, it is enough to construct an~$\epsilon$-HC algorithm $A_\epsilon$ with the~prescribed behavior for arbitrary $\epsilon>0$. The~desired $0$-HC algorithm can then be constructed in a~standard way -- that is, by using a~1-Hannan consistent algorithm $A_{1}$ for some period $t_{1}$, then a~$\frac 1 2$-HC algorithm $A_{1/2}$ for a~longer period $t_{2}$ and so on. By choosing a~sequence $\left(t_{n}\right)_n$, which increases quickly enough, we can guarantee that the~resulting combination of algorithms $\left(A_{1/n}\right)$ is 0-Hannan consistent.

Furthermore, we can assume without loss of generality that the~algorithm $A$ knows if it is playing as the~first or the~second player and that in each node of the~game, we can actually use a~different algorithm $A$. This is true because the~algorithm always accepts the~number of available actions as input. Therefore we could define the~algorithm differently based on this number, and modify our game in some trivial way which would not affect our example (such as duplicating rows or columns).

\subsection{The~counterexample}
The structure of the~proof of Theorem \ref{thm: counterexample} is now as follows. First, we introduce the~game $G$ (Figure \ref{fig: hra}) and a~sequence of joint actions in $G$ which leads to
\[ \underset{h\in\mathcal{H}}\sum r_1^h(t^h(T)) = \underset{h\in\mathcal{H}}\sum r_2^h(t^h(T)) = 0 \ \ \& \ 
\ u_1(\bar\sigma(T)) \leq v^G - \frac 1 4 . \]
This behavior will serve as a~basis for our counterexample.
However, the~``algorithms'' generating this sequence of actions will be oblivious to the~actions of the~opponent, which means that they will not be $\epsilon$-HC. In the~second step of our proof (Lemma \ref{lemma: modification of A}), we modify these algorithms in such a~way that the~resulting sequence of joint actions stays similar to the~original sequence, but the~new algorithms are $\epsilon$-HC. In combination with the simplifying remarks from Section~\ref{sec: WLOG}, this gives Theorem \ref{thm: counterexample}.

\subsubsection{Deterministic version of the counterexample} \label{sec: deterministic counterexample}

Let $G$ be the~game from Figure \ref{fig: hra}. First, we will note what the~Nash equilibrium strategy $\pi$ in $G$ looks like. Having done that, we will describe a~sequence of actions in $G$ which leads to a~different (non-NE) strategy $\bar\sigma$. We will then analyze the~properties of this sequence, showing that utility of $\bar\sigma$ is sub-optimal, even though  the~regrets $r^h_p$ for every $h\in \mathcal H = \{I,J\}$ will be equal to zero. The~counter-intuitive part of the~example is the~fact that this pathological sequence satisfies $r^I_1=0$.
\medskip

\begin{figure} 
\centering
\begin{subfigure}{0.6\textwidth}
	\tikzset{
	solid node/.style={circle,draw,inner sep=1.5,fill=black},
	hollow node/.style={circle,draw,inner sep=1.5}
	}
	\begin{tikzpicture}[scale=1.5,font=\footnotesize]
	\tikzstyle{level 1}=[level distance=10mm,sibling distance=20mm]
	\tikzstyle{level 2}=[level distance=12.5mm,sibling distance=25mm]
	\tikzstyle{level 3}=[level distance=12.5mm,sibling distance=15mm]
	\node(0)[solid node,label=above:{Player 1 (node $I$)}]{}
		child{node(1)[hollow node,label=below: {0}]{} edge from parent node[left,xshift=-3]{X}}
		child{node(2)[solid node,label=right:{Player 1 (node $J$)}]{}
			child{node(3)[solid node]{}
				child{node[hollow node,label=below:{$1$}]{} edge from parent node[left]{$L$}}
				child{node[hollow node,label=below:{$0$}]{} edge from parent node[right]{$R$}}
				edge from parent node[left,xshift=-3]{$U$}
			}
			child{node(4)[solid node]{}
				child{node[hollow node,label=below:{$0$}]{} edge from parent node[left]{$L$}}
				child{node[hollow node,label=below:{$1$}]{} edge from parent node[right]{$R$}}
				edge from parent node[right,xshift=3]{$D$}
			}
			edge from parent node[right,xshift=3]{$Y$}
		}
	;
	\draw[dashed,rounded corners=10]($(3) + (-.2,.25)$)rectangle($(4) +(.2,-.25)$);
	\node at ($(3)!.5!(4)$) {Player 2 (node $J$)};
	\end{tikzpicture}
\end{subfigure}
\begin{subfigure}{0.3\textwidth}
	\includegraphics[width=\textwidth]{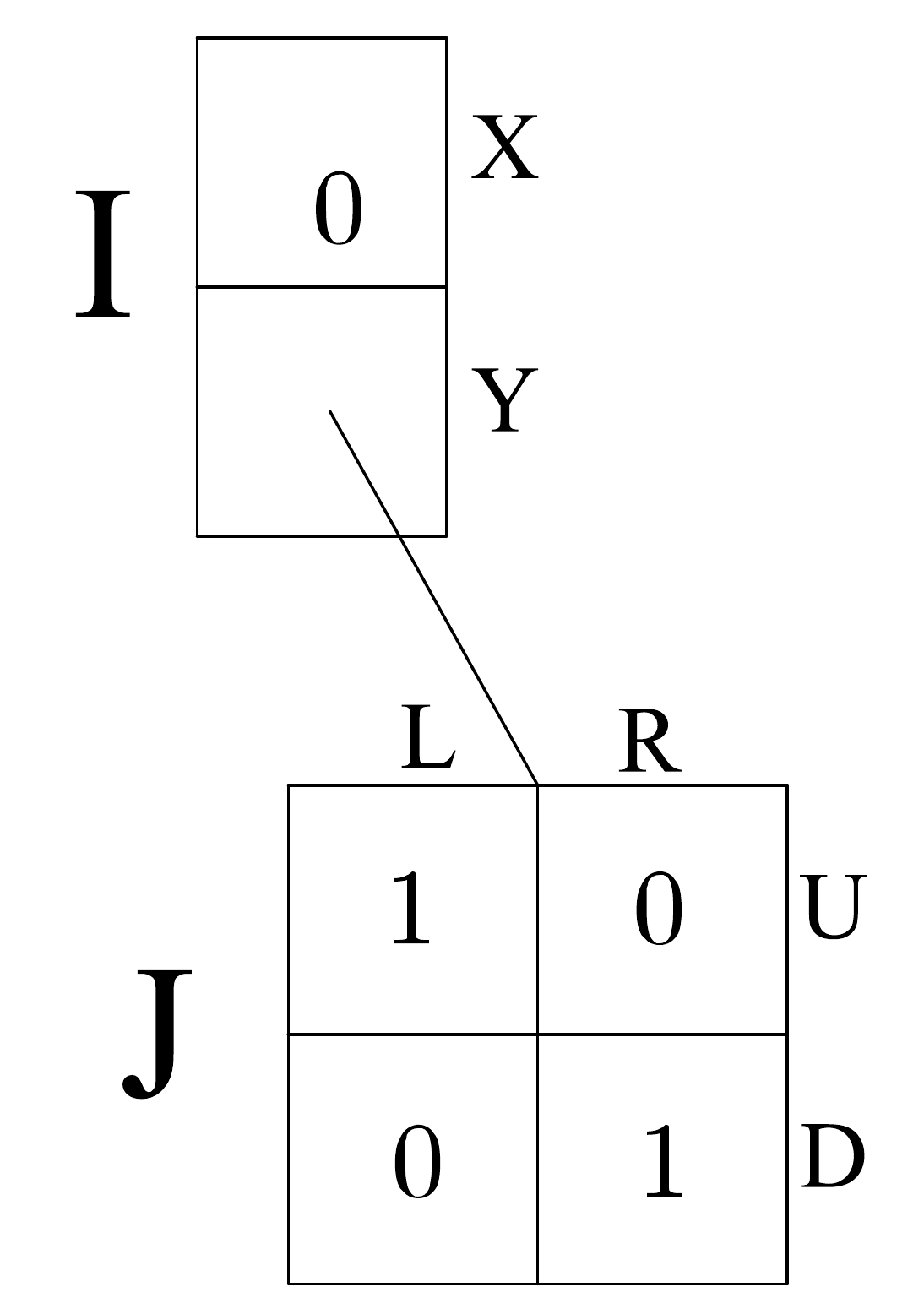}
\end{subfigure}
\caption[Game $G$]{Example of a~game in which it is possible to minimize regret at each of the~nodes while producing highly sub-optimal average strategy. The~extensive form representation of the~game is on the~left side, the~SMG representation is on the~right side. Since the~game is zero-sum, only the~utilities for Player 1 are displayed. \label{fig: hra}}
\end{figure}
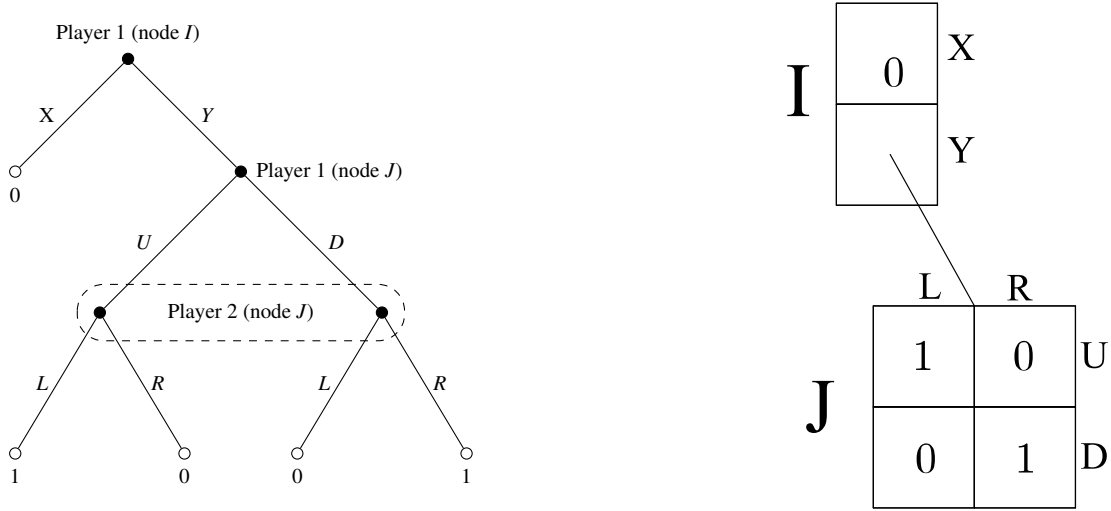

When it comes to the~optimal solution of $G$, we see that $J$ is the~well known game of matching pennies. The~equilibrium strategy in $J$ is $\pi^J_1=\pi^J_2=\left(\frac{1}{2},\frac{1}{2}\right)$ and the~value of this subgame is $v^J=\frac 1 2$. Consequently, player $1$ always wants to play $Y$ at $I$, meaning that the~NE strategy at $I$ is $\pi^I_1=(0,1)$ and the~value of the~whole game is $v^G=\frac 1 2$.

We define the~``pathological'' action sequence as follows.
Let $h\in{I,J}$ and $t\in\mathbb N$.
We set\footnote{Recall that $i^h(t)$ and $j^h(t)$ are the~actions of player 1 and 2 in $t$-th visit of a~node $h$.}
\begin{eqnarray*}
	\left( i^I(1), i^I(2), \ \dots \right) & = & \left( Y,X,X,Y,\ Y,X,X,Y,\ \dots \right), \\
	\left( \left(i^J(1),j^J(1)\right), \ \left(i^J(2),j^J(2)\right), \ \dots \right) & = & \left( (U,L), (U,R), (D,R), (D,L), \ \dots \right),	
\end{eqnarray*}
where the~dots mean that both sequences are $4$-periodic. The~resulting average strategy converges to $\bar \sigma$, where $\bar \sigma^I=\bar \sigma^J_1=\bar \sigma^J_2=\left(\frac{1}{2},\frac{1}{2}\right)$.

We now calculate the~observed rewards which correspond to the~behavior described above. We will use the~notation from Section \ref{sec: application to SM-MCTS(-A)} to describe the~events at $I$ and $J$ as MAB problems $(P^I)$ and $(P^J)$. The~action sequence at $J$ is defined in such a~way that we have
\begin{equation} \label{equation: x J}
\left( x^J\left(1\right),x^J\left(2\right),\ \dots \right) =\left( 1,0,1,0,\ \dots \right) .
\end{equation}
At the~root $I$ of $G$, we clearly have  $x^I_X(t)=0$ for all $t\in\mathbb N$. Since for $t=2,3,6,7,10,11,\dots$ we have $i^I(t)=X$, the~rewards received at $I$ during these iterations will be $x^I(t)=0$. Combining this with \eqref{equation: x J}, we see that the~rewards received at $I$ are
\begin{equation} \label{equation: x I} \begin{split}
\left( x^I(1),x^I(2),\ \dots \right) & = \left( x^J(1), 0, 0, x^J(2), x^J(3), 0, 0, x^J(4), \ 			\dots \right) \\
	& = \left( 1,0,0,0,1,0,0,0, \ \dots\right) .
\end{split} \end{equation}
In particular, the~average reward in $G$ converges to the~above mentioned $\frac 1 4$.

We claim that the~limit average strategy $\bar{\sigma}$ is far from being optimal. Indeed, as we observed earlier, the~average payoff at $I$ converges to $\frac 1 4$, which is strictly less than the~game value $v^G=\frac 1 2$. Clearly, player $1$ is playing sub-optimally by achieving only $u_1(\bar{\sigma}) = \frac{1}{4}$. By changing their action sequence at $I$ to $(Y,Y,Y,\dots)$, they increase their utility to $\frac 1 2$.

However, none of the players observes any \emph{local} regret. First, we check that no regret is observed at $J$. We see that the~average payoff from $J$ converges to $\frac 1 2$. In the~node $J$, each player takes each action exactly\footnote{To be precise, the ratio is exactly 50\% for iterations divisible by 4, and slightly different for the rest.} half the time. Since we are in the~matching pennies game, this means that neither of the~players can improve their payoff at $J$ by changing all their actions to any single action. Thus for both players $p=1,2$, we have $r_p^J(t)=0$.

Next, we prove that player 1 observes no regret at $I$. We start by giving an~intuitive explanation of why this is so. When computing regret, the~player compares the~rewards they received by playing as they did with the~rewards they would receive if they changed their actions to $Y,Y,\dots$. If, at time $t=1$, they asks themselves: ``What would I receive if I played $Y$?'', the~answer is ``$1$''. They \emph{do} play $Y$, and so the~strategy in $J$ changes and when they next asks themselves the~same question at time $t=2$, the~answer is ``$0$''. However this time, they \emph{do not} play $Y$ until $t=4$, and so the~answer remains ``$0$'' for $t=3$ and $t=4$. Only at $t=5$ does the~answer change to ``1'' and the~whole process repeats. This way, they is ``tricked'' into thinking that the~average reward coming from $J$ is $\frac{1+0+0+0}{4}=\frac 1 4$, rather than $\frac 1 2$.

To formalize the above idea, we use the~notation introduced in Section \ref{sec: application to SM-MCTS(-A)}. In order to get $r^I_1(t)\rightarrow 0$, it suffices to show that for $t=4k$, the~$\frac{1}{t}\sum_{s=1}^t x^I_Y(s) = \frac{1}{4}$. Since the~whole pattern is clearly 4-periodic, this reduces to showing that
\[ x^I_Y(1)+x^I_Y(2)+x^I_Y(3)+x^I_Y(4) = 1+0+0+0 = 1 . \]
Since for $t=1$, the~action $Y$ actually got chosen, we have $Y=i^I(1)$ and
\[ x^I_Y(1) = \ x^I_{i^I(1)}(1) \overset{\text{def.}}=
	x^I(1) \overset{\eqref{equation: x J}}= 1 . \]
The~same argument yields $x^I_Y(4)=0$. For $t=2$ and $t=3$, the~first iteration $s$ satisfying $s\geq t$ and $i^I(s)=Y$ is $s=4$. By definition in \eqref{equation: equivalent MAB description of SM-MCTS}, it follows that both $x^I_Y(2)$ and $x^I_Y(3)$ are equal to $x^I(4)$, which is zero by \eqref{equation: x I}.

\subsubsection{Hannan consistent version of the counterexample} \label{sec: HC counterexample}
The following lemma states that the~(deterministic) algorithms of player $1$ and $2$ described in Section \ref{sec: deterministic counterexample} can be modified into $\epsilon$-Hannan consistent algorithms in such a~way that when facing each other, their behavior remains similar to the~original algorithms.

\begin{restatable}{lemma}{modificationofA}
\label{lemma: modification of A}Let $G$ be the~game from Figure~\ref{fig: hra}. Then for each $\epsilon>0$
there exist $\epsilon$-HC algorithms $A^{I},\, A_{1}^{J},\, A_{2}^{J}$,
such that when these algorithms are used for SM-MCTS in $G$, the~resulting average strategy $\bar{\sigma}\left(t\right)$
converges to $\bar{\sigma}$, where $\bar{\sigma}^{I}=\bar{\sigma}_{1}^{J}=\bar{\sigma}_{2}^{I}=\left(\frac{1}{2},\frac{1}{2}\right)$.
\end{restatable}

The strategy $\bar{\sigma}$ satisfies $u_{1}\left(\bar{\sigma}\right)=\frac{1}{4}$, while the~equilibrium strategy $\pi$ (where $\pi^{I}=\left(0,1\right)$, $\pi_{1}^{J}=\pi_{2}^{J}=\left(\frac{1}{2},\frac{1}{2}\right)$, as shown in Section~\ref{sec: deterministic counterexample}), gives utility $u_{1}\left(\pi\right)=\frac{1}{2}$. Therefore the~existence of algorithms from Lemma \ref{lemma: modification of A} proves Theorem \ref{thm: counterexample}.

The key idea behind Lemma \ref{lemma: modification of A} is the~following: both players repeat the~pattern from Section~\ref{sec: deterministic counterexample}, but we let them perform random checks which detect any major deviations from this pattern. If both players do this, then by Section~\ref{sec: deterministic counterexample} they observe no regret at any of the~nodes. On the~other hand, if one of them deviates enough to cause a~non-negligible regret, they will be detected by the~other player, who then switches to a~``safe'' $\epsilon$-HC algorithm, leading again to a~low regret. The~definition of the~modified algorithms used in Lemma \ref{lemma: modification of A}, along with the~proof of their properties, can be found in Appendix~\ref{sec:appendix:counterexample}.

\subsection{Breaking the~counterexample: SM-MCTS-A}

We now discuss the~pathological behavior of SM-MCTS above and explain how these issues are avoided by SM-MCTS-A and some other algorithms.

Firstly, what makes the~counterexample work? The~SM-MCTS algorithm repeatedly observes parts of the~game tree to estimate the~value of available decisions. When sampling the~node $J$, half of the~rewards propagated upwards to $I$ have value $0$, and half have value $1$. SM-MCTS uses these payoffs directly and our rigged sampling scheme abuses this by introducing a~payoff observation bias -- that is, by making the~ zeros `stay' three times longer than the~ones\footnote{For an~explanation of the~payoff observation bias, see ``No-regret-at-$I$'' paragraph in Section~\ref{sec: deterministic counterexample}.}. This causes the~algorithm sitting at $I$ to estimate the~value $v^J$ of $J$ as $\frac{1}{4}$ (when in reality, it is $\frac{1}{2}$). Since $\frac{1}{4}$ is the~average payoff at $I$, the~algorithm at $I$ suffers no regret, and it can keep on working in this pathological manner.

The culprit here is the~\emph{payoff observation bias}\footnote{More on this in Section \ref{sec: eps-UPOproperty}.}, made possible by the~synchronization of the~selection algorithm used at $I$ with the~rewards coming from $J$.

So, why does SM-MCTS-A work in the~counterexample above, when SM-MCTS did not? SM-MCTS-A also uses the~biased payoffs, but instead of directly using the~most recent sample, it works with the~average of samples observed thus far. Since these averages converge to $v^J=\frac{1}{2}$, the~estimates of $v^J$ will also converge to this value, and no amount of `rigged weighting' can ruin this. The~only way to obtain the~average reward of $\frac{1}{2}$ is to (almost) always play $Y$ at $I$ and because the~algorithm used at $I$ is HC, this is exactly what will happen.

Thus when using SM-MCTS-A, we will successfully find a~NE of the~game $G$, even when the~observations made by the~selection algorithms are very much biased. This argument can be generalized for an~arbitrary game and a~HC algorithm -- we will do so in Theorem \ref{thm: SM-MCTS-A convergence}.

We conjecture that guaranteed convergence of SM-MCTS might still be possible, provided that the~algorithms used as selection policies were HC and that the~strategies prescribed by them changed slowly enough - such as is the~case with Exp3 (where the~strategies change slower and slower).\footnote{When the~the sampling strategies are constant, no synchronization like above is possible, which makes the~payoff observations unbiased. We believe that when the~strategies change slowly enough, the~situation might be similar.}

\subsection{Breaking the~counterexample: CFR}
In Section \ref{sec: CFR} we described the~CFR algorithm, which provably converges in our setting. We now explain how CFR deals with the~counterexample.

It is a~feature of MCTS that at each iteration, we only `care' about the~nodes we visited and we ignore the~rest. The~downside is that we never realize that every time we do \emph{not} visit $J$, the~strategy suggested for $J$ by our algorithm performs extremely poorly. On the~other hand, CFR visits every node in the~game tree in every iteration. This makes it immune to our counterexample -- indeed, if we are forced to care about what happens at $J$ during every iteration, we can no longer keep on getting three zero payoffs for every 1 while still being Hannan consistent.
The same holds for MCCFR, a Monte Carlo variant of CFR which no longer traverses the~whole tree each iteration, but instead only samples a~small portion of it.

\section{Convergence of SM-MCTS and SM-MCTS-A}\label{sec: convergence}

In this section, we present the~main positive results -- Theorems \ref{thm: SM-MCTS-A convergence} and \ref{thm: SM-MCTS convergence}. Apart from a~few cases, we only present the~key ideas of the~proofs here, while the~full proofs can be found in the~appendix. For an~overview of the~notation we use, see Table \ref{tab:notation}.

To ensure that the~SM-MCTS(-A) algorithm will eventually visit each node, we need the~selection policy to satisfy the~following property. 
\begin{definition}
We say that $A$ is an~\emph{algorithm with guaranteed exploration}
if, for any simul\-taneous-move zero-sum game\footnote{We define the~guaranteed exploration property this way (using extensive form games) to avoid further technicalities. Alternatively, we could define "two-player adversarial MAB problem" analogously to how adversarial MAB problem is defined in Definition \ref{def: MABproblem}, except that it would use a~setting similar to the~first part of Definition \ref{def:doubleMABproblem}. We would then say that $A$ is an~algorithm with guaranteed exploration if, for any reward assignment $\left(a_{ij}(t)\right)$, the~limit is almost surely infinity for all joint actions.} $G$ (as further specified in Section \ref{sec:SM games}) where $A$ is used by both players as a~selection policy for SM-MCTS(-A), and for any game node $h\in\mathcal{H}$, $\lim_{t\rightarrow\infty}t^h_{ij}=\infty$ holds almost surely for every joint action $(i,j)\in\mathcal{A}_1(h)\times\mathcal{A}_2(h)$.
\end{definition}
It is an~immediate consequence of this definition that when an~algorithm with guaranteed exploration is used in SM-MCTS(-A), every node of the~game tree will be visited infinitely many times. From now on, we will therefore assume that, at the~start of our analysis, the~full game tree is already built --- we do this because it will always happen after a~finite number of iterations and, in most cases, we are only interested in the~limit behavior of SM-MCTS(-A) (which is not affected by the~events in the~first finitely many steps). 

Note that most of the~HC algorithms, namely RM and Exp3, guarantee exploration without the~need for any modifications, but there exist some HC algorithms, which do not have this property. However, they can always be adjusted in the~following way: 
\begin{definition}\label{def:modified_A}
Let $A$ be a~bandit algorithm. For fixed \emph{exploration parameter} $\gamma\in\left(0,1\right)$
we define a~modified algorithm $A^{\gamma}$ as follows.
For time $t=1,2,...$ either:
\begin{enumerate}[a)]
\item explore with probability $\gamma$, or
\item run one iteration of $A$ with probability $1-\gamma$
\end{enumerate}
(where ``explore'' means we choose
the action randomly uniformly over available actions, without updating
any of the~variables belonging to $A$).

We define an algorithm $A^{\sqrt{\cdot}}$ analogously, except that at time $t$, the probability of exploration is $1/\sqrt{t}$ rather than $\gamma$.
\end{definition}
Fortunately, $\epsilon$-Hannan consistency is not substantially influenced
by the~additional exploration: 

\begin{restatable}{lemma}{AjeHC}
\label{lemma: A* je HC}For $\epsilon\geq 0$, let $A$ be an~$\epsilon$-Hannan consistent
algorithm.
\begin{enumerate}[(i)]
\item For any $\gamma>0$, $A^{\gamma}$ is an~$(\epsilon+\gamma)$-HC algorithm with guaranteed exploration.
\item $A^{\sqrt{\cdot}}$ is an~$\epsilon$-HC algorithm with guaranteed exploration.
\end{enumerate}
\end{restatable}
The proof of this lemma can be found in the appendix.

\subsection{Asymptotic convergence of SM-MCTS-A}\label{sub: SM-MCTS-A convergence}

A~HC selection will always minimize the~regret with respect to the~values used as an~input. But as we have seen in Section \ref{sec:Counterexample}, using the~observed values with no modification might introduce a~bias, so that we end up minimizing the~wrong quantity. One possible solution is to first modify the~input by taking the~averages, as in SM-MCTS-A. The high-level idea behind averaging is that it forces each pair of selection policies to optimize with respect to the~subgame values $v^h_{ij}$, which leads to the~following result:

\begin{theorem}
\label{thm: SM-MCTS-A convergence}
Let $\epsilon\geq 0$ and let $G$ be a~zero-sum game with perfect information and simultaneous moves with maximal
depth $D$ and let $A$ be an~$\epsilon$-Hannan consistent algorithm
with guaranteed exploration, which we use as a~selection policy for SM-MCTS-A.

Then almost surely, the~empirical frequencies $(\hat{\sigma}_{1}(t),\hat{\sigma}_{2}(t))$ will eventually get arbitrarily close to a~subgame-perfect $C\epsilon$-equilibrium, where $C=2D\left(D+1\right)$.
\end{theorem}

For $\epsilon=0$, Thorem \ref{thm: SM-MCTS-A convergence} gives the following:
\begin{corollary}
If the~algorithm from Thorem \ref{thm: SM-MCTS-A convergence} is Hannan-consistent, the~resulting strategy will eventually get arbitrarily close to a~Nash equilibrium.
\end{corollary}

To simplify the proofs, we assume that $\epsilon>0$ - the variant with $\epsilon=0$ can be obtained by sending $\epsilon$ to zero, or by minor modifications of the proofs. We will first state two preliminary results, then we use an~extension of the~later one to prove Theorem \ref{thm: SM-MCTS-A convergence} by backward induction. Firstly, we recall the~following well-known fact, which relates the~quality of the~best responses available to the~players with the~concept of an~equilibrium.
\begin{lemma}\label{Lem: u(br) and NE}
In a~zero-sum game with value $v$ the~following holds: 
\[
\left(u_{1}(br,\hat{\sigma}_{2})<v+\frac{\epsilon}{2}\,\ \& \ \, u_{1}(\hat{\sigma}_{1},br)>v-\frac{\epsilon}{2}\right)\Longrightarrow
\]
\[
\Big(u_{1}\left(br,\hat{\sigma}_{2}\right)-u_{1}\left(\hat{\sigma}_{1},\hat{\sigma}_{2}\right)<\epsilon\, \ \& \ \, u_{2}\left(\hat{\sigma}_{1},br\right)-u_{2}\left(\hat{\sigma}_{1},\hat{\sigma}_{2}\right)<\epsilon \Big) \overset{\textrm{def}}{\iff}
\]
\[
(\hat{\sigma}_{1},\hat{\sigma}_{2})\textrm{ is an~}\mbox{\ensuremath{\epsilon}}\textrm{-equilibrium.}
\]
\end{lemma}

In order to start the~backward induction, we first need to study what happens on the~lowest level of the~game tree, where the~nodes consist of matrix games. It is well-known that in a~zero sum matrix game, the~average strategies of two Hannan consistent players will eventually get arbitrarilly close to a~Nash equilibrium -- see \cite{waugh09d} and \cite{Blum07}. We prove a~similar result for the~approximate versions of the~notions.
\begin{restatable}{lemma}{HCandNE}
\label{L: HC-and-NE}Let $\epsilon\geq0$ be a~real number. If both
players in a~matrix game $M$ are $\epsilon$-Hannan consistent, then
the following inequalities hold almost
surely: 
\begin{equation}
\mbox{\hspace{0.3cm}}v^M-\epsilon\leq\underset{t\rightarrow\infty}{\liminf}\, g^M(t)\leq\underset{t\rightarrow\infty}{\limsup}\, g^M(t)\leq v^M+\epsilon,\label{eq: HC and NE 2}
\end{equation}
\begin{equation}
v^M-2\epsilon\leq\underset{t\rightarrow\infty}{\liminf}\, u^M_1\left(\hat{\sigma}_{1}(t),br\right)\mbox{\hspace{0.3cm} \ensuremath{\&}\hspace{0.3cm} }\underset{t\rightarrow\infty}{\limsup}\, u^M_1\left(br,\hat{\sigma}_{2}(t)\right)\leq v^M+2\epsilon.\label{eq: HC and NE 1}
\end{equation}
\end{restatable}

The inequalities (\ref{eq: HC and NE 2}) are a~consequence of the~definition
of $\epsilon$-HC and the~game value $v^M$. The~proof of inequality (\ref{eq: HC and NE 1})
then shows that if the~value caused by the~empirical frequencies was
outside of the~interval infinitely many times with positive probability,
it would be in contradiction with the~definition of $\epsilon$-HC.

Now, assume that the~players are repeatedly playing some matrix game $M$, but they receive slightly distorted information about their payoffs. This is exactly the~situation which will arise when a~stacked matrix game is being solved by SM-MCTS-A and the~``players'' are selection algorithms deployed at a~non-terminal node $h$ (where the~average payoffs $g^{h_{ij}}(t)$ coming from subgames rooted in $h_{ij}$ are never exactly equal to $v^h_{ij}$). Such a~setting can be formalized as follows:

\begin{definition}[Repeated matrix game with bounded distortion]\label{def:doubleMABproblem}
Consider the~following general problem $(\widetilde{M})$. Let $m,n\in\mathbb N$. For each $t\in\mathbb N$, the~following happens
\begin{enumerate}
\item Nature chooses a~matrix $\widetilde{M}(t)=\left(v^{\widetilde{M}}_{ij}(t)\right) \in [0,1]^{m\times n}$;
\item Players 1 and 2 choose actions $i\leq m$ and $j\leq n$ and observe the~number $v^{\widetilde{M}}_{ij}(t)$. This choice can be random and might depend on the~previously observed rewards and selected actions (that is $v^{\widetilde{M}}_{i(s)j(s)}(s)$ and $i(s)$, resp. $j(s)$ for player 2, for $s<t$).
\end{enumerate}

Let $\delta>0$. We call $({\widetilde{M}})$ \emph{a~repeated matrix game with $\delta$-bounded distortion} if the~following holds:
\begin{itemize}
\item For each $i,j,t$, $v^{\widetilde{M}}_{ij}(t)$ is a~random variable possibly depending on the~choice of actions $i(s),j(s)$ for $s<t$.
\item There exists a~matrix $M=(v^M_{ij})\in[0,1]^{m\times n}$, such that almost surely:
\[ \exists t_0\in \mathbb N \ \forall t \geq t_0 \ \forall i,j \ : \ \left| v^{\widetilde{M}}_{ij}(t)-v^M_{ij}\right|\leq \delta. \]
\end{itemize}
\end{definition}

From first player's point of view, each repeated matrix game with bounded distortion can be seen as an~adversarial MAB problem $(P^{\widetilde M})$ with reward assignment
\[ x^{\widetilde M}_i(t) := v^{\widetilde M}_{ij(t)}(t) . \]
When referring to the~quantities related to $(P^{\widetilde M})$, we add a~superscript $\widetilde{M}$ -- in particular, this concerns $R^{\widetilde{M}}$, $g^{\widetilde{M}}$, $g^{\widetilde{M}}_{\max}$ (and so on) from \eqref{eq:MAB_problem_1} and \eqref{eq:MAB_problem_2} and the~empirical frequencies $\hat{\sigma}^{\widetilde M}(t)$ corresponding to the~actions selected in $(P^{\widetilde M})$. Recall that $v^M$ and $u_p^{M}$ for $p=1,2$ stand for the~value and utility function corresponding to the~matrix game $M$.

We formulate the~following analogy of Lemma \ref{L: HC-and-NE}, which shows that $\epsilon$-HC algorithms perform well even if they observe slightly perturbed rewards.
\begin{restatable}{proposition}{hryschybou}
\label{Prop: hry s chybou} Let $c,\epsilon>0$ and let $(\widetilde M)$ be a~repeated matrix game with $c\epsilon$-bounded distortion, played by two $\epsilon$-HC players. Then the~following inequalities hold almost surely: 
\begin{align}
v^M-(c+1)\epsilon \leq \underset{t\rightarrow\infty}{\liminf} \ g^{\widetilde M}(t)
 & \leq \underset{t\rightarrow\infty}{\limsup} \ g^{\widetilde M}(t) \leq v^M+(c+1)\epsilon , \label{eq: hra s chybou2} \\
v^M-2(c+1)\epsilon \leq \underset{t\rightarrow\infty}{\liminf} \ u^M_1\left(\hat{\sigma}_{1},br\right)
 & \leq \underset{t\rightarrow\infty}{\limsup} \ u^M_1\left(br,\hat{\sigma}_{2}\right) \leq v^M+2(c+1)\epsilon . \label{eq: hra s chybou1}
\end{align}
\end{restatable}
The proof is similar to the~proof of Lemma~\ref{L: HC-and-NE}. It
needs an~additional claim that if the~algorithm is $\epsilon$-HC
with respect to the~observed values with errors, it still has a~bounded
regret with respect to the~exact values.

Next, we present the~induction hypothesis around which the~proof of Theorem~\ref{thm: SM-MCTS-A convergence} revolves.
We consider the~setting from Theorem \ref{thm: SM-MCTS-A convergence} and let $\eta>0$ be a small positive constant.
For $d\in\left\{ 1,...,D\right\} $ we denote $C_{d}:=d+O(\eta)$.\footnote{Our proofs use many inequalities that hold up to a depth-dependent noise, which we denote as $M_d\eta$. However, this noise can be made arbitrarily small by choosing the right $\eta$ at the start of each proof. The exact value of $M_d$ is therefore unimportant, and we hide it using the big $O$ notation. Note that multiplying $O(\eta)$ by $\epsilon$ still yields $O(\eta)$, since we can always assume $\epsilon\leq 1$ (the utilities are bounded by $1$ by definition).}
If a~node $h$ in the~game tree is terminal, we denote $d_h=0$. For other nodes, we inductively define $d_h$, the~depth of the~sub-tree rooted at $h$, as the~maximum of $d_{h_{ij}}$ over its children $h_{ij}$, increased by one. Recall that for $(i,j)\in\mathcal{A}_1(h)\times\mathcal{A}_2(h)$, $v^h_{ij}=v^{h_{ij}}\in[0,1]$ is the~value of the~subgame rooted at the~child node $h_{ij}$ of $h$.

\paragraph{Induction hypothesis $\left(IH_{d}\right)$\label{par: (IH)}} is the~claim that for each node $h$ with $d_{h}=d$, there almost surely exists $t_{0}$ such that for each $t\geq t_{0}$ 
\begin{enumerate}[(i)]
\item the~payoff $g^{h}(t)$ will fall into the~interval $\left(v^{h}-C_{d}\epsilon,v^{h}+C_{d}\epsilon\right)$;
\item the~utilities $u^{M_h}_1\left(\hat{\sigma}_{1}(t),br\right)$ and $u^{M_h}_1\left(br,\hat{\sigma}_{2}(t)\right)$ with respect to the~matrix game ${M_h}:=\left(v_{ij}^{h}\right)_{i,j}$ will fall into the~interval $\left(v^{h}-2C_{d}\epsilon,v^{h}+2C_{d}\epsilon\right)$.
\end{enumerate}

As we mentioned above, any node $h$ with $d_h=1$ is actually a~matrix game. Therefore Lemma \ref{L: HC-and-NE} ensures that $\left(IH_{1}\right)$ holds. In order to get Theorem \ref{thm: SM-MCTS-A convergence}, we only need the~property $(ii)$ to hold for every $h\in\mathcal{H}$. The~condition $(i)$ is ``only'' required for the~induction itself to work.

We now show that in the~setting of Theorem \ref{thm: SM-MCTS-A convergence}, the~induction step works. This is the~part of the~proof where we see the~difference between SM-MCTS and SM-MCTS-A -- in SM-MCTS, the~rewards will be different over time, due to the~randomness of used strategies. But the~averaged rewards, which SM-MCTS-A uses, will eventually be approximately the~same as values of the~respective game nodes. This allows us to view the~situation in every node as a~repeated game with bounded distortion and apply Proposition \ref{Prop: hry s chybou}.

\begin{proposition}
\label{prop: ind.step}In the~setting from Theorem \ref{thm: SM-MCTS-A convergence}, the~implication $\left(IH_{d}\right)\implies\left(IH_{d+1}\right)$ holds for every $d\in\{1,\dots,D-1\}$.
\end{proposition}
\begin{proof}
Assume that $\left(IH_{d}\right)$ holds and let $h$ be a~node with $d_{h} = d+1$. We will describe the~situation in $h$ as a~repeated matrix game $(\widetilde{M}_h)$ with bounded distortion:
\begin{itemize}
\item The~corresponding `non-distorted' matrix is $M_h:=\left(v^h_{ij}\right)_{i,j}$.
\item The~actions available to the~players $1$ and $2$ are $i\in\mathcal{A}_1(h)$ and $j\in\mathcal{A}_2(h)$.
\item The~actions $i(t),j(t)$ chosen by the~players are the~actions $i^h(t)$, $j^h(t)$.\footnote{Recall that by definition in \eqref{equation: equivalent MAB description of SM-MCTS}, $i^h(t)$ and $j^h(t)$ are the~actions chosen by SM-MCTS-A during the~$t$-th visit of the~node $h$.}
\item The~``distorted'' rewards $v^{\widetilde{M}_h}_{ij}(t)$ are the~average rewards $g^{h}_{ij}$ coming from the~subgames rooted in the~children of $h$. The~exact correspondence is
\begin{equation*}
v^{\widetilde{M}_h}_{ij}(t) := g^{h}_{ij}(t^h_{ij}) .
\end{equation*}
In other words, $v^{\widetilde{M}_h}_{ij}(t)$ is the~average of the~first $t^h_{ij}$ rewards obtained by visiting $h_{ij}$ during SM-MCTS-A.
\end{itemize}
In particular, we have $x^{\widetilde{M}_h}_i(t) = v^{\widetilde{M}_h}_{ij(t)}(t) = g^{h}_{ij^h(t)}(t^h_{ij}) $.
Comparing this with the definition of $\bar x_i^h(t)$ from Section~\ref{sec: application to SM-MCTS(-A)}, we see that the~MAB problems $(\bar P^h)$ and $(P^{\widetilde{M}_h})$ coincide. As a result, Proposition~\ref{Prop: hry s chybou} can be applied to the~actions taken at $h$ (assuming we can bound the distortion).

By $(i)$ from $\left(IH_{d}\right)$, the~distortion of $(\widetilde{M}_h)$ is eventually bounded by $C_d\epsilon=d\epsilon+O(\eta)$. Consequentially, we can apply Proposition \ref{Prop: hry s chybou} --- the~inequality \eqref{eq: hra s chybou2} shows that the~average rewards $g^h$ will eventually belong to the~interval $\left(v^{h}-\left((C_{d}+1)\epsilon+\eta\right),v^{h}+(C_{d}+1)\epsilon+\eta\right)$. Since $\left(C_d+1\right)\epsilon+\eta=C_{d+1}\epsilon$, $(i)$ from $(IH_{d+1})$ holds. The~property $(ii)$ follows from the~inequality \eqref{eq: hra s chybou1} in Proposition \ref{Prop: hry s chybou}.
\end{proof}
We are now ready to give the~proof of Theorem \ref{thm: SM-MCTS-A convergence}. Essentially, it consists of summing up the~errors from the~first part of the~induction hypothesis and using Lemma \ref{Lem: u(br) and NE}.

\begin{proof}[Proof of Theorem \ref{thm: SM-MCTS-A convergence}]
First, we observe that by Lemma \ref{L: HC-and-NE}, $\left(IH_{1}\right)$ holds, and consequently by Proposition \ref{prop: ind.step}, $\left(IH_{d}\right)$ holds for every $d=1,...,D$.

Denote by $u^{h}(\sigma)$ (resp. $u_{ij}^{h}(\sigma)$) the~expected player 1 payoff corresponding to the~strategy $\sigma$ used in the~subgame rooted at node $h\in\mathcal{H}$ (resp. its child $h_{ij}$). By Lemma \ref{Lem: u(br) and NE} we know that in order to prove Theorem \ref{thm: SM-MCTS-A convergence}, it is enough to show that for every $h\in\mathcal{H}$, the~strategy $\hat{\sigma}\left(t\right)$ will eventually satisfy
\begin{equation}
u^{h}\left(br,\hat{\sigma}_{2}\left(t\right)\right)\leq v^{h}+\left(d_h+1\right)d_h\epsilon+O(\eta). \label{eq: T8}
\end{equation}
The~role of both players in the~whole text is symmetric, therefore \eqref{eq: T8} also implies the~second inequality required by Lemma \ref{Lem: u(br) and NE}. We will do this by backward induction.

Since any node $h$ with $d_{h}=1$ is a~matrix game, it satisfies $(P^h)=(P^{M_h})$ and $u^h=u^{M_h}_1$. In particular, $(ii)$ from $\left(IH_{1}\right)$ implies that the~inequality \eqref{eq: T8} holds for any such $h$. Let $1<d\leq D$, $h\in\mathcal{H}$ be such that $d_h=d$ and assume, as a~hypothesis for backward induction, that the~inequality \eqref{eq: T8} holds for each $h'$ with $d_{h'}<d$. We observe that 
\begin{eqnarray*}
u^{h}\left(br,\hat{\sigma}_{2}\left(t\right)\right) & = & \max_{i\in\mathcal A_1(h)}\sum_{j\in\mathcal A_2(h)}\hat{\sigma}^h_{2}\left(t\right)\left(j\right)u^h_{ij}\left(br,\hat{\sigma}_{2}\left(t\right)\right)\\
 & \leq & v^{h}+\left(\max_{i\in\mathcal A_1(h)}\sum_{j\in\mathcal A_2(h)}\hat{\sigma}^h_{2}\left(t\right)\left(j\right)v_{ij}^{h}-v^{h}\right)+\\
 &  & +\max_{i\in\mathcal A_1(h)}\sum_{j\in\mathcal A_2(h)}\hat{\sigma}^h_{2}\left(t\right)\left(j\right)\left(u_{ij}^{h}\left(br,\hat{\sigma}_{2}\left(t\right)\right)-v_{ij}^{h}\right).
\end{eqnarray*}
Since $v^h_{ij}=v^{M_h}_{ij}$, the~first term in the~brackets is equal to $u^{M_h}_1(br,\hat\sigma^h_2)-v^{_h}$. By $(ii)$ from $\left(IH_{d}\right)$, this is bounded by $2d\epsilon+\eta$. By the~backward induction hypothesis we have
\[ u_{ij}^{h}\left(br,\hat{\sigma}_{2}\left(t\right)\right)-v_{ij}^{h}\leq d\left(d-1\right)\epsilon+O(\eta) \]
Therefore we have
\[ u^{h}\left(br,\hat{\sigma}_{2}\left(t\right)\right)\leq v^{h}+2d\epsilon+d\left(d-1\right)\epsilon+O(\eta)=v^{h}+\left(d+1\right)d\epsilon+O(\eta). \]
For $d=D$, $h=\textrm{root}$, Lemma~\ref{Lem: u(br) and NE} implies that $\left(\hat{\sigma}_{1}\left(t\right),\hat{\sigma}_{2}\left(t\right)\right)$ forms a~$\left(2D\left(D+1\right)\epsilon+O(\eta)\right)$-equilibrium of the~whole game.
\end{proof}

\subsection{Asymptotic convergence of SM-MCTS} \label{sec: SM-MCTS convergence}

We would like to prove an~analogous result to Theorem \ref{thm: SM-MCTS-A convergence} for SM-MCTS. Unfortunately, such a~goal is unattainable in general -- in Section \ref{sec:Counterexample} we presented a~counterexample, showing that a~such a~theorem with no additional assumptions does not hold.
We show that if an~$\epsilon$-HC algorithm with guaranteed exploration has an~additional property of `having $\epsilon$-unbiased payoff observations' ($\epsilon$-UPO, see Definition \ref{def: UPO}), it can be used as a~selection policy for SM-MCTS, and it will always find an~approximate equilibrium (Theorem \ref{thm: SM-MCTS convergence}). While we were unable to prove that specific $\epsilon$-HC algorithms have the~$\epsilon$-UPO property, in Section \ref{sec:Experimental}, we provide empirical evidence supporting our hypothesis that the~`typical' algorithms, such as regret matching or Exp3, indeed do have $\epsilon$-unbiased payoff observations.
	
\subsubsection{Definition of the~UPO property} \label{section: UPO definition}
First, we introduce the~notation required for the~definition of the~$\epsilon$-UPO property.
Recall that for $h\in \cH$ and a~joint action $(i,j)\in \cA_1(h)\times\cA_2(h)$, $x^h_{ij}(m)=x^{h_{ij}}(m)$ denotes the~$m\textnormal{-th}$ reward\footnote{Note that $x^h_{ij}(t)$ are random variables and their distribution depends on the~game to which $h$ belongs and on the~selection policies applied at all of the~nodes in the~subgame rooted in $h_{ij}$.} received when visiting $h$'s child node $h_{ij}$.  We denote the~arithmetic average of the~sequence $(x^h_{ij}(1)$, \dots, $x^h_{ij}(n))$ as 
\begin{equation}\label{not: property}
\bar{x}^h_{ij}\left(n\right)=\frac{1}{n}\sum_{m=1}^{n}x^h_{ij}\left(m\right) .
\end{equation}
The numbers $\bar x^h_{ij}$ are the~variables we would prefer to work with. Indeed, their definition is quite intuitive and when $A$ is $\epsilon$-HC and $d_{h_{ij}}=1$, the averages $\bar{x}^h_{ij}(n)$ can be used to approximate $v^h_{ij}$:
\[ \limsup_{n}\left|\bar{x}^h_{ij}\left(n\right)-v^h_{ij}\right|\leq\epsilon .\]

Unfortunately, the~variables $\bar{x}^h_{ij}\left(n\right)$ do not suffice for our analysis, because as we have seen in Section \ref{sec:Counterexample}, even when the~differences $\left| \bar{x}^h_{ij}\left(n\right)-v^h_{ij} \right|$ get arbitrarily small, SM-MCTS might still perform very poorly. In \eqref{notation: tilde s}, we define differently weighted averages $\tilde x^h_{ij}(n)$ of $(x^h_{ij}(1)$, \dots, $x^h_{ij}(n))$ which naturally appear in this setting and are crucial in the~proof of the~upcoming Theorem \ref{thm: SM-MCTS convergence}. Before we jump into details, let us expand on an~idea that we already hinted at earlier:

Assume that a~node $h_{ij}$ has already been visited $(m-1)$-times or, equivalently, that $t^h_{ij}=m$. We can compute the~$m$-th reward $x^h_{ij}(m)$ in advance and when player 2 selects $j$ for the~first time at some $s\geq t$, we offer $x^h_{ij}(m)$ as a~possible reward $x^h_i(s)$. However, if player $1$ plays \emph{something other than $i$} at $s$, $x^h_{ij}(m)$ does not ``get used up'' -- instead it ``waits'' to be used in future. Therefore, each reward $x^h_{ij}(m)$ must be weighted proportionally to the~number $w^h_{ij}(m)$ of times it got offered as a~possible reward (including the~one time when it eventually got chosen), because this is the~number of times it will show up when regret $R^h$ is calculated.

We define the~weights $w_{ij}\left(n\right)$ and the~weighted averages $\tilde{x}_{ij}\left(n\right)$ as follows\footnote{\textnormal{By $w^h_{ij}\left(m\right) * x^h_{ij}\left(m\right)$ we mean the sum of $w^h_{ij}\left(m\right)$ independent copies of $x^h_{ij}\left(m\right)$.}}: 
\begin{align}
w^h_{ij}\left(m\right) := & \left|\left\{ t\in\mathbb{N}| \ t^h_{ij} = m \ \& \ j^h(t)=j \right\}\right| \label{notation: tilde s} \\
\tilde{x}^h_{ij}\left(n\right) := & \frac{1}{\sum_{m=1}^{n}w^h_{ij}\left(m\right)}\sum_{m=1}^{n}w^h_{ij}\left(m\right) * x^h_{ij}\left(m\right), \nonumber
\end{align}
We are now ready to give the~definition:
\begin{definition}[UPO] \label{def: UPO}
We say that an~algorithm $A$ guarantees $\epsilon$-unbiased payoff observations, if for every (simultaneous-move zero-sum perfect information) game $G$ in which $A$ is used as a~selection policy for SM-MCTS, for its every node $h$ and every joint action $(i,j)$, the~arithmetic averages $\bar{x}^h_{ij}$ and weighted averages $\tilde{x}^h_{ij}$ almost surely satisfy
\[ \underset{n\rightarrow\infty}{\limsup}\left|\tilde{x}^h_{ij}\left(n\right)-\bar{x}^h_{ij}\left(n\right)\right|\leq\epsilon. \]
We will abbreviate this by saying that ``$A$ is an~$\epsilon$-UPO algorithm''.
\end{definition}

Observe in particular that for an $\epsilon$-UPO algorithm and $c>0$, we have,
\[ \limsup_{n\rightarrow\infty}\,\left|\bar{x}^h_{ij}\left(n\right)-v^h_{ij}\right|\leq c\epsilon \text{ a.s.}
\implies
\limsup_{n\rightarrow\infty}\,\left|\tilde{x}^h_{ij}\left(n\right)-v^h_{ij}\right|\leq\left(c+1\right)\epsilon \ \text{ a.s..} \]
The most relevant `examples' related to the UPO property are
\begin{enumerate}[a)]
\item the pathological algorithms from Section~\ref{sec: deterministic counterexample}, where the UPO property does not hold (and SM-MCTS fails to find a reasonable strategy),
\item Theorem \ref{thm: SM-MCTS convergence} below, which states that when a~HC algorithm is UPO, it finds an equilibrium strategy,
\item if at each node $h$, players chose their actions as independent samples of some probability distributions $P^h_1$ and $P^h_2$ over $\mathcal A_1(h)$ and $\mathcal A_2(h)$, then these selection `algorithms' are UPO.
\end{enumerate}

We agree that the definition of the UPO property, as presented in Definition \ref{def: UPO}, is quite impractical. Unfortunately, we were unable to find a replacement property which would be easier to check while still allowing the ``HC \& ? $\implies$ SM-MCTS finds a~NE'' proof to go through. We at least provide some more examples and a further discussion in the appendix (Section \ref{sec:UPO examples}).

\subsubsection{The convergence of SM-MCTS}
The goal of this section is to prove the~following theorem:
\begin{theorem}\label{thm: SM-MCTS convergence}
For $\epsilon\geq 0$, let $A$ be an~$\epsilon$-HC algorithm with guaranteed exploration that is $\epsilon$-UPO. If $A$ is used as selection policy for SM-MCTS, then the~average strategy of $A$ will eventually get arbitrarily close to a~$C'\epsilon$-NE of the~whole game, where $C'=12\left(2^{D}-1\right)-8D$.
\end{theorem}

Note that if we prove the statement for $\epsilon>0$, we get the variant with $\epsilon=0$ `for free'.
The proof of Theorem \ref{thm: SM-MCTS convergence} is similar to the~proof of Theorem \ref{thm: SM-MCTS-A convergence}. The~only major difference is Proposition \ref{prop: when tilde s = bar s}, which serves as a~replacement for Proposition \ref{Prop: hry s chybou} (which cannot be applied to SM-MCTS). Essentially, Proposition \ref{prop: when tilde s = bar s} shows that under the~specified assumptions, having low ``local'' regret in some $h\in\mathcal{H}$ with respect to $x^h_i(t)$ is sufficient to bound the~regret with respect to the~rewards originating from the~matrix game $\left(v^h_{ij}\right)$.

\begin{restatable}{proposition}{whentildesbars} \label{prop: when tilde s = bar s}
For $\epsilon,c\geq0$ let an~algorithm $A$ be $\epsilon$-HC and $h\in\mathcal{H}$. If $A$ is $\epsilon$-UPO and $\underset{n\rightarrow\infty}{\limsup}\,\left|\bar{x}^h_{ij}\left(n\right)-v^h_{ij}\right|\leq c\epsilon$
holds a.s. for each $i,j$, then
\begin{equation} \label{eq: regret}
\limsup_{t\rightarrow\infty}\frac{1}{t}\left(\max_{i}\sum_{s=1}^{t}v^h_{ij(s)}-\sum_{s=1}^{t}v^h_{i(s)j(s)}\right)\leq2\left(c+1\right)\epsilon \ \textrm{ holds a.s..}
\end{equation}
In particular, the~regret $r^{M_h}(t)$ of $A$ with respect to the~matrix game $M_h=\left(v^h_{ij}\right)_{ij}$ satisfies
\[ \limsup_{t\rightarrow\infty} \ r^{M_h}(t) \leq 2\left(c+1\right)\epsilon . \]
\end{restatable}
The proof of this proposition consists of rewriting the~sums in inequality \eqref{eq: regret} and using the~fact that the~weighted averages $\tilde{x}_{ij}$ are
close to the~standard averages $\bar{x}_{ij}$.

Consider the~setting from Theorem \ref{thm: SM-MCTS convergence} and let $\eta>0$ be a~positive constant. For $d\in\left\{ 1,...,D\right\} $ we denote $C'_{d}:=3\cdot2^{d-1}-2+O(\eta$).

\paragraph{Induction hypothesis $\left(IH'_{d}\right)$} is the~claim that for each node $h$ with $d_{h}=d$, there almost surely exists $t_{0}$ such that for each $t\geq t_{0}$ 
\begin{enumerate}
\item the~payoff $g^{h}(t)$ fall into the~interval $\left(v^{h}-C'_{d}\epsilon,v^{h}+C'_{d}\epsilon\right)$;
\item the~utilities $u^{M_h}\left(\hat{\sigma}_{1}(t),br\right)$ and $u^{M_h}\left(br,\hat{\sigma}_{2}(t)\right)$ with respect to the~matrix game $M_h=\left(v_{ij}^{h}\right)_{i,j}$ fall into the~interval $\left(v^{h}-2C'_{d}\epsilon,v^{h}+2C'_{d}\epsilon\right)$.
\end{enumerate}

\begin{proposition} \label{prop: ind.step SM-MCTS} In the~setting from Theorem \ref{thm: SM-MCTS convergence}, $\left(IH'_{d}\right)$ holds for every $d\in\{1,\dots,D\}$.
\end{proposition}
\begin{proof}
We proceed by backward induction -- since the~algorithm $A$ is $\epsilon$-HC, we know that, by Lemma \ref{L: HC-and-NE}, $\left(IH_{1}^{'}\right)$ holds with $C'_{1}=1+O(\eta)$. Assume that $(IH'_d)$ holds for some $d\in\{1,\dots,D-1\}$. By the~``in particular'' part of Proposition \ref{prop: when tilde s = bar s}, the~choice of actions in $h$ is $2\left(C'_{d}+1\right)$-HC with respect to the~matrix game $(v^h_{ij})$. By Lemma \ref{L: HC-and-NE}, $(IH'_{d+1})$ holds with $C'_{d+1}:=2\left(C'_{d}+1\right)+O(\eta)$. It remains to show that $C'_{d+1}$ is equal to $3\cdot2^{d}-2+O(\eta)$:
\[ C'_{d+1}=2\left(C'_{d}+1\right)+O(\eta) = 2\left(3\cdot2^{d-1}-2+1\right)+O(\eta) = 3\cdot2^d-2+O(\eta).\]
\end{proof}
Analogously to Proposition \ref{prop: ind.step} from Section \ref{sub: SM-MCTS-A convergence}, this in turn implies the~main theorem of this section, the~proof of which is similar to the~proof of Theorem \ref{thm: SM-MCTS-A convergence}.
\begin{proof}[Proof of Theorem {\ref{thm: SM-MCTS convergence}}]
Using Proposition \ref{prop: ind.step SM-MCTS}, the~proof is identical to the~proof of Theorem \ref{thm: SM-MCTS-A convergence} -- it remains to determine the~value of $C'$. As in the~proof of Theorem \ref{thm: SM-MCTS-A convergence} we have $C'=2\cdot2\sum_{d=1}^{D}C'_d$, and we need to calculate this sum: 
\[
\sum_{d=1}^{D}C'_{d}=\sum_{d=1}^{D}\left(3\cdot2^{d-1}-2\right)=3\left(1+...+2^{D-1}\right)-2D=3\left(2^{D}-1\right)-2D.
\]
\end{proof}

\subsection{Dependence of the~eventual NE distance on the~game depth}

In Theorems \ref{thm: SM-MCTS-A convergence} and \ref{thm: SM-MCTS convergence}, the bound on the distance of the average strategy to a NE (as measured by the constant $C$) is quadratic, resp. exponential, in the game depth.
We investigate how far can the dependence by improved, constructing a counterexample where we lower-bound $C$ by some constant times the depth.
However, the more general question of lower-bounds remains open and is posed as Problem~\ref{prob: linearity in D}.

\begin{proposition}
There exists $k>0$, such that none of the~Theorems \ref{thm: SM-MCTS-A convergence} and \ref{thm: SM-MCTS convergence} hold if the~constant $C$ is replaced by $\tilde C=kD$. This holds even when the~exploration is removed from $\hat \sigma$.
\end{proposition}
The proposition above follows from Example \ref{example: upper bound}.

\begin{figure}[t]
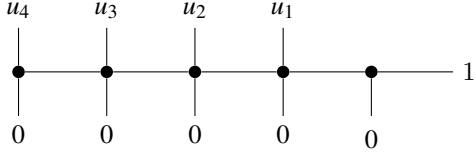

\centering
\usebox{\LinBound}
\caption{A game tree of a~single-player game where the~quality of a~strategy has linear dependence on the~exploration parameter and on the~game depth $D$. The~numbers $u_d$ satisfy $0<u_1<u_2<\dots<u_D<1$. The~player begins in the~leftmost node and the~numbers correspond to utilities in terminal nodes.}\label{fig: C is linear in D}
\end{figure}

\begin{example}[Distance from NE)]\label{example: upper bound}
Let $G$ be the~single player game\footnote{The other player always has only a~single \emph{no-op} action.} from Figure \ref{fig: C is linear in D}, $\eta>0$ some small number, and $D$ the~depth of the~game tree. Let Exp3 with exploration parameter $\gamma=k\epsilon$ be our $\epsilon$-HC algorithm (for a~suitable choice of $k$). We recall that this algorithm will eventually identify the~optimal action and play it with frequency $1-\gamma$, and it will choose randomly otherwise. Denote the~available actions at each node as (up, right, down), resp. (right, down) at the~rightmost inner node. We define each of the~rewards $u_d$, $d=1,...,D-1$ in such a~way that Exp3 will always prefer to go up, rather than right. By induction over $d$, we can see that the~choice $u_1=1-\gamma/2+\eta$, $u_{d+1}=(1-\gamma/3)u_d$ is sufficient and for $\eta$ small enough, we have
\[ u_{D-1}=(1-\frac \gamma 2+\eta)(1-\frac \gamma 3)\dots (1-\frac \gamma 3)\leq \left(1-\frac \gamma 3\right)^{D-1} \doteq 1-\frac {D-1} 3 \gamma \]
(where by $\doteq$ we mean that for small $\gamma$, in which we are interested, the~difference between the~two terms is negligible).
Consequently in each of the~nodes, Exp3 will converge to the~strategy $\left(1-\frac 2 3 \gamma, \frac 1 3 \gamma, \frac 1 3 \gamma\right)$  (resp.  $\left(1-\frac \gamma 2,\frac \gamma 2\right)$), which yields the~payoff of approximately $(1-\gamma/3)u_{D-1}$. Clearly, the~expected utility of such a~strategy is approximately
\[ u={\left(1-\gamma/3\right)}^D \doteq 1-\frac D 3\gamma. \]
On the~other hand, the~optimal strategy of always going right leads to utility 1, and thus our strategy $\hat \sigma$ is $\frac{D}{3}\gamma$-equilibrium.

Note that in this particular example, it makes no difference whether SM-MCTS or SM-MCTS-A is used. We also observe that when the~exploration is removed, the~new strategy is to go up at the~first node with probability 1, which again leads to regret of approximately $\frac D 3 \gamma$.
\end{example}
 By increasing the~branching factor of the~game in the~previous example from 3 to $b$ (adding more copies of the~``$0$'' nodes) and modifying the~values of $u_d$ accordingly, we could make the~above example converge to $2 \frac {b-2} b D\gamma$-equilibrium (resp. $\frac {b-2} b D\gamma$ once the~exploration is removed).
 
We were able to construct a~game of depth $D$ and $\epsilon$-HC algorithms, such that the~resulting strategy $\hat \sigma$ converged to $3D\epsilon$-equilibrium ($2D\epsilon$ after removing the~exploration). However, the~$\epsilon$-HC algorithms used in this example are non-standard and would require the~introduction of more technical notation. Therefore, since in our main theorem we use quadratic dependence $C=kD^2$, we instead choose to highlight the~following open question:
\begin{problem}\label{prob: linearity in D}
Does Theorem \ref{thm: SM-MCTS-A convergence} (and Theorem \ref{thm: SM-MCTS convergence}) hold with $C=kD$ for some $k>0$? If not, what is the optimal value of $C$?
\end{problem}
We hypothesize that the~answer is affirmative (and possibly the~values $k=3$, resp. $k=2$ after exploration removal, are optimal), but the~proof of such proposition would require techniques different from the~one used in the~proof of Theorem \ref{thm: SM-MCTS-A convergence}.
Moreover, it might be more interesting to consider Problem~\ref{prob: linearity in D} restricted to some class of ``reasonable'' no-regret algorithms, as opposed to arbitrary HC algorithms such as those presented in Section~\ref{sec:Counterexample}.

\section{Exploitability and exploration removal\label{sec: exploitability}}

In this section, we discuss which strategy should be considered the~output of our algorithms.
We also introduce the concept of exploitability -- a common approach to measuring strategy strength that focuses on the~worst-case performance (see for example \citealt{Johanson2011}):

\begin{definition}
\emph{Exploitability} of strategy $\sigma_{1}$ of player 1 is the~quantity
\[ \textrm{expl}_{1}\left(\sigma_{1}\right):=v-u_1\left(\sigma_{1},\textrm{br}\right),
\]
where $v$ is the~value of the~game and $\textrm{br}$ is a~second player's
best response strategy to $\sigma_{1}$.
Similarly, we denote $\textrm{expl}_{2}\left(\sigma_{2}\right) := u_1\left(\textrm{br}, \sigma_{2}\right) - v$ and $\textrm{expl}(\sigma) := \textrm{expl}_{1}(\sigma_1) + \textrm{expl}_{2}(\sigma_1)$.
\end{definition}

\noindent By definition of the~game value, exploitability is always non-negative.
Exploitability is closely related Nash equilibria, since we apparently have $\textrm{expl}_{1}\left(\sigma_{1}\right)=\textrm{expl}_{2}\left(\sigma_{2}\right)=0$ if and only if $\sigma=\left(\sigma_{1},\sigma_{2}\right)$ forms a~NE.

If SM-MCT(-A) is run for $t$ iterations, the obvious candidates for its output are the strategy $\sigma(t)$ from the last iteration, the~average strategy $\bar \sigma(t) = \frac{1}{t}\sum_{s=1}^t \sigma(s)$ (defined in Eq. \ref{eq:avg_str}), and the~empirical frequencies $\hat \sigma(t)$ (which are equivalent to $\bar \sigma(t)$ by Lemma~\ref{lem:emp_and_avg}).
Most theoretical results only give guarantees for $\bar \sigma(t)$, and indeed, using $\sigma(t)$ would be naive as it is often very exploitable.
However, there is a better option than using $\bar \sigma(t)$.

To improve the learning rate and ensure that the important parts of the game tree are found quickly, the selection functions used by SM-MCTS(-A) often have a~fixed exploration rate $\gamma>0$.
This can happen either naturally (like with Exp3) or because it has been added artificially (cf. Definition~\ref{def:modified_A}).
Each strategy $\sigma(t)$ is then of the form
\begin{equation}\label{eq:expl_rate}
\sigma(s) = (1-\gamma)\mu(s) + \gamma \cdot \textnormal{rnd}
\end{equation}
for some $\mu(t)$ (where rnd denotes the uniformly random strategy).
Rewriting $\bar \sigma(t)$ as
\begin{equation}\label{eq:rewriting_avg_str}
\bar \sigma(s)
= \frac{1}{t} \sum_{s=1}^t \sigma(s)
= \frac{1}{t} \sum_{s=1}^t \left[ (1-\gamma)\mu(s) + \gamma \cdot \textnormal{rnd} \right]
= (1-\gamma)\bar \mu(t) + \gamma \cdot \textnormal{rnd} ,
\end{equation}
we see that it contains the exact same amount of the~random noise caused by exploration.
While the~exploration speeds up learning, it also increases the~exploitability of $\bar \sigma_1$ in three ways: (a) it introduces noise to each $\bar \sigma_1^h$, (b) it does the same for the opponent, giving us inaccurate expectations about how an optimal opponent looks like and (c) it introduces a~noise to the~rewards SM-MCTS(-A) propagates upward in the tree.
As a~heuristic, \citealt{teytaud2011upper} suggest that the exploration should be removed.
In our case, we can use \eqref{eq:rewriting_avg_str} to obtain the ``denoised'' average strategy $\bar \mu(t)$:
\begin{equation}\label{eq:denoised_avg_str}
\bar{\mu}\left(t\right) :=
 \frac{1}{\left(1-\gamma\right)} \left( \bar{\sigma}\left(t\right) - \gamma \cdot \textup{\textrm{rnd}} \right) .
\end{equation}
Using $\bar \mu(t)$ instead of $\bar \sigma(t)$ cannot help with (b) and (c), but it does serve as an effective remedy for (a).
This intuition is captured by Proposition~\ref{prop:expl_removal}, which shows that a~better exploitability bound can be obtained if the~exploration is removed.
The experiments in Section \ref{sec:Experimental} verify that the~benefit of removing the~exploration in SM-MCTS is indeed large.

\begin{restatable}{proposition}{explRemoval}\label{prop:expl_removal}
Let $M$ be a matrix game, $A$ a~HC algorithm, $\gamma>0$, and $A^\gamma$ the ``$\gamma$-exploration modification'' of $A$ from Definition~\ref{def:modified_A}.
Let $\bar \sigma(t)$ be the average strategy obtained if both players use $A^\gamma$ in $M$, and let $\bar \mu(t)$ be as in \eqref{eq:denoised_avg_str}.
Then for each $p$ we have
\begin{enumerate}[(i)]
\item $\limsup_t \textnormal{expl}_p ( \bar \sigma_p(t) ) \leq 2\gamma$,\label{case:expl}
\item $\limsup_t \textnormal{expl}_p ( \bar \mu_p(t) ) \leq \gamma$.\label{case:no_expl}
\end{enumerate}
\end{restatable}

Lemma~\ref{L: HC-and-NE} gives \eqref{case:expl} from Proposition~\ref{prop:expl_removal} as a~special case, and, as shown in the last part of the~appendix, going through its proof in more detail yields \eqref{case:no_expl}.

\section{Experimental evaluation}\label{sec:Experimental}
In this section, we present the~experimental data related to our theoretical results. First, we empirically evaluate our hypothesis that Exp3 and regret matching algorithms ensure the~$\epsilon$-UPO property. Second, we test the~empirical convergence rates of SM-MCTS and SM-MCTS-A on synthetic games as well as smaller variants of games played by people. We investigate the~practical dependence of the~convergence error based on the~important parameters of the~games and evaluate the~effect of removing the~samples due to exploration from the~computed strategies.
We show that SM-MCTS generally converges as close to the equilibrium as SM-MCTS-A, but does it substantially faster. Therefore, since the commonly used Hannan consistent algorithms seem to satisfy the ~$\epsilon$-UPO property, SM-MCTS if the preferable algorithm for practical applications.

At the time of writing this paper, the best theoretically sound alternative to SM-MCTS known to us is a modification of MCCFR called Online Outcome Sampling (OOS, \citealt{lisy2015online}). For a comparison of SM-MCTS, this variant of MCCFR, and several exact algorithms, we refer the reader to \cite{bosansky2016aij}. While OOS converges the fastest out of the sampling algorithms in small games, it is often inferior to SM-MCTS in online game playing in larger games, because of the large variance caused by the importance sampling corrections.

\subsection{Experimental Domains}

\paragraph{\textbf{Goofspiel}}
Goofspiel is a~card game that appears in many works dedicated to simultaneous-move games (for example \citealt{Ross71Goofspiel,Rhoads12Computer,saffidine2012,lanctot2013goof,bosansky2013-ijcai}).
There are $3$ identical decks of cards with values $\{0,\dots, d-1\}$ (one for nature and one for each player).
Value of $d$ is a~parameter of the~game.
The deck for the~nature is shuffled at the~beginning of the~game. In each round, nature reveals the~top card from its deck. Each player selects any of their remaining cards and places it face-down on the~table so that the~opponent does not see the~card. Afterward, the~cards are turned face-up, and the~player with the~higher card wins the~card revealed by nature. The~card is discarded in case of a~draw. At the~end, the~player with the~higher sum of the~nature cards wins the~game or the~game is a~draw.
People play the~game with $13$ cards, but we use smaller numbers to be able to compute the~distance from the~equilibrium (that is, exploitability) in a~reasonable time.
We further simplify the~game by a~common assumption that both players know the~sequence of the~nature's cards in advance.

\paragraph{\textbf{Oshi-Zumo}}

Each player in Oshi-Zumo (for example, \citealt{buro2003}) starts with $N$ coins, and a~one-dimensional playing board with $2K+1$ locations (indexed $0, \ldots, 2K$) stretches between the~players.
At the~beginning, there is a~stone (or a~wrestler) located in the~center of the~board (that is, at position $K$).
During each move, both players simultaneously place their bid from the~amount of coins they have (but at least one if they still have some coins).
Afterward, the~bids are revealed, the~coins used for bids are removed from the~game, and the~highest bidder pushes the~wrestler one location towards the~opponent's side.
If the~bids are the~same, the~wrestler does not move.
The game proceeds until the~money runs out for both players, or the~wrestler is pushed out of the~board.
The player closer to the~wrestler's final position loses the~game.
If the~final position of the~wrestler is the~center, the~game is a~draw.
In our experiments, we use a~version with $K=2$ and $N=5$.

\paragraph{\textbf{Random Game}}

To achieve more general results, we also use randomly generated games. The~games are defined by the~number of actions $B$ available to each player in each decision point and a~depth $D$ ($D=0$ for leaves), which is the~same for all branches.
Each joint action is associated with a uniformly-random reward from $\{ -1,0,1 \}$, and the~corresponding accumulated utilities in leaves are integers between $-D$ and $D$.

\begin{figure}[t]
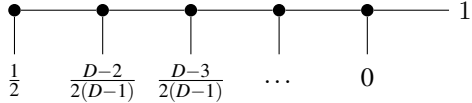

\centering
\usebox{\AntiL}
\caption{The Anti game used for evaluation of the~algorithms.}\label{fig:anti}
\end{figure}

\paragraph{\textbf{Anti}}

The last game we use in our evaluation is based on the~well-known single player game which demonstrates the~super-exponential convergence time of the~UCT algorithm \citep{coquelin2007bandit}. The~game is depicted in Figure~\ref{fig:anti}. In each stage, it deceives the~MCTS algorithm to end the~game while it is optimal to continue until the~end.

\subsection{$\epsilon$-UPO property}\label{sec: eps-UPOproperty}
To be able to apply Theorem \ref{thm: SM-MCTS convergence} (that is, the~convergence of SM-MCTS without averaging) to Exp3 and regret matching, the~selection algorithms have to assure the~$\epsilon$-UPO property for some $\epsilon$.
So far, we were unable to prove this hypothesis.
Instead, we support this claim by the~following numerical experiments.
Recall that having $\epsilon$-UPO property is defined as the~claim that for every game node $h\in\mathcal H$ and every joint action $\left(i,j\right)$  available at $h$, the~difference $\left|\bar{x}_{ij}\left(n\right)-\tilde{x}_{ij}\left(n\right)\right|$ between the~weighted and arithmetical averages decreases below $\epsilon$, as the~number $n$ of uses of $(i,j)$ at $h$ increases to infinity. We can think of this difference as a~bias between real and observed average payoffs.

\begin{figure}[t]
\includegraphics[width=\textwidth]{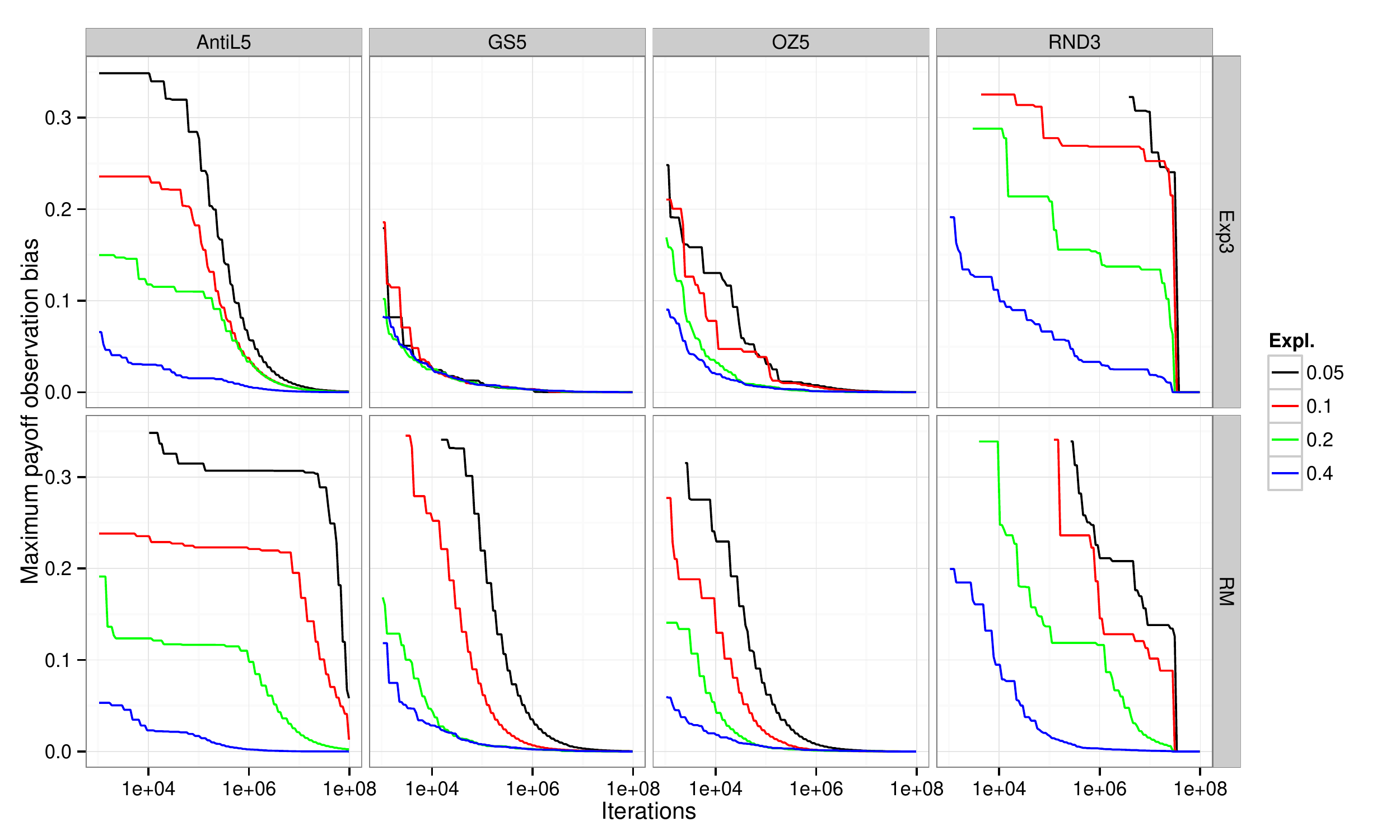}
\caption{The maximum of the~bias in payoff observations in SM-MCTS without averaging the~sample values (see the~second paragraph in Section \ref{sec: eps-UPOproperty}).}\label{fig:upoALL}
\end{figure}

We measured the~value of the~difference in the~sums in the~root node of the~four domains described above.
Besides the~random games, the~depth of the~game was set to 5.
For the~random games, the~depth and the~branching factor was $B=D=3$. Figure~\ref{fig:upoALL} presents one graph for each domain and each algorithm. The~x-axis is the~number of iterations and the~y-axis depicts the~maximum value of the~difference $\left|\bar{x}_{ij}\left(n\right)-\tilde{x}_{ij}\left(n\right)\right|$ from the~iteration on the~x-axis to the~end of the~run of the~algorithm ($\max_{n\in{(x,\dots,10^8)}, i,j} \left|\bar{x}_{ij}\left(n\right)-\tilde{x}_{ij}\left(n\right)\right|$). The~presented value is the~maximum from 50 runs of the~algorithms. For all games, the~difference eventually converges to zero. Generally, larger exploration ensures that the~difference goes to zero more quickly and the~bias in payoff observation is smaller.

\begin{figure} 
\centering
\begin{subfigure}{0.49\textwidth}
\includegraphics[width=\textwidth]{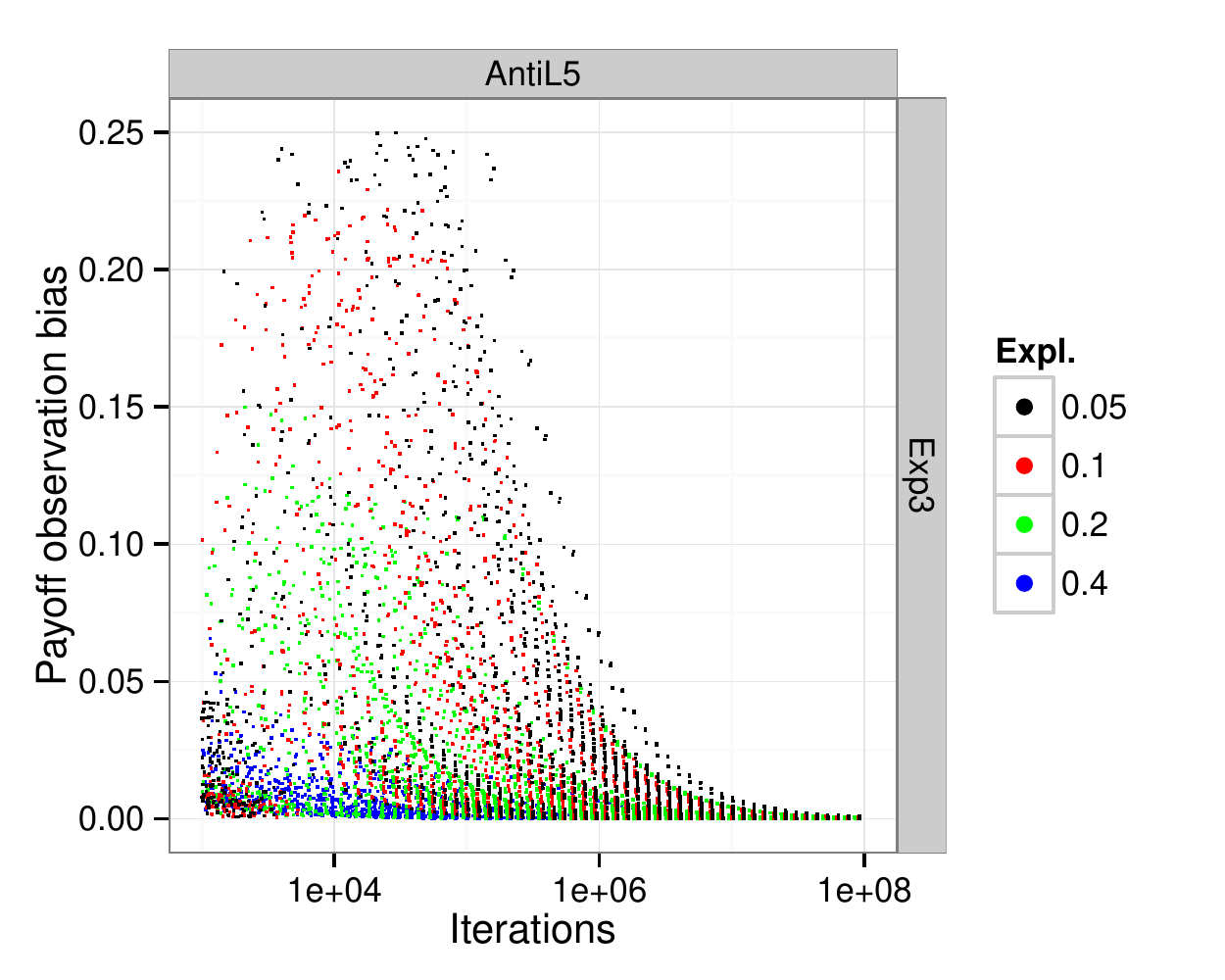}
\caption{Various exploration factors}\label{fig:upoAnti}
\end{subfigure}
\begin{subfigure}{0.49\textwidth}
\includegraphics[width=\textwidth]{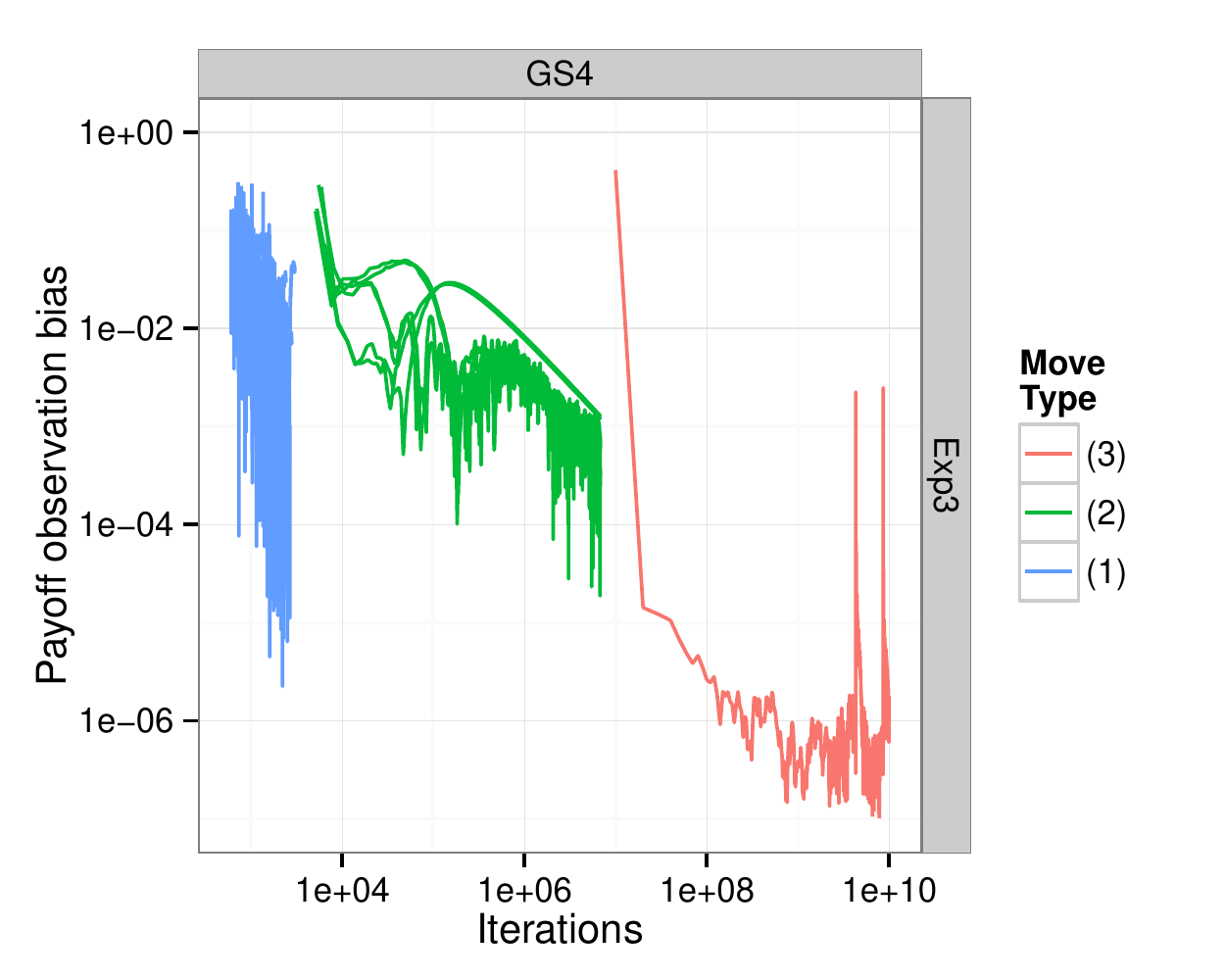}
\caption{Various joint actions}\label{fig:upoGS4}
\end{subfigure}
\caption{The dependence of the~current value of $\left|\bar{x}_{ij}\left(n\right)-\tilde{x}_{ij}\left(n\right)\right|$
on the~number of iterations that used the~given joint action in (a) Anti game and (b) Goofspiel with 4 cards per deck.}
\end{figure}

The main reason for the~bias is easy to explain in the~Anti game. Figure~\ref{fig:upoAnti} presents the~maximal values of the~bias in small time windows during the~convergence from all 50 runs. It is apparent that the~bias during the~convergence tends to jump very high (higher for smaller exploration) and then gradually decrease. This, however, happens only until a~certain point in time. The~reason for this behavior is that if the~algorithm learns an~action is good in a~specific state, it will use it very often and will do the~updates for the~action with much smaller weight in $\tilde{x}_{ij}\left(n\right)$ than the~updates for the~other action. However, when the~other action later proves to be substantially better, the~value of that action starts increasing rather quickly. At the~same time, its probability of being played starts increasing, and as a~result, the~weights used for the~received rewards start decreasing. This will cause a~strong dependence between the~rewards and the~weights, which causes the~bias. With smaller exploration, it takes more time to identify the~better alternative action; hence, when it happens, the~wrong action has already accumulated larger reward and the~discrepancy between the~right values and the~probability of playing the~actions is even stronger.

We also tested satisfaction of the~UPO property in the~root node of depth 4 Goofspiel, using  Exp3 algorithm and exploration $\epsilon=0.001$. The~results in Figure~\ref{fig:upoGS4} serve as evidence supporting our conjecture that Exp3 with exploration $0.001$ possesses the~$0.001$-UPO property (however this time, higher $n_{0}$ is required -- around $5\cdot10^{6}$).


We can divide the~joint actions $\left(i,j\right)$ at the~root into three groups: (1) the~actions which both players play (nearly) only when exploring, (2) the~actions which one of the~players chooses only because of the~exploration, and (3) the~actions which none of the~players uses only because of the~exploration. In Figure~\ref{fig:upoGS4}, (1) is on the~left, (2) in the~middle and (3) on the~right. The~third type of actions easily satisfied $\left|\bar{x}_{ij}\left(n\right)-\tilde{x}_{ij}\left(n\right)\right|\leq\epsilon$, while for the~second type, this inequality seems to eventually hold as well. The~shape of the~graphs suggests that the~difference between $\bar{x}_{ij}\left(n\right)$ and $\tilde{x}_{ij}\left(n\right)$ will eventually get below $\epsilon$ as well, however, the~$10^9$ iterations we used were not sufficient for this to happen. Luckily, even if the~inequality did not hold for these cases, it does not prevent the~convergence of SM-MCTS to an~approximate equilibrium. In the~proof of Proposition~\ref{prop: when tilde s = bar s}, the~term $\left|\bar{x}_{ij}\left(n\right)-\tilde{x}_{ij}\left(n\right)\right|$ is weighted by the~empirical frequencies $t_{j}/t$, resp. even $t_{ij}/t$. If an~action is only played when exploring, this number converges to $\epsilon/\left(\mbox{number of actions}\right)$, resp. $\left(\epsilon/\left(\mbox{number of actions}\right)\right)^{2}$, so even if we had $\left|\bar{x}_{ij}\left(n\right)-\tilde{x}_{ij}\left(n\right)\right|=1$, we could still bound the~required term by $\epsilon$, which is needed in the~proof of the~respective theorem.

\subsubsection{Removing exploration}

\begin{figure}[t]
\centering
\includegraphics[width=0.3\textwidth]{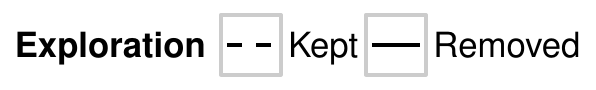}

\begin{subfigure}{0.244\textwidth}
\includegraphics[width=\textwidth]{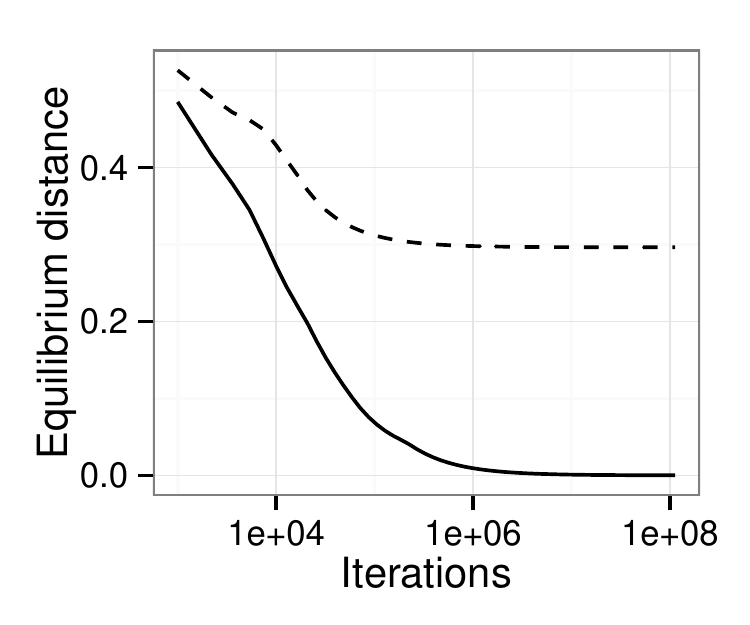}
\caption{Anti(5)}\label{fig:convRem_Anti}
\end{subfigure}
\begin{subfigure}{0.244\textwidth}
\includegraphics[width=\textwidth]{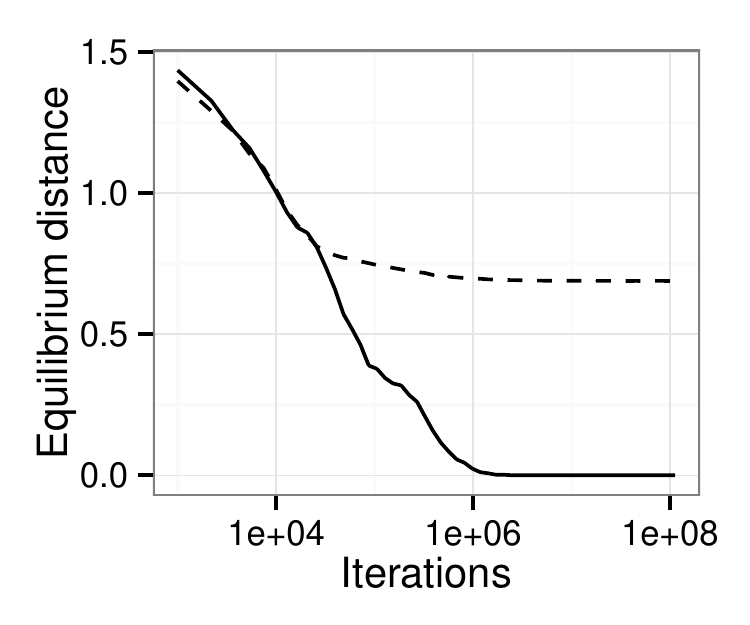}
\caption{Goofspiel(5)}\label{fig:convRem_GS}
\end{subfigure}
\begin{subfigure}{0.244\textwidth}
\includegraphics[width=\textwidth]{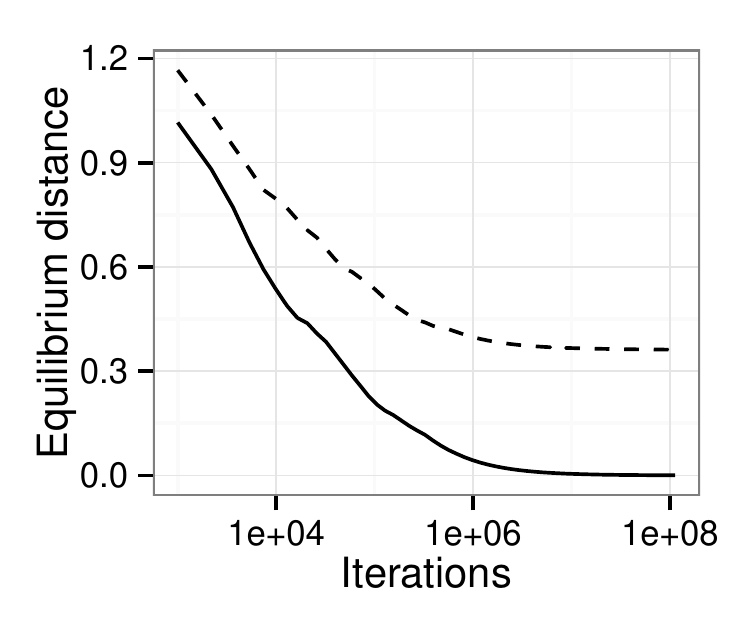}
\caption{Oshi-Zumo(5)}\label{fig:convRem_OZ}
\end{subfigure}
\begin{subfigure}{0.244\textwidth}
\includegraphics[width=\textwidth]{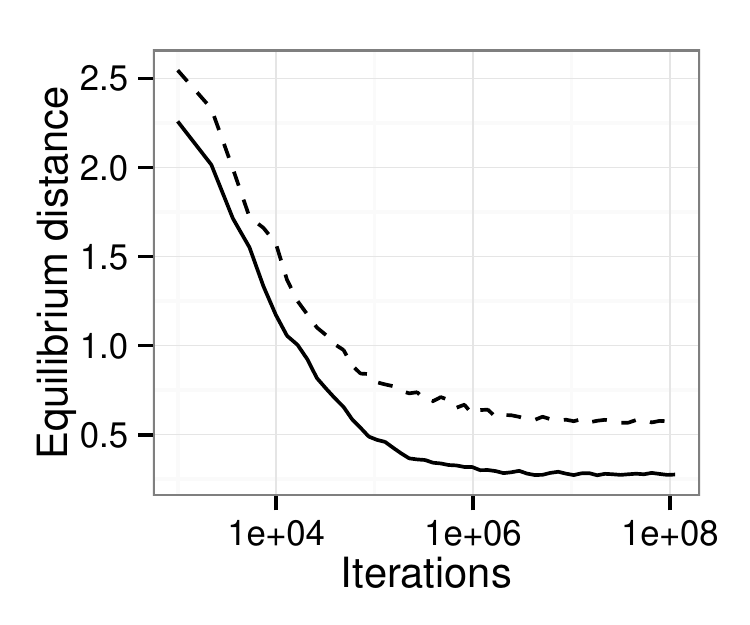}
\caption{Random(3)}\label{fig:convRem_RND}
\end{subfigure}
\caption{The effect of removing exploration in SM-MCTS with Exp3 selection and $\epsilon=0.2$.}\label{fig:explRem}
\end{figure}

In Section~\ref{sec: exploitability}, we show that the~computed strategy should generally improve when we disregard the~samples caused by exploration.
Figure~\ref{fig:explRem} shows that the practical benefit is large, and suggests that the~exploration should always be removed, and we do so in all the experiments below.

\subsection{Empirical convergence rate of SM-MCTS(-A) algorithms}\label{sec:emp_conv_rate}

\begin{figure}[tp]
\centering
\includegraphics[width=0.7\textwidth]{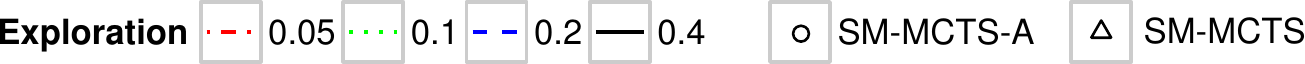}

\begin{subfigure}{0.45\textwidth}
\includegraphics[width=\textwidth]{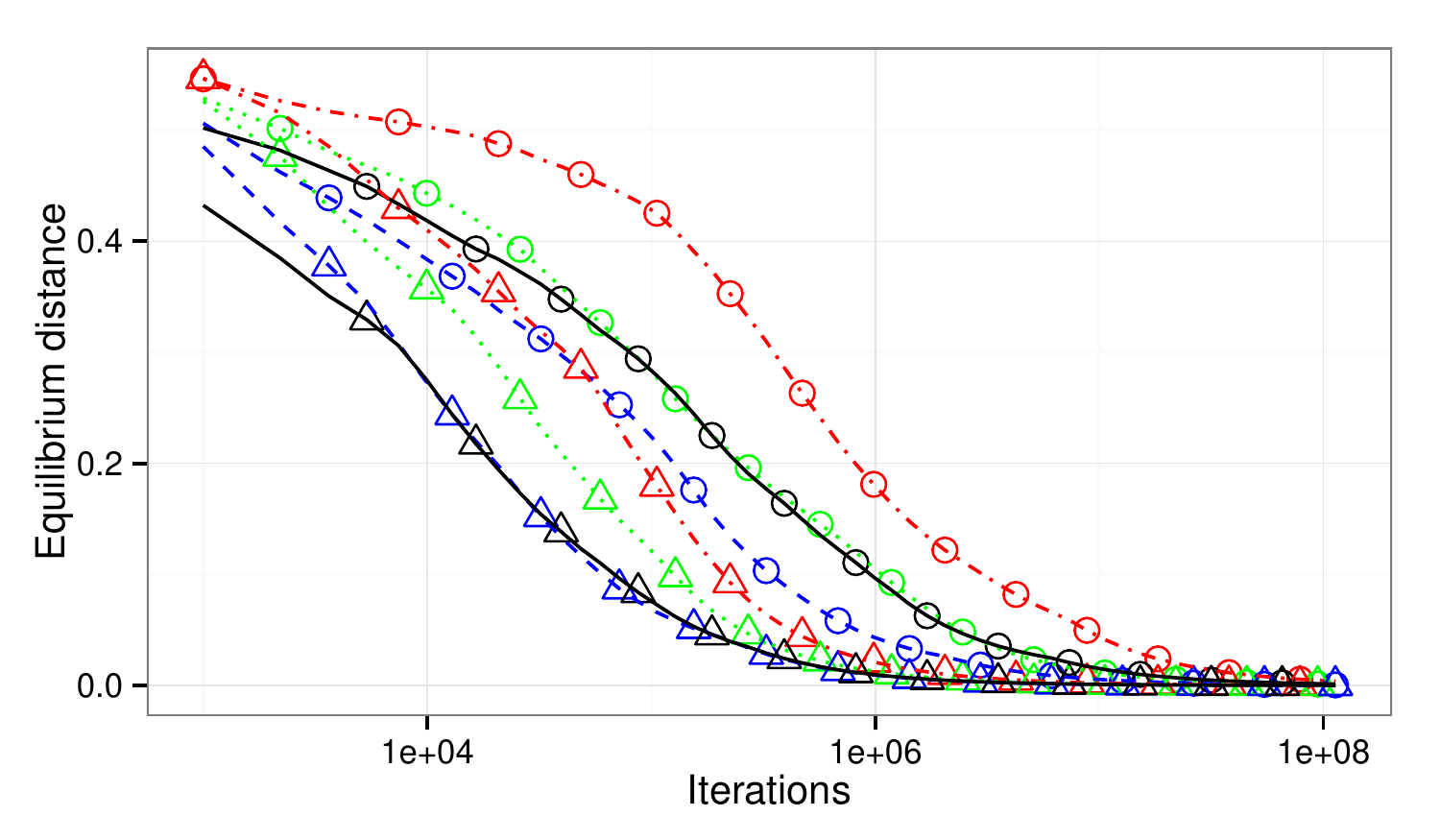}
\caption{Anti(5), Exp3}\label{fig:convAnti_Exp3}
\end{subfigure}
\begin{subfigure}{0.45\textwidth}
\includegraphics[width=\textwidth]{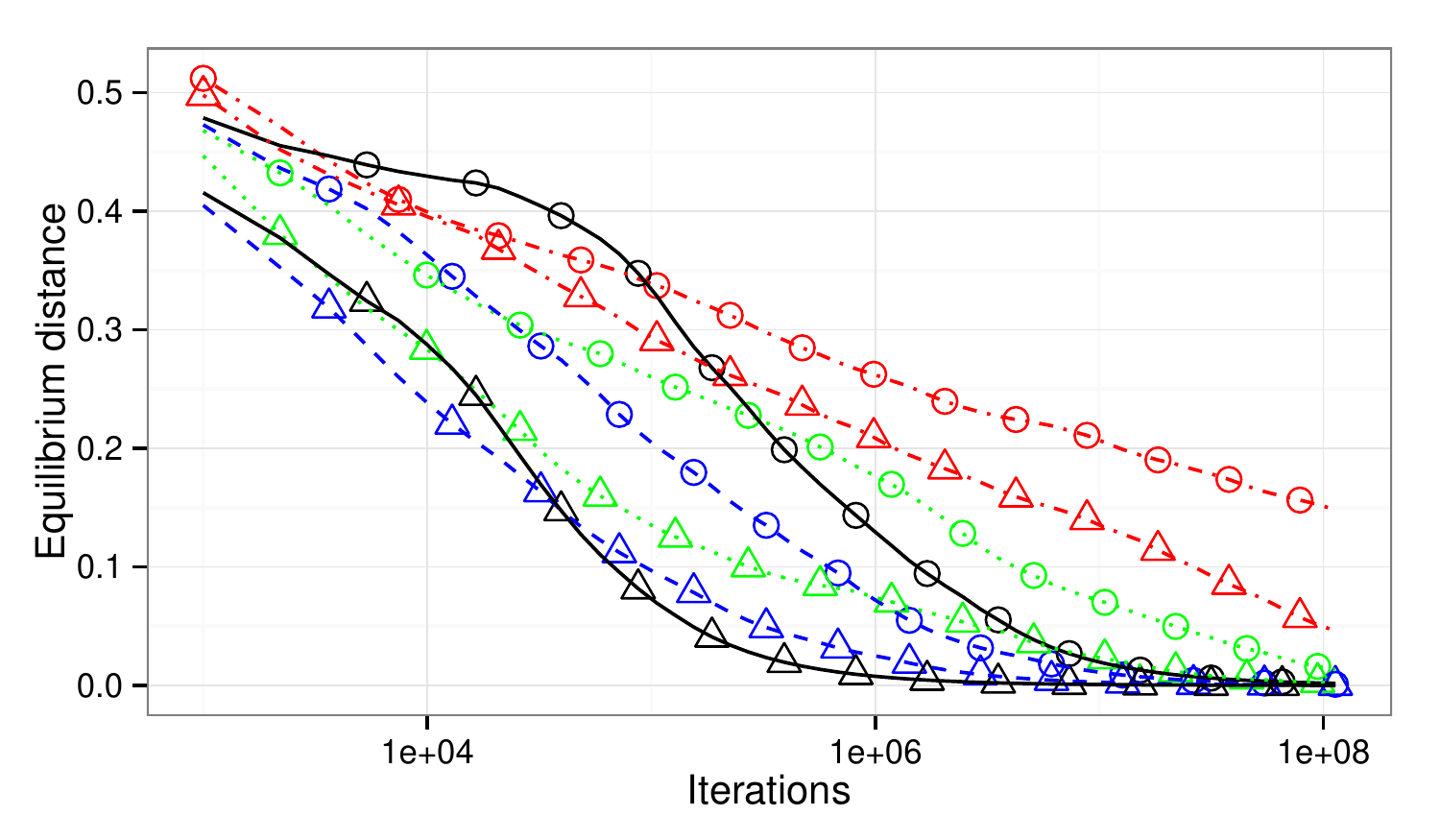}
\caption{Anti(5), RM}\label{fig:convAnti_RM}
\end{subfigure}

\begin{subfigure}{0.45\textwidth}
\includegraphics[width=\textwidth]{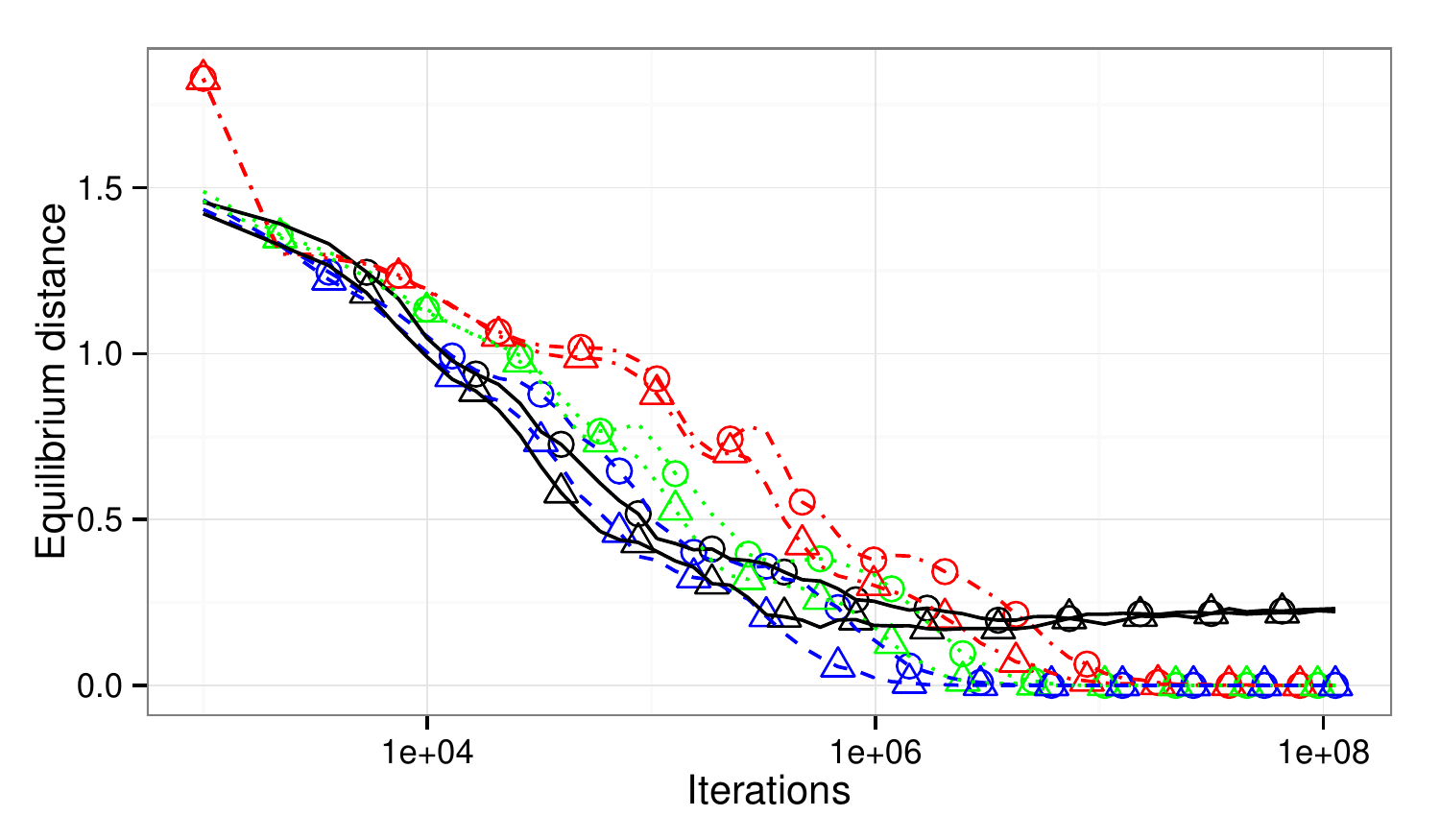}
\caption{Goofspiel(5), Exp3}\label{fig:convGS_Exp3}
\end{subfigure}
\begin{subfigure}{0.45\textwidth}
\includegraphics[width=\textwidth]{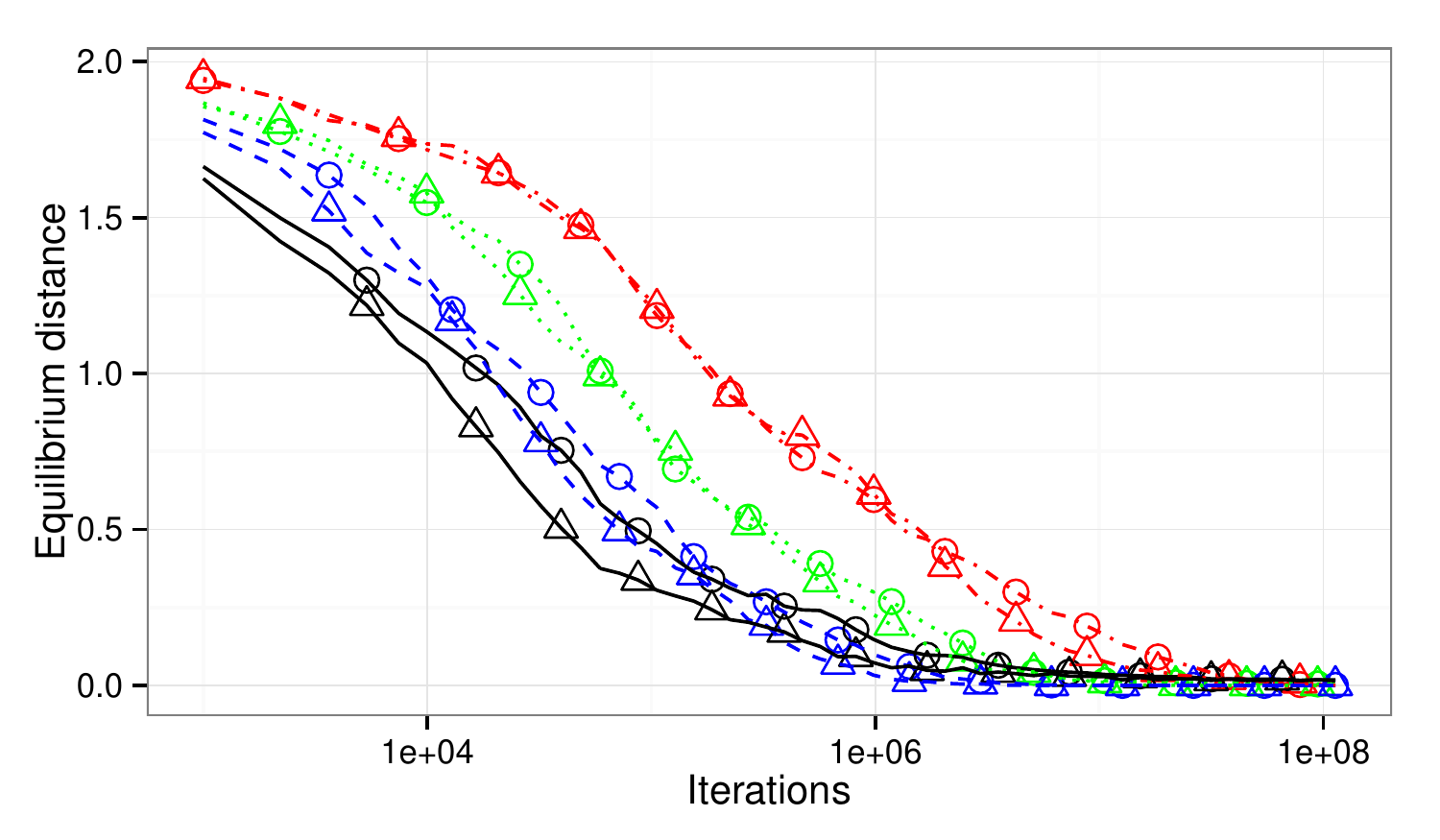}
\caption{Goofspiel(5), RM}\label{fig:convGS_RM}
\end{subfigure}

\begin{subfigure}{0.45\textwidth}
\includegraphics[width=\textwidth]{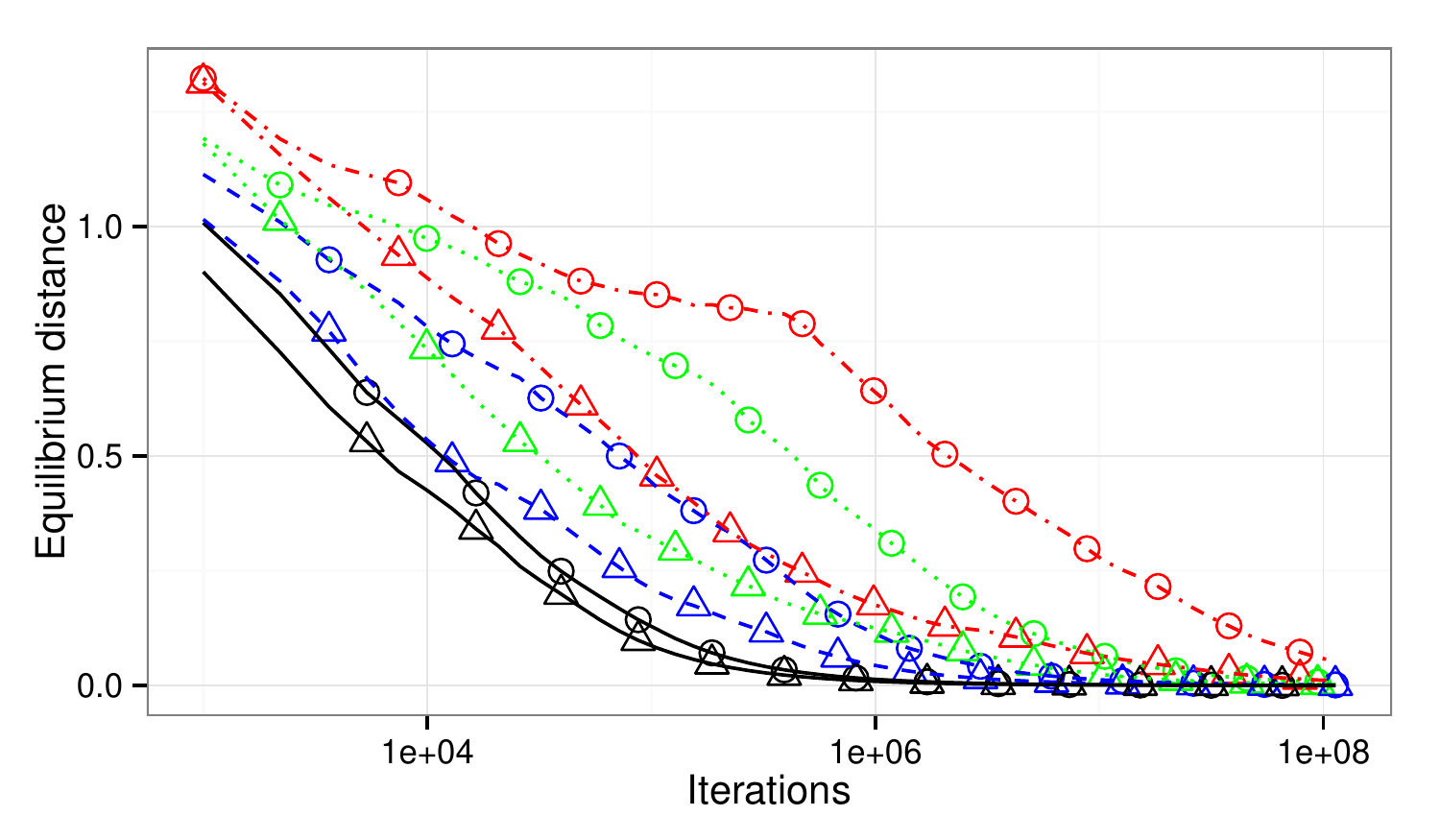}
\caption{Oshi-Zumo(5), Exp3}\label{fig:convAnti_Exp3b}
\end{subfigure}
\begin{subfigure}{0.45\textwidth}
\includegraphics[width=\textwidth]{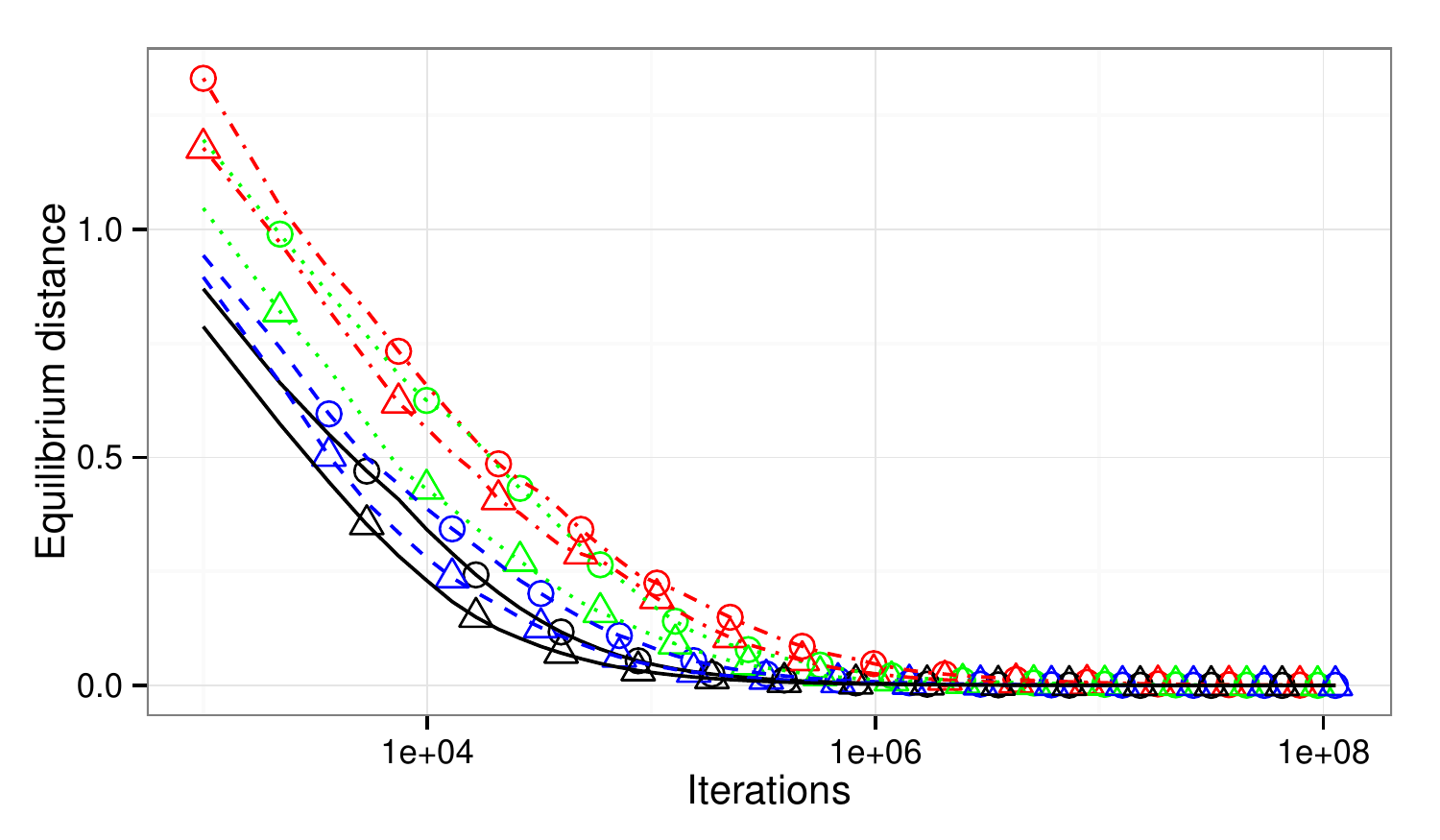}
\caption{Oshi-Zumo(5), RM}\label{fig:convAnti_RMb}
\end{subfigure}

\begin{subfigure}{0.45\textwidth}
\includegraphics[width=\textwidth]{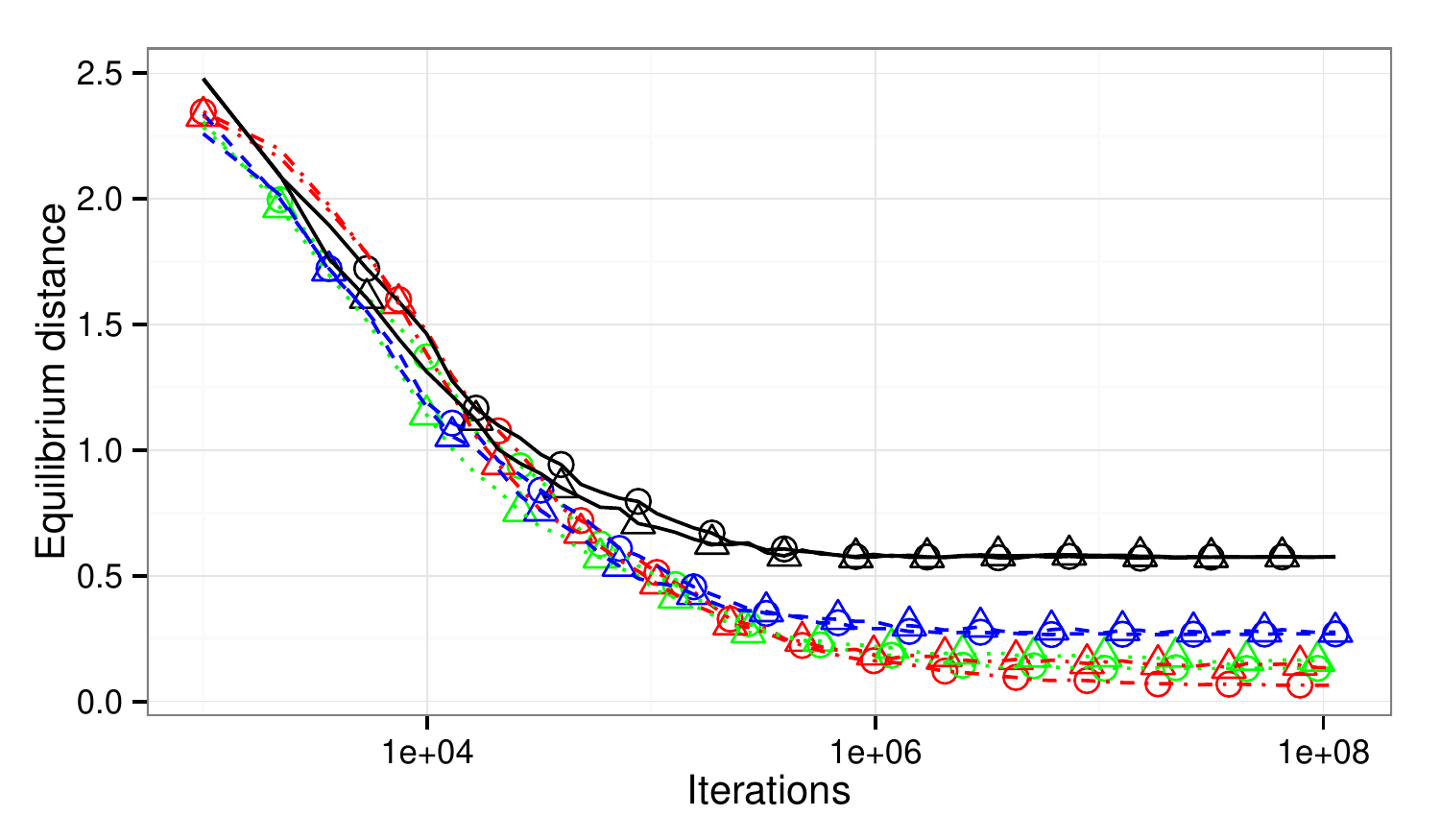}
\caption{Random(3,3), Exp3}\label{fig:convAnti_Exp3c}
\end{subfigure}
\begin{subfigure}{0.45\textwidth}
\includegraphics[width=\textwidth]{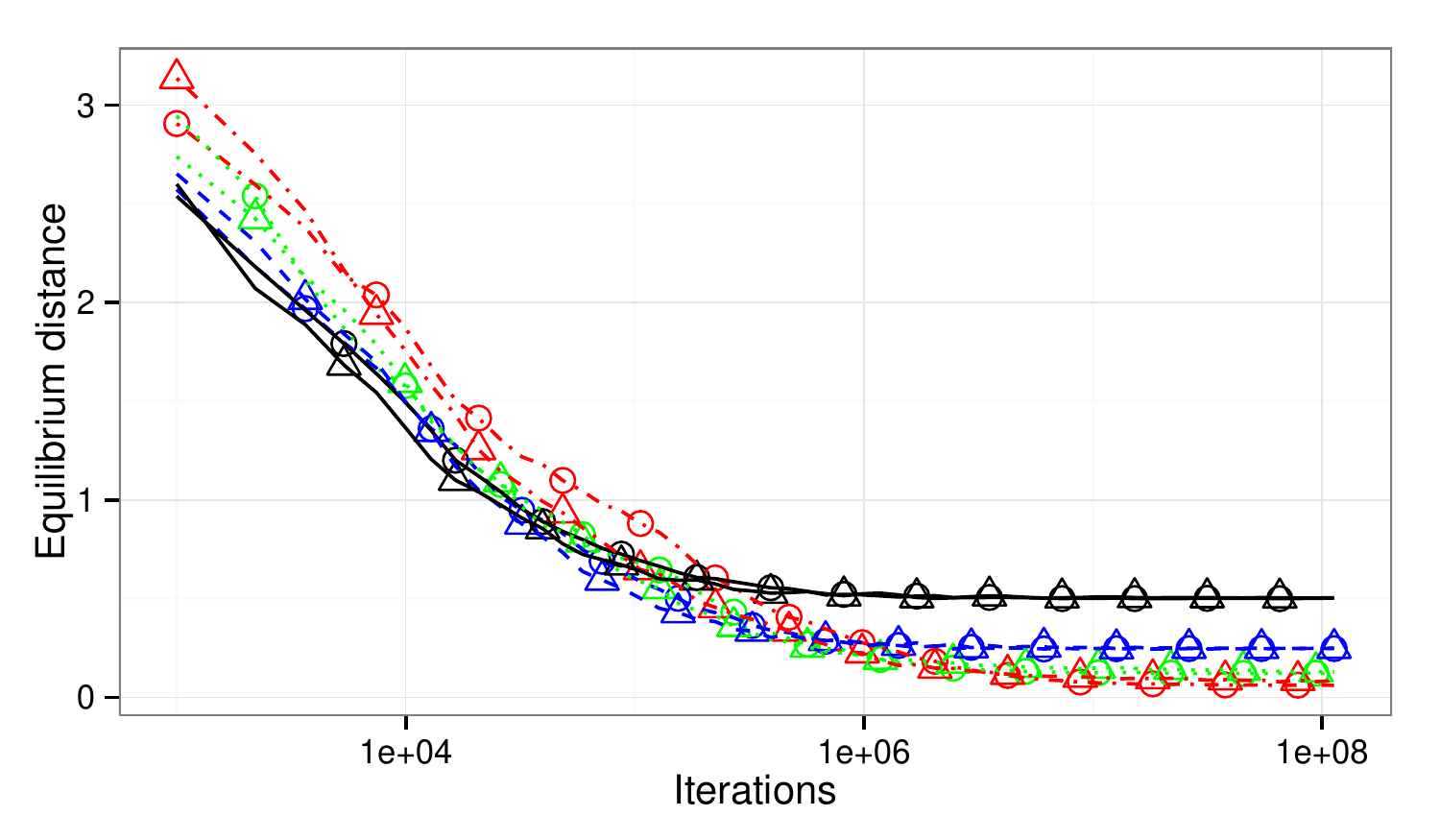}
\caption{Random(3,3), RM}\label{fig:convAnti_RMc}
\end{subfigure}
\caption{Comparison of empirical convergence rates of SM-MCTS (triangles) and SM-MCTS-A (circles) with Exp3 and RM selection functions in various domains.}\label{fig:conv}
\end{figure}

In this section, we investigate the~empirical convergence rates of the~analyzed algorithms. We first compare the~speeds of convergence of SM-MCTS-A and SM-MCTS and then investigate the~dependence of the~error of the~eventual solution of the~algorithms on relevant parameters. Finally, we focus on the~effect of removing the~exploration samples discussed in Section~\ref{sec: exploitability}.
In all cases, ``equilibrium distance'' is measured by the exploitability $\textrm{expl}(\cdot)$ of the evaluated strategy.

\subsubsection{SM-MCTS with and without averaging}

Figure~\ref{fig:conv} presents the~dependence of the~exploitability of the~strategies produced by the~algorithms on the~number of executed iterations. We compare the empirical speed and quality of convergence of SM-MCTS and SM-MCTS-A with different setting of the exploration parameter, however, the~samples caused by exploration are removed from the~strategy for both algorithms, as suggested in Section~\ref{sec: exploitability}. All iterations are executed from the~root of the~game. The~colors (and line types) in the~graphs represent different settings of the~exploration parameter. SM-MCTS-A (circles) seems to always converge to the~same distance from the~equilibrium as SM-MCTS (triangles), regardless of the~used selection function and exploration setting.
The convergence of the~variant with averaging is generally slower.
The difference is most visible in the~Anti game with Exp3 selection function, where the~averaging can cause the~convergence to require even ten times more iterations to reach the~same distance from a~NE as the~algorithm without averaging. The~situation is similar also in Oshi-Zumo. However, the~effect is much weaker with RM selection. With the~exception of the~Anti game, the~variants with and without averaging converge at almost the~same speed.


\subsubsection{Distance from the~equilibrium}\label{sec:emp_eq_dist}

Even though the~depth of most games in Figure~\ref{fig:conv} was 5, even with large exploration ($0.4$), the~algorithm was often able to find the~exact equilibrium. This indicates that in practical problems, even the~linear bound on the~distance from the~equilibrium from the~example in Section~\ref{sec:Discussion} is too pessimistic.

If the~game contains pure Nash equilibria, as in Anti and the~used setting of Oshi-Zumo, exact equilibrium can often be found.
If (non-uniform) mixed equilibria are required, the~gap between the~final
solution and the~equilibrium increases both with the~depth of the~game and the~amount of exploration.
The effect of the~amount of exploration is visible in Figure~\ref{fig:convGS_Exp3}, where the~largest exploration prevented the~algorithm from converging to the~exact equilibrium. A more gradual effect is visible in Figures~\ref{fig:conv}(g,h), where the~exploitability seems to increase linearly with increasing exploration.
Note that in all cases, the~exploitability (computed as the~sum $\textnormal{expl}_1+\textnormal{expl}_2$) was less than $2\cdot\epsilon M$, where $\epsilon$ is the~amount of exploration and M is the~maximum utility value.

\begin{figure}[t]
\centering
\includegraphics[width=0.4\textwidth]{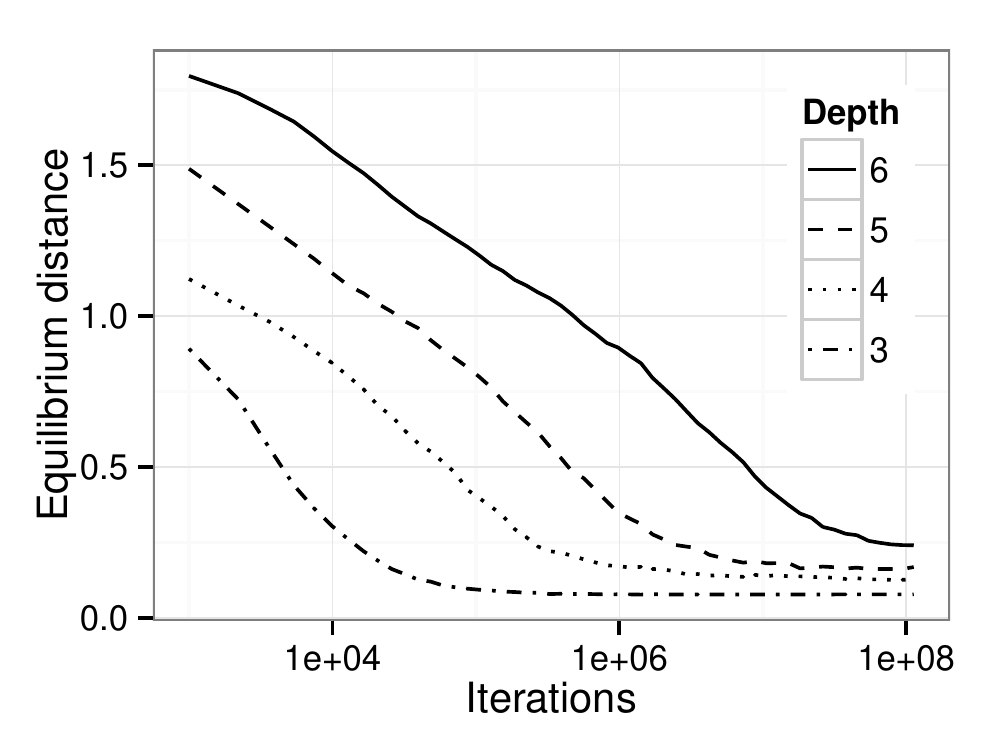}
\caption{Convergence (measured by exploitability) of SM-MCTS with Exp3 selection on random games with three actions of each player in each stage and various depths.}\label{fig:convD}
\end{figure}

Figure~\ref{fig:convD} presents the~exploitability of SM-MCTS with Exp3 selection and exploration $\epsilon=0.2$ in random games with $B=3$ and various depths. The~eventual error increases with increasing depth, but even with the~depth of 6, the~eventual error was on average around 0.25 and always less than $0.3$.

\section{Conclusion\label{sec:Discussion}}

We studied the~convergence properties of learning dynamics in zero-sum extensive form games with simultaneous moves, but otherwise perfect information (SMGs). These games generalize normal form games and are a~special case of extensive form games. Hannan consistent action selection by each player is known to converge to a~Nash equilibrium in normal form games.
In general extensive form games, the~convergence to a~Nash equilibrium is guaranteed if both players use in each of their decision points a~separate selection strategy that is Hannan consistent with respect to the~counterfactual values.
However, with a~single simulation of the~game play, the estimates of these values suffers from a~large variance. This causes the~simulation-based algorithms that use counterfactual values to converge slowly in larger games.
We investigated whether the~convergence can be achieved without using counterfactual values in SMGs.
The studied learning dynamic directly corresponds to a~class of Monte Carlo Tree Search algorithms commonly used for playing SMGs. There was very little pre-existing theory that would describe the~behavior of the~algorithms in these games and provide guarantees on their performance.

Our main results show that using a~separate Hannan consistent algorithm to choose actions in each decision point of an~SMG, which is the~standard SM-MCTS, is not sufficient to guarantee the~convergence of the~learning dynamic to a~Nash equilibrium. This guarantee is obtained if the~used selection function, in addition to being HC, guarantees the~Unbiased Payoff Observations property. We hypothesize that Exp3 and regret matching selection functions have this property, and provide experimental evidence supporting this claim.
Alternatively, the convergence can be guaranteed with any HC selection, if one is willing to use $\textnormal{SM-MCTS-A}$ and accept the (empirically) slower convergence rate caused by averaging the back-propagated values (Section~\ref{sec:emp_conv_rate}).

Our other results are as follows.
In Section~\ref{sec: exploitability}, we provide an~analysis of a~previously suggested improvement of SM-MCTS which proposes to remove the~exploration samples from the algorithm's output. We provide formal and empirical evidence which suggests that this modification is sound and generally improves the~strategy.

Table~\ref{tab:summary} summarizes some of the more detailed results. In Theorem \ref{thm: SM-MCTS-A convergence} (\ref{thm: SM-MCTS convergence}) we show that SM-MCTS-A (SM-MCTS) algorithm with an~$\epsilon$-HC selection function eventually converges at least to a~$C\epsilon$-NE of a~game, where for game depth $D$, $C$ is of the~order $D^2$ ($2^D$). In Section~\ref{sec:Counterexample}, we show that the~worst-case dependence of $C$ on $D$ cannot be sublinear, even after the~exploration is removed. This gives us both lower and upper bounds on the~value of $C$, but it remains to determine whether these bounds are tight.
The same lower bound holds for both SM-MCTS and SM-MCTS-A, which means they may in fact converge to the same distance from the equilibrium, as suggested by the experiments in Section~\ref{sec:emp_eq_dist}.
We hypothesize that the~upper bound bound is not tight (and that after exploration-removal, the optimal value is $C=2D$).

\begin{table}
\small
\renewcommand{\arraystretch}{1.5}
\begin{center} \begin{tabular}{ r|c|c|c| }
\multicolumn{1}{r}{}
 &  \multicolumn{1}{c}{SM-MCTS-A}
 & \multicolumn{2}{c}{SM-MCTS~~~} \\
\cline{2-4}
Assumptions & $\epsilon$-HC & $\epsilon$-HC, $\epsilon$-UPO & $\epsilon$-HC only\\
\cline{2-4}
Upper bound & $2D(D+1)\epsilon$ & $\left(12(2^D-1)-8D\right)\epsilon$ & might not converge\\
\cline{2-3}
Lower bound & $2D\epsilon$ & $2D\epsilon$ & to an~approx. NE at all\\
\cline{2-4}
\end{tabular}
\end{center}
\caption{Summary of the~proven bounds on the~worst-case eventual exploitability of the~strategy produced by SM-MCTS(-A) run with an~$\epsilon$-Hannan consistent selection function}\label{tab:summary}
\end{table}


\subsection{Open problems and future work}
While this paper provides a~significant step towards understanding learning and Monte Carlo tree search in simultaneous-move games, it also leaves some problems open.
First of all, SM-MCTS is used because of the speed with which it can find a good strategy.
However, a practically useful bound on this speed of convergence is missing.
Furthermore, many of the~guarantees presented in the~paper have not been shown to be tight, so an~obvious future research direction would be to improve the~guarantees or show the~tightness of these results (Problem~\ref{prob: linearity in D}). These could be further improved by allowing the exploration rate to depend on the position in the game tree.
Finally, a~better description of the~properties (of the~selection policy) which guarantee that SM-MCTS converges to an optimal strategy could be provided.
In particular, it would be interesting to better understand the UPO property and, more generally, the propagation of information and uncertainty in game trees with simultaneous moves. This would in turn enable either proving that the~common no-regret bandit algorithms indeed produce unbiased payoff observations, or finding some suitable alternative to the UPO property.
Moreover, focusing on these specific (non-pathological) algorithms or their suitable modifications might simplify the derivation of tighter performance bounds.

\section*{Acknowledgements}
We are grateful to Tor Lattimore for helpful comments on earlier version of this paper, and to the anonymous reviewers for their patience and feedback.
This research was supported by the~Czech Science Foundation (grants no. P202/12/2054 and 18-27483Y).
Access to computing and storage facilities owned by parties and projects contributing to the~National Grid Infrastructure MetaCentrum, provided under the~programme "Projects of Large Infrastructure for Research, Development, and Innovations" (LM2010005), is appreciated.

\begin{table}[htbp]
\begin{framed}
\centering
\begin{tabular}{ r l }
& \textbf{Abbreviations}\\
\text{SM-MCTS(-A)} & \text{(averaged) simultaneous-move Monte Carlo tree search}\\
\text{NE} & \text{Nash equilibrium}\\
\text{MAB} & \text{multi-armed bandit}\\
\text{HC} & \text{Hannan consistent}\\
\text{Exp3} & \text{exponential-weight algorithm for exploration and exploitation}\\
\text{RM} & \text{regret matching algorithm}\\
\text{UPO} & \text{unbiased payoff observations}\\
\text{CFR} & \text{the counterfactual regret minimization algorithm} \\
\text{EFG} & \text{extensive form game}\\
\\
& \textbf{Notation related to a multi-armed bandit problem $(P)$}\\
$x_i(t)$		& the reward for action $i$ at time $t$\\
$i(t)$ 			& the action chosen at time $t$\\
$\sigma$, $\bar \sigma$, $\hat \sigma$ & strategy (at the given time, average, and empirical frequencies)\\
$G(t)$, $G_{\max}(t)$ & the actually obtained (resp. maximum achievable) cumulative reward\\
$R(t)$ & the external regret at time $t$\\
$g$, $g_{\max}$, $r$ & the averaged variants of $G$, $G_{\max}$, and $R$\\
$\gamma$, $\epsilon$ & exploration rate, $\epsilon$-Hannan consistency parameter\\
\\
& \textbf{Matrix games}\\
$M = (v^M_{ij})_{i,j}$ & matrix game\\
$u^M$, $v^M$ & utility in $M$, minimax value of $M$\\
$br$ & best response\\
$(P^M)$ & MAB problem corresponding to $M$\\
$\widetilde M(t) = (v^{\widetilde M}_ij(t))_{ij}$ & repeated matrix game with bounded distortion\\
$(P^{\widetilde M})$, $R^{\widetilde M}$, $g^{\widetilde M}$, \dots & the corresponding MAB problem and related variables\\
\\
& \textbf{Notation related to simultaneous-move games...}\\
$h$, $v^h$, $t^h$ & state in the game, is value, and the number of visits\\
$i^h(t)$, $j^h(t)$ & actions selected at $h$ during its $t$-th visit\\
$t^h_i$, $t^h_j$, $t^h_{ij}$ & variables related to the number of uses of each (joint) action\\
$D$, $d_h$ & game depth, depth of the sub-tree rooted at $h$\\
& \textbf{...and the corresponding MAB problems}\\
$(P^h)$, $G^h$, $x^h$, \dots & player 1 MAB problem corresponding to SM-MCTS at $h$, related variables\\
$r^h_1$, $r^h_2$ & average regret corresponding to $(P^h)$, resp. player 2 variant of $(P^h)$\\
$h_{ij}$, $v^h_{ij}$, $x^h_{ij}$, \dots & state resulting from taking $(i,j)$ at $h$, corresponding variables $v^{h_{ij}}$, $x^{h_{ij}}$, \dots \\
$(\bar P^h)$ & player 1 MAB problem corresponding to SM-MCTS-A at $h$\\
$\bar x^h_{ij}(n)$, $\tilde x^h_{ij}(n)$ & the arithmetical (resp. weighted) average of values $x^h_{ij}(1)$, $\dots$, $x^h_{ij}(n)$\\
\end{tabular}
\end{framed}
\caption{The abbreviations and common notation for quick reference}\label{tab:notation}
\end{table}

\bibliography{refs}


\appendix
\section*{Appendix}

The appendix contains the~proofs for those results which have not been already proven in the~main text.

\section{Proofs of Lemma~\ref{lem:emp_and_avg} and \ref{lemma: A* je HC}{}}

We start with the~proof of Lemma \ref{lem:emp_and_avg}, which states that eventually, there is no difference between the~empirical and the~average strategies.

\empaavg*
\begin{proof}
It is enough to show that $ \underset{t\rightarrow\infty}{\limsup}\, |\hat{\sigma_1}^h(t)(i)-\bar{\sigma_1}^h(t)(i)|=0$ holds almost surely for any $h$ and $i$. Without changing the~limit behavior of $\hat \sigma$, we can assume that $t^h_i$ is defined as the~number of uses of $i$ during iterations $1$, $2$, \dots, $t$.\footnote{See the~definition of $t^h_i$ and footnote \ref{footnote: t_i}.} Using the~definitions of $\hat{\sigma_1}^h(t)(i)$ and $\bar{\sigma_1}^h(t)(i)$, we get
\[
\hat{\sigma_1}^h(t)(i)-\bar{\sigma_1}^h(t)(i)
  = \frac{1}{t} \left(t_i - \sum_{s=1}^t \sigma_1^h(s)(i)\right)
  = \frac{1}{t} \sum_{s=1}^t \left(\delta_{i, i^h(s)} - \sigma_1^h(s)(i) \right),
\]
where $\delta_{i,j}$ is the~Kronecker delta. Using the~(martingale version of) Central Limit Theorem on the~sequence of random variables $X_t=\sum_{s=1}^t \left(\delta_{i, i^h(s)} - \sigma_1^h(s)(i) \right)$ gives the~result (the conditions clearly hold, since $\mathbf{E}\left[\delta_{i,i^h(t)} - \sigma_1^h(t)(i)\big| X_1,...,X_{t-1}\right]=0$ implies that $X_t$ is a~martingale and $\delta_{i,i^h(t)} - \sigma_1^h(t)(i)\in\left[-1,1\right]$ guarantees that all required moments are finite).
\end{proof}

Next, we prove Lemma \ref{lemma: A* je HC}, which states that $\epsilon$-Hannan consistency is not substantially affected by additional exploration.
\AjeHC*
\begin{proof}
$(i)$: The~guaranteed exploration property of $A^{\gamma}$ is trivial. Denoting by {*} the~variables corresponding to the~algorithm $A^{\gamma}$
we get
\begin{eqnarray*}
r^{*}(t) & = & \frac{1}{t}R^{*}(t)\leq\frac{1}{t}\left(1\cdot t_{\textrm{ex}}+R(t-t_{\textrm{ex}})\right) \\
 & = & \frac{t_{\textrm{ex}}}{t}+\frac{R(t-t_{\textrm{ex}})}{t-t_{\textrm{ex}}}\cdot\frac{t-t_{\textrm{ex}}}{t},
\end{eqnarray*}
where, for given $t\in\mathbb{N}$, $t_{\textrm{ex}}$ denotes the~number of times $A^{\gamma}$ explored up to $t$-th iteration. By the~strong law of large numbers we have that $\underset{t\rightarrow\infty}{\lim} \ t_{\textrm{ex}}/ t=\gamma$ holds almost surely. This implies 
\begin{eqnarray*}
\limsup_{t\rightarrow\infty}r^{*}\left(t\right) & \leq & \limsup_{t\rightarrow\infty}\frac{t_{\textrm{ex}}}{t}+\limsup_{t-t_{\textrm{ex}}\rightarrow\infty}\frac{R\left(t-t_{\textrm{ex}}\right)}{t-t_{\textrm{ex}}}\cdot\limsup_{t\rightarrow\infty}\frac{t-t_{\textrm{ex}}}{t}\\
 & \leq & \gamma+\epsilon\left(1-\gamma\right)\\
 & \leq & \gamma+\epsilon,
\end{eqnarray*}
which means that $A^{\gamma}$ is $\left(\epsilon+\gamma\right)$-Hannan
consistent.

$(ii)$: The guaranteed exploration follows from the fact that if the second player also uses $A^{\sqrt{\cdot}}$, the probability that both players explore at once is at least $\frac{1}{\sqrt{t}} \cdot \frac{1}{\sqrt{t}} = \frac{1}{t}$. Since $\sum \frac{1}{t} = \infty$, this will happen infinitely many times, so each joint action will be sampled infinitely many times. The proof of $\epsilon$-Hannan consistency is similar to $(i)$ and we omit it.
\end{proof}

\section{The counterexample for Theorem~\ref{thm: SM-MCTS convergence}}

We first present some details related to the UPO property, and then proceed to give a proof of Lemma~\ref{lem: A_I properties} from Section~\ref{sec: HC counterexample}.

\subsection{Examples and a sufficient condition for the UPO property}\label{sec:UPO examples}

We present a~few examples which are related to the UPO property and support the~discussion that follows. We first give the~negative examples, and then continue with the~positive ones.

\begin{example}[Examples related to the~UPO property]\label{ex: UPO}$\ $
\begin{enumerate}[(i)]
\item Application to Section~\ref{sec: deterministic counterexample}: Assume that $(x(n))_{n=1}^\infty=(1,0,1,0,1,...)$ and $(w(n))_{n=1}^\infty=(1,3,1,3,1,...)$. Then we have $\bar x(n)\rightarrow\frac 1 2$, but $\tilde x(n)\rightarrow\frac 1 4$.
\item In Section \ref{sec:Counterexample} we constructed a~HC algorithm which produces a~pattern nearly identical to $(i)$ and leads to $\underset {n\rightarrow\infty} \limsup \ |\bar x^{h_0}_{ij}(n)-\tilde x^{h_0}_{ij}(n)|\geq \frac 1 4$ (in a~certain game).
\item Suppose that $w(n)$, $x(n)$, $n\in\mathbb{N}$ do not necessarily originate from SM-MCTS, but assume they satisfy $(a)$ and $(b)$:
\begin{enumerate}
\item $w(n)\in[0,C]$ and $x(n)\in[0,1]$  are independent random variables.
\item $x(n)$ have expectations approximately constant over $n\in\mathbb N$:
\begin{equation} \label{eq: UPO induction hypothesis}
\exists v\in[0,1] \forall n\in\mathbb{N}:\ \mathbf{E}[x(n)]\in\left[ v-\frac{\epsilon}{2},v+\frac{\epsilon}{2}\right].
\end{equation}
\end{enumerate}
Then, by strong law of large numbers, we almost surely have
\[ \underset {n\rightarrow\infty} \limsup |\bar x(n)-v|\leq\frac{\epsilon}{2} \ \textrm{ and } \ \underset {n\rightarrow\infty} \limsup |\tilde x(n)-v|\leq\frac{\epsilon}{2}, \]
which leads to
\begin{equation}
\underset {n\rightarrow\infty} \limsup |\bar x(n)-\tilde x(n)|\leq\epsilon. \label{eq: X}
\end{equation}
\item The~previous case can be generalized in many ways -- for example it is sufficient to replace bounded weights $w(n)$ by ones satisfying
\[ \exists q\in(0,1)\ \forall n\ \forall i,j:\ \mathbf{Pr}[w(n)\geq k]\leq q^k \]
(an assumption which holds with $q=\gamma/\left|\mathcal{A}_1(h)\right|$ when $w(n)$, $x(n)$  originate from SM-MCTS with fixed exploration). Note that the weights still need to be sufficiently independent, to avoid the pathological behavior from Section \ref{sec:Counterexample}.
\item In Section \ref{sec:Experimental} we provide empirical evidence which suggests that when the~variables $x(n),\ w(n)$ originate from SM-MCTS with Exp3 or RM selection policy, then the~assertion \eqref{eq: X} holds as well (and thus these two $\epsilon$-HC algorithms are likely $\epsilon$-UPO).
\end{enumerate}
\end{example}
Parts $(i)$ and $(ii)$ from Example \ref{ex: UPO} show that the~implication $(\textrm{A is }\epsilon\textrm{-HC}$ $\implies$ $\textrm{A is }$ $C\epsilon\textrm{-UPO})$ does not hold, no matter how huge is the~constant $C>0$. On the~other hand, $(iii)$ and $(iv)$ suggest that it is possible to prove that specific $\epsilon$-HC algorithms are $\epsilon$-UPO. However, the~guarantees we have about the~behavior of, for example, Exp3 are much weaker than the~assumptions made in $(iii)$ -- there is no independence between $w_{ij}(n), x_{ij}(m), \ m,n\in\mathbb{N}$, at best we can use some martingale theory. Moreover, even in nodes $h\in\mathcal{H}$ with $d_h=1$, we only have $\underset {n\rightarrow\infty} \limsup |\bar x_{ij}(n)-v^h_{ij}|\leq \epsilon$, instead of assumption  \eqref{eq: UPO induction hypothesis} from $(iii)$ $(b)$.

The observations above suggest that the~proof that specific $\epsilon$-HC algorithms are $\epsilon$-UPO will not be trivial. However, we can at least state a~sufficient condition for an~algorithm to be $\epsilon$-UPO, which is not dependent on the~tree structure of simultaneous-move games.

\begin{remark}[Sufficient condition for $\epsilon$-UPO]
Let $A$ be a~bandit algorithm, $K,L\in\mathbb N$ and for each $i\leq K$, $j\leq L$, let $(x_{ij}(n))_{n=1}^\infty$ be an~arbitrary sequence in $[0,1]$. Assume that player 1 uses $A$ to choose actions $i\in\{1,...,K\}$ and player 2 uses a~different instance of $A$ to choose actions $j\in\{1,....,L\}$, where the~reward for playing actions $(i,j)$ at time $t$ is $x_{ij}(t_{ij})$ (resp. $1-x_{ij}(t_{ij})$ for player 2). As before, denote 
\begin{align*}
\bar{x}_{ij}\left(n\right) & := \frac{1}{t} \sum_{m=1}^{n}x_{ij}\left(m\right),\\
w_{ij}\left(m\right) & := \left|\left\{ t\in\mathbb{N}|\,t_{ij}=m\ \& \ j(t)=j\right\}\right| , \\
\tilde{x}_{ij}\left(n\right) & := \frac{1}{\sum_{m=1}^{n}w_{ij}\left(m\right)}\sum_{m=1}^{n}w_{ij}\left(m\right)x_{ij}\left(m\right).
\end{align*}
Clearly, if
\[ \forall i\leq K \ \forall j\leq L:\ \underset {n\rightarrow\infty} \limsup |\bar x_{ij}(n)-\tilde x_{ij}(n)|\leq\epsilon \ \textrm{ a.s.} \]
holds for all such sequences $(x_{ij}(n))_{n=1}^\infty$, then $A$ is $\epsilon$-UPO.
\end{remark}
This condition should be easier to verify in practice than the~original $\epsilon$-UPO definition. On the~other hand, the~quantification over all possible sequences $(x_{ij}(n))_{n=1}^\infty$ might be a~too strong requirement. Moreover, it would be even better to have an~equivalent -- or at least sufficient, yet not too strong -- condition for an~algorithm to be $\epsilon$-UPO which only operates with one player.

\subsection{Details related to the~counterexample for Theorem \ref{thm: SM-MCTS convergence}} \label{sec:appendix:counterexample}

In Section \ref{sec:Counterexample} (Lemma \ref{lemma: modification of A}) we postulated the~existence of algorithms, which behave similarly to those from Section~\ref{sec: deterministic counterexample}, but unlike those from Section~\ref{sec: deterministic counterexample}, the~new algorithms are $\epsilon$-HC:
\modificationofA*
First, we define these algorithms, and then we prove their properties in Lemma \ref{lem: A_I properties}. Lemma \ref{lemma: modification of A}, then follows directly from Lemma \ref{lem: A_I properties}.

\begin{remark}
The intention behind the~notation which follows is that $ch$ denotes how many times the~other player ``cheated'', while $\bar{ch}$ is the~average ratio of cheating in following the
cooperation pattern. The~variables $\tilde{ch}$ and $\tilde{\bar{ch}}$
then serve as estimates of $ch$ and $\bar{ch}$. We present the~precise definitions below. The~nodes I, J, actions X, Y, L, R, U and D and the~respective payoffs refer to the~game $G$ from Figure \ref{fig: hra}.
\end{remark}
\textbf{Definition of the~algorithm $A_{1}^{J}$:}

Fix an~increasing sequence of integers $b_{n}$ and repeat for $n\in\mathbb{N}$:

\begin{enumerate}
\item \textbf{Buffer building $B_{n}$: }Play according to some $\epsilon$-HC
algorithm for $b_{n}$ iterations (continuing with where we left of
in the~$\left(n-1\right)$-th buffer building phase).
\item \textbf{Cooperation $C_{n}$:} Repeat $U,U,D,D$ for $t=1,2,...$
and expect the~other player to repeat $L,R,R,L$. At each iteration
$t$, with probability $\epsilon$, check whether the~other player
is cooperating -- that is play the opposite of the pattern-prescribed action and if the~payoff does not correspond to the~expected pattern the~second player should be following, set $\tilde{ch}\left(t\right)=\frac{1}{\epsilon}$. If the~other player passes this check, or if we did not perform it,
set $\tilde{ch}\left(t\right)=0$.
\item \textbf{End of cooperation }(might not happen)\textbf{:} While executing
step 2, denote $\tilde{\bar{ch}}\left(t\right):=\frac{1}{t}\sum_{s=1}^{t}\tilde{ch}\left(s\right)$.
Once $t$ satisfies 
\begin{equation}\label{eq:end_coop}
\epsilon\cdot\frac{b_{n}}{b_{n}+t}+1\cdot \frac{t}{b_{n}+t}\geq2\epsilon,
\end{equation}
we check at each iteration whether the~estimate $\tilde{\bar{ch}}\left(t\right)$
threatens to exceed $2\epsilon$ during the~next iteration or not.
If it does, we end the~cooperation phase, set $n:=n+1$ and continue
by the~next buffer building phase.
\item \textbf{Simulation of the~other player:} While repeating steps 1, 2 and 3, we simulate the~other player's algorithm $A_{2}^{J}$ (this is possible since from the~knowledge of our action and the~received payoff we can recover the~adversary's action).
If it ends the~cooperation phase and starts the~next buffer building
phase, we do the~same.
\item Unless the~cooperation phase is terminated, we stay in phase $C_n$ indefinitely.
\end{enumerate}

\noindent \textbf{Definition of $A_2^J$:} The~algorithm $A_{2}^{J}$ is identical to $A_{1}^{J}$, except for
the fact that it repeats the~pattern U, D, D, U instead of L, R, R, L and expects
the other player to repeat L, R, R, L.

\noindent \textbf{Definition of $A^I$:} The~algorithm $A^{I}$ is a
straightforward modification of $A_1^{J}$ (replacing $2\epsilon$ and $epsilon$ by $2\epsilon$ and $3\epsilon$ in Eq. \eqref{eq:end_coop}) -- it repeats the~sequence Y, X, X, Y and expects to
receive payoffs 0, 0, 0, ... whenever playing $X$ and payoffs 1, 0, 1, 0, ...
when playing $Y$. However, whenever it deviates from the~Y, X, X, Y
pattern to check whether these expectations are met, it plays
the same action once again (to avoid disturbing the~payoff
pattern for $Y$).

\begin{remark}\label{rem: A_IJ is correct}
The steps 2 and 3 from the~algorithm description are correctly defined because the~opponent's action choice can be recovered from the~knowledge of our action choice and the~resulting payoff. Regarding step 3, we note that the~condition here is trivially satisfied for $t=1,...,\epsilon b_{n}$, so the~length $t_n$ of the~cooperation phases $C_{n}$ tends to infinity as $n\rightarrow\infty$, regardless of the~opponent's actions.
\end{remark}

\begin{lemma}\label{lem: A_I properties}
$\left(1\right)$ When facing each other, the~average strategies of algorithms $A^{I},\, A_{1}^{J}$ and $A_{2}^{J}$ will converge to the~sub-optimal
strategy $\sigma^{I}=\sigma_{1}^{J}=\sigma_{2}^{I}=\left(\frac{1}{2},\frac{1}{2}\right)$. However the~algorithms will suffer regret no higher than $C\epsilon$ for some $C>0$ (where $C$ is independent of $\epsilon$).

$\left(2\right)$ There exists a~sequence $b_{n}$ (controlling the~length of phases $B_{n}$), such that when facing a~different adversary, the~algorithms
suffer regret at most $C\epsilon$. Consequently $A^{I},\, A_{1}^{J}\textrm{ and }A_{2}^{J}$
are $C\epsilon$-Hannan consistent.
\end{lemma}
\begin{proof}
\textbf{Part $\left(1\right)$:} Note that, disregarding the~checks made by the~algorithms, the~average strategies in cooperation phases converge to $\left(\frac{1}{2},\frac{1}{2}\right)$. Furthermore, the~probability of making the~checks is the~same, regardless on which action is supposed to be played next. Therefore if the~algorithms eventually settle in cooperative phase, the~average strategies converge to $\left(\frac{1}{2},\frac{1}{2}\right)$.

\vspace{5mm}
\noindent \textbf{Step $\left(i\right)$: The~conclusion of (2) holds for $A_{p}^{J}$, $p=1,2$.}

We claim that both algorithms $A_{p}^{J}$, $p=1,2$ will
eventually settle in the~cooperative mode, thus generating the~payoff
sequence $1,0,1,0$ during at least $\left(1-\epsilon\right)^{2}$-fraction
of the~iterations and generating something else at the~remaining $1-\left(1-\epsilon\right)^{2}<2\epsilon$
iterations. It is then immediate that the~algorithms $A_{p}^{J}$
suffer regret at most $2\epsilon$ (since as observed in Section~\ref{sec: deterministic counterexample}, the~undisturbed pattern produces no regret).

\noindent  \textbf{Proof of $\left(i\right)$:} If the~other player uses the~same algorithm,
we have $\mathbf{E}\left[\tilde{ch}\left(t\right)\right]=\epsilon$
and $\mathbf{Var}\left[\tilde{ch}\left(t\right)\right]\leq\frac{1}{\epsilon}<\infty$,
thus by the~Strong Law of Large Numbers $\tilde{\bar{ch}}\left(t\right)\longrightarrow\epsilon$ holds almost surely.
In particular, there exists $t_{0}\in\mathbb{N}$ such that
\[ \mathbf{Pr}\left[\forall t\geq t_{0}:\,\tilde{\bar{ch}}\left(t\right)\leq2\epsilon\right]>0. \]
In Remark \ref{rem: A_IJ is correct} we observed that the~cooperative
phase always lasts at least $\epsilon b_{n}$ steps. Therefore once
we have $b_{n}\geq\frac{1}{\epsilon}t_{0}$, there is non-zero probability
both players will stay in $n$-th cooperative phase forever. If they
do not, they have the~same positive probability of staying in the~next cooperative
phase and so on -- by Borel-Cantelli lemma, they will almost surely
stay in $C_{n_{0}}$ for some $n_{0}\in\mathbb{N}$.

\vspace{5mm}
\noindent \textbf{Step $\left(ii\right)$: $A^{I}$ will eventually settle in the~cooperative mode.}

\noindent \textbf{Proof of $\left(ii\right)$:} This statement can be proven by a~similar argument as $\left(i\right)$. The~only difference is that instead of checking whether the~other player is cheating, we check whether the~payoff sequence is the~one expected by $A^I$. The~fact that $A^{I}$ settles in cooperative mode then immediately implies that $A^{I}$ will repeat the~Y, X, X, Y pattern during approximately a~$\left(1-2\epsilon\right)$-fraction
of the~iterations (we check with probability $\epsilon$, and we always check twice). It is then also immediate that $A^{I}$ will receive
the expected payoffs during at least $\left(1-\epsilon\right)$-fraction of the~iterations. (The payoff is always $0$, as expected, when $X$ is played, and it is correct in at least $1-2\epsilon$ cases when playing $Y$. $X$ and $Y$ are both played with same frequency, which gives the~result.)

\vspace{5mm}
\noindent \textbf{Step $\left(iii\right):$ If $A_{p}^{J}$, $p=1,2$ and $A^{I}$ settle in cooperative mode, then $A^{I}$ suffers regret at most $C\epsilon$.}

\noindent \textbf{Proof of $\left(iii\right)$:} Recall the~definition \eqref{equation: equivalent MAB description of SM-MCTS} of $x^h_i(t)$, definition of $t^h_i$ which follows it, and the notation introduced in \eqref{not: property}. We denote by $i\left(t\right):=i^I(t)$ the~$t$-th action chosen by $A^{I}$ and by $(i^*(s))_s$ the~4-periodic of sequence of ``ideal case actions'' starting with Y, X, X, Y. Recall that $x^J\left(n\right)$ is the~payoff obtained during the~$n$-th visit of $J$ and $x^I_{Y}\left(t\right)=x^J(t_Y)$ is the~payoff we \emph{would} receive for playing $Y$ in $I$ at time $t$ (and $x^I_{X}\left(t\right)=0$ is always zero). Finally we denote by  $(x^*_Y(s))_s$ the~4-periodic of sequence of ``ideal case payoffs'' starting with 1, 0, 0, 0. Our goal is to show that the~average regret 
\[ r^I\left(t\right)=\frac{1}{t}\left(\sum_{s=1}^{t}x^I_{Y}\left(s\right)-\sum_{s=1}^{t}x^I_{i\left(s\right)}\left(s\right)\right) \]
is small. As we already observed in Section~\ref{sec: deterministic counterexample}, if neither the~payoff pattern $(X^*_Y(s))_s = (1,0,1,0,...)$ coming from $J$, nor the~action pattern $(i^*(s))_s=(Y, X, X, Y, ...)$ at $I$ are disturbed, then we have $\left(i(s)\right)=\left(i^*(s)\right)$ and $\left(x^I_Y(s)\right)=\left(x^*_Y(s)\right)$, and therefore $A^{I}$ suffers no regret.
As long as the algorithms stay in the cooperative mode, the only way the pattern can be disturbed is when the random checks occur.
The checks at $J$ occur with probability $\epsilon$ for each player, which increases the average asymptotic regret by no more than $2\epsilon$ (these checks will actually cancel out in expectation, but we ignore this).
Moreover, the combination of the~2-periodicity of the~reward-pattern at $J$ and the fact that $A^I$ checks twice in a row guarantees that any check performed at $I$ only affects 2 the two iterations during which it occurs. In the worst case, this increases the average asymptotic regret by no more than $2\epsilon$. Hence the total average regret is at most $4\epsilon$ in the limit.

\vspace{5mm}
\noindent \textbf{Part $\left(2\right)$: Algorithms $A^J_p$, $p=1,2$, and $A^I$ are $C\epsilon$-HC.}

\noindent \textbf{a)} Firstly, we assume that both players stick to their assigned patterns during at least $\left(1-2\epsilon\right)$-fraction of iterations (in other words, assume that $\limsup\,\bar{ch}\left(t\right)\leq2\epsilon$). By the~same argument as in $\left(1\right)$, we can show that the~algorithms will then suffer regret at most $C\epsilon$.

\noindent \textbf{b)} On the~other hand, if $\limsup\,\bar{ch}\left(t\right)>2\epsilon$, the~algorithm will almost surely keep switching between $B_{n}$ and
$C_{n}$ (consequence of Strong Law of Large Numbers). Denote by $r_{b,n}$ and $r_{c,n}$ the~regret from phases $B_{n}$ and $C_{n}$, recall that $b_n, \ t_n$ are the~lengths of these phases and set 
\[
r_{n}:=\frac{b_{n}}{b_{n}+t_{n}}r_{b,n}+\frac{t_{n}}{b_{n}+t_{n}}r_{c,n}=\textrm{overall regret in }B_{n}\textrm{ and }C_{n}\textrm{ together}.
\]
Finally, let $r=\limsup \,r^J\left(t\right)$ denote the~bound on the~limit of regret of $A_{p}^{J}$. We need to prove that $r\leq C\epsilon$. To do this, it is sufficient to show that $\limsup_{n}\, r_{n}$ is small -- thus our goal will be to prove that if the~sequence $b_{n}$ increases quickly enough, then $\limsup\, r_{n}\leq C\epsilon$ holds almost surely. Denote by $\left(F_{n}\right)$ the~formula
\[
\forall t\geq\frac{1}{\epsilon}b_{n}:\,\tilde{\bar{ch}} \left(t\right) \leq 2\epsilon\implies\bar{ch} \left(t\right)\leq3\epsilon.
\]
We know that $\left|\tilde{\bar{ch}}-\bar{ch}\right|\rightarrow0$ a.s., therefore we can choose $b_{n}$ such that 
\[
\mathbf{Pr}\left[\left(F_{n}\right)\textrm{ holds}\right]\geq1-2^{-n}
\]
holds. Since $\sum2^{-n}<\infty$, Borel-Cantelli lemma gives that $\left(F_{n}\right)$ will hold for all but finitely many $n\in\mathbb{N}$. Note that if $\left(F_{n}\right)$ holds for both players, then their empirical strategy is at most $3\epsilon$ away from the~NE strategy, and thus $r_{c,n}\leq C\epsilon$. Since $r_{n}$ is a~convex combination of $r_{b,n}$ and $r_{c,n}$ and in $B_{n}$ we play $\epsilon$-consistently, we can compute 
\begin{equation*}
r\leq\limsup\, r_{n}\leq\max\left\{ \limsup\, r_{b,n},\limsup\, r_{c,n}\right\} \leq\max\left\{ \epsilon,C\epsilon\right\} =C\epsilon.
\end{equation*}
The proof of $C\epsilon$-Hannan consistency of $A^{I}$ is analogous.
\end{proof}

\section{Proofs related to the convergence of SM-MCTS(-A)}

\subsection{Convergence of SM-MCTS-A} \label{sec:SM-MCTS-A proofs}
In this section, we give the~proofs for lemmas which were used to prove Theorem \ref{thm: SM-MCTS-A convergence}. We begin with Lemma \ref{L: HC-and-NE}, which established a~connection between average payoff $g$ and game value $v$ for matrix games.

\HCandNE*
\begin{proof}
It is our goal to draw conclusion about the~quality of the~empirical strategy based on information about the~performance of our $\epsilon$-HC algorithm. Ideally, we would like to relate the~utility $u^M_1(\hat \sigma)$ to the~average payoff $g^M(t)$. However, as this is generally impossible, we can do the~next best thing:
\begin{align*}
u^M_1 \left(br,\hat{\sigma}_{2}(t)\right)=\, & \,\underset{i}{\max}\,\underset{j}{\sum}\hat{\sigma}_{2}(t)(j)v^M_{ij}=\underset{i}{\max}\,\underset{j}{\sum}\frac{t_{j}}{t}v^M_{ij}=\frac{1}{t}\underset{i}{\max}\,\underset{j}{\sum}t_{j}v^M_{ij}\nonumber \\
= \ & \,\frac{1}{t}\underset{i}{\max}\,\underset{s=1}{\overset{t}{\sum}}v^M_{ij(s)}=\frac{1}{t}G_{\max}(t)=g_{\max}(t).\label{eq: u(br, )}
\end{align*}
\textbf{Step 1:} Let $\eta>0$. Using the $\epsilon$-HC property gives us the~existence of such $t_{0}$ that $g^M_{\max}(t)-g^M(t)<\epsilon+\frac{\eta}{2}$ holds for all $t\geq t_{0}$, which is equivalent to $g^M(t)>g^M_{\max}-(\epsilon+\frac{\eta}{2})$. However, in our zero-sum matrix game setting, $g^M_{\max}$ is always at least $v^M$, which implies that $g^M(t)>v^M-(\epsilon+\frac{\eta}{2})$.
Using the~same argument for player $2$ gives us that $g^M(t)<v^M+\epsilon+\frac{\eta}{2}$. Therefore we have the~following statement, which proves the~inequalities \eqref{eq: HC and NE 2}:
\begin{equation*}
\forall t\geq t_0:\ v^M- (\epsilon+\frac{\eta} 2 ) < g^M(t)<v^M+\epsilon+\frac{\eta}{2} \textrm{ holds almost surely.}
\end{equation*}

\vspace{5mm}
\textbf{Step 2:} We assume, for contradiction with inequalities \eqref{eq: HC and NE 1}, that with non-zero probability, there exists an~increasing sequence of time steps $t_{n}\nearrow\infty$, such that $u^M_1\left(br,\hat{\sigma}_{2}(t_{n})\right)\geq v^M+2\epsilon+\eta$ for some $\eta>0$. Combing this with the~inequalities we proved above, we see that 
\begin{eqnarray*}
\limsup_{t\rightarrow\infty}r^M\left(t\right) & \geq & \limsup_{n\rightarrow\infty}r^M\left(t_{n}\right)=\limsup_{n\rightarrow\infty}\left(g^M_{\max}\left(t_{n}\right)-g^M\left(t_{n}\right)\right)\\
 & = & \limsup_{n\rightarrow\infty}\left(u^M_1\left(br,\hat{\sigma_{2}}\left(t_{n}\right)\right)-g^M\left(t_{n}\right)\right)\\
 & \geq & v^M+2\epsilon+\eta-\left(v^M+\epsilon+\eta/2\right) = \epsilon + \eta > \epsilon
\end{eqnarray*}
holds with non-zero probability, which is in contradiction with $\epsilon$-Hannan consistency. 
\end{proof}

\hryschybou*

\begin{remark}
\label{rem: tilde notation}In the~following proof, and in the~proof of Proposition \ref{prop: when tilde s = bar s},  we will be working with regrets, average payoffs and other quantities related to matrix games with bounded distortion, in which we have two sets of rewards -- the~``precise'' rewards $v^M_{ij}$ corresponding to the~matrix M and the~``distorted'' rewards $v^{\widetilde M}_{ij}(t)$. We denote the~variables related to the~distorted rewards $v^{\widetilde M}_{ij}(t)$ with the~superscript `${\widetilde M}$' (for example $g^{\widetilde M}_{\max}(t)=\max_i \frac 1 t \sum_{s=1}^t v^{\widetilde M}_{i j(s)}(s)$) and use the~superscript `$M$` for variables related to $v^M_{ij}$ (for example $g^M_{\max}(t)= \max_i \frac 1 t \sum_{s=1}^{t}v^M_{i j(s)}$).
\end{remark}

\begin{proof}[Proof of Proposition \ref{Prop: hry s chybou}]
This proposition strengthens the~result of Lemma \ref{L: HC-and-NE} and its proof will also be similar. The~only additional technical ingredient is the~inequality \eqref{eq:gmax-tildegmax}.

Since $(\widetilde M)$ is a~repeated game with $c\epsilon$-bounded distortion, there almost
surely exists $t_{0}$, such that for all $t\geq t_{0}$, $\,\left|v^{\widetilde M}_{ij}(t)-v^M_{ij}\right|\leq c\epsilon$
holds. This leads to
\begin{equation}
\left|g^{\widetilde M}_{\max}(t)-g^M_{\max}(t) \right| \leq \underset{i}{\max}\,\left|\frac{1}{t}\underset{s=1}{\overset{t}{\sum}}\left(v^{\widetilde M}_{ij(s)}(s)-v^M_{ij(s)}\right)\right|\leq\frac{t_{0}}{t}+c\epsilon\cdot\frac{t-t_{0}}{t}\overset{t\rightarrow\infty}{\longrightarrow}c\epsilon.\label{eq:gmax-tildegmax}
\end{equation}
The remainder of the~proof contains no new ideas, and it is nearly exactly the~same as the~proof of Lemma \ref{L: HC-and-NE}, therefore we just note what the~two main steps are:

\noindent \textbf{Step 1:} Hannan consistency gives us that
\begin{equation*}
\forall \eta>0\ \exists t_0\in\mathbb N\  \forall t\geq t_0:\ g^{\widetilde M}\left(t\right) \geq v^M -\left(\epsilon+\eta\right)-\left|g^{\widetilde M}_{\max}\left(t\right)-g^M_{\max}\left(t\right)\right| \textrm{ holds a.s.},
\end{equation*} 
from which we deduce the~inequalities
\begin{equation*}
v^M - ( \epsilon+c\epsilon) \leq \underset{t\rightarrow\infty}{\liminf}\, g^{\widetilde M}(t) \leq \underset{t\rightarrow\infty}{\limsup}\, g^{\widetilde M}(t) \leq v^M+\epsilon+c\epsilon.
\end{equation*}

\noindent \textbf{Step 2:} For contradiction we assume that there exists an~increasing sequence of time steps $t_{n}\nearrow\infty$, such that 
\begin{equation}
u^M_1(br,\hat{\sigma}_{2}(t_{n}))\geq v^M+2(c+1)\epsilon+\eta \label{eq: contradiction}
\end{equation}
holds for some $\eta>0$. Using the~identity $g^M_{\max}\left(t_{n}\right)=u^M_1\left( br,\hat{\sigma}_{2}(t_{n}) \right)$ and inequalities \eqref{eq:gmax-tildegmax} and \eqref{eq: contradiction}, we then compute that the~regret $r^{\widetilde M}(t_n)$ is too high, which completes the~proof:
\begin{align*}
\epsilon \geq \limsup_{t\rightarrow\infty} r^{\widetilde M}\left(t\right)
 &	\geq \limsup_{n\rightarrow\infty} r^{\widetilde M}\left(t_{n}\right) = \limsup_{n\rightarrow\infty} \left(g^{\widetilde M}_{\max} \left(t_{n}\right)- g^{\widetilde M}\left(t_{n}\right)\right) \\
 & = \limsup_{n\rightarrow\infty} \left( g^{M}_{\max} \left(t_{n}\right) -\left( g^{M}_{\max} \left(t_{n}\right) - g^{\widetilde M}_{\max} \left(t_{n}\right) \right) - g^{\widetilde M}\left(t_{n}\right)\right) \\
 &	\geq \limsup_{n\rightarrow\infty} u^M_1\left( br,\hat{\sigma}_{2} (t_{n}) \right) - \limsup_{n\rightarrow\infty} \left| g^M_{\max}(t_{n}) - g^{\widetilde M}_{\max} (t_{n}) \right| - \limsup_{n \rightarrow\infty} g^{\widetilde M}(t_{n}) \\
 &	\overset{\eqref{eq:gmax-tildegmax},\eqref{eq: contradiction}}{\geq} \left(v^M+2(c+1)\epsilon+\eta			\right) - c\epsilon - \left(v^M+\epsilon+c\epsilon\right) \\
 & 	= \epsilon+\eta > \epsilon.
\end{align*}
\end{proof}

\subsection{Convergence of SM-MCTS}
In Section \ref{sec: SM-MCTS convergence}, we stated Theorem \ref{thm: SM-MCTS convergence} which claimed that if a SM-MCTS algorithm uses a selection function which is both $\epsilon$-UPO and $\epsilon$-HC, it will find an approximate equilibrium. To finish the proof of Theorem \ref{thm: SM-MCTS convergence}, it remains to prove Proposition \ref{prop: when tilde s = bar s}, which establishes a~connection between regrets $R^h$ and $R^{M_h}$ of the~selection function with respect to rewards $x^h_i(t)$ , resp. w.r.t. the~matrix game $M_h=(v^h_{ij})_{ij}$. The~goal is to show that if $R^h(T)$ is small and algorithm $A$ is $\epsilon$-UPO, then the~regret $R^{M_h}(T)$ is small as well.

\whentildesbars*

\begin{proof}
Let $\epsilon$, $c$, $A$ and $h$ be as in the~proposition. Since in this proof we will only deal with actions taken in $h$, we will temporarily omit the~upper index `$h$' in variables $t^h_i$, $t^h_j$, $t^h_{ij}$, $i^h(t)$ and $j^h(t)$. As in Remark \ref{rem: tilde notation}, $R^{M_h}$ is the~regret of player 1 with respect to the~matrix game $M_h=\left(v^h_{ij}\right)_{ij}$:
\[ R^{M_h}\left(T\right)=\max_{i^*}\sum_{t=1}^{T}v^h_{i^*j(t)}-\sum_{t=1}^{T}v^h_{i(t)j(t)}=:\max_{i^*}S^{M_h}_{i^*}(T) . \]
Recall that by \eqref{eq:reward_for_action}, we have $x^h_i \left(t\right) = x^h_{ij(t)} \left(t_{ij(t)} \right)$.
This allows us to rewrite the~regret $R^{\widetilde M_h}=R^h$ that player 1 is attempting to minimize:
\begin{align*}
R^h\left(T\right) \overset{\text{def.}}= & \max_{{i^*}} \sum_{t=1}^{T} x^h_{i^*} \left( t \right) - \sum_{t=1}^{T} x^h(t) \\
 \overset{\eqref{eq:reward_for_action}}= & \max_{{i^*}} \sum_{t=1}^{T} x^h_{{i^*}j(t)} \left(t_{{i^*}j(t)}\right) - \sum_{t=1}^{T} x^h_{i(t)j(t)}\left(t_{i(t)j(t)}\right) \\
 =: & \max_{i^*}S^h_{i^*}(T).
 \end{align*}
Since $A$ is $\epsilon$-HC, we have $\limsup_{T}R^h\left(T\right)/T\leq\epsilon$ a.s. and our goal infer that $\limsup_{T}R^{M_h}\left(T\right)/T$ $\leq2$ $\left(c+1\right)\epsilon$ holds a.s. as well. Let $i^*$ be an~action of player 1. Note the~following two facts:
\begin{align*}
& T_{j} = \sum_{m=1}^{T_{i^*j}}w^h_{i^*j}\left(m\right) && \text{...by definition of $w^h_{ij}\left(n\right)$,} \tag{a} \\
& \left( \ j(t) = j \ \& \ t_{i^*j(t)} = m \ \right)\implies x^h_{i^*j(t)}(t_{i^*j(t)})=x^h_{i^*j}(m) && \text{...by \eqref{eq:reward_for_action}} \tag{b}
\end{align*}
We can rewrite $S^h_{i^*}\left(T\right)$ as follows:
\begin{align*}
S^h_{i^*}\left(T\right) = & \sum_{t=1}^{T}x^h_{i^*j(t)}\left(t_{i^*j(t)}\right)-\sum_{t=1}^{T}					x^h_{i(t)j(t)}\left(t_{i(t)j(t)}\right)\\
\overset{(b)}{=} & \sum_{j}\sum_{m=1}^{T_{i^*j}}x^h_{i^*j}\left(m\right)\left|\left\{ t\leq T|\, 				t_{i^*j(t)} = m\,\&\, j(t)=j\right\}\right| - \sum_{i,j}\sum_{m=1}^{T_{ij}}x^h_{ij}\left(m\right)\\
	\overset{(a)}{=} & \sum_{j}\frac{T_{j}}{\sum_{m=1}^{T_{i^*j}}w_{i^*j}\left(m\right)} \sum_{m=1}^{ 			T_{i^*j}} x^h_{i^*j}\left(m\right) * w_{ij}\left(m\right)-\sum_{i,j}\frac{T_{ij}}{T_{ij}}\sum_{m=1}			^{T_{ij}} x^h_{ij} \left(m\right)\\
 = & \sum_{j}T_{j}\tilde{x}^h_{i^*j}\left(T_{i^*j}\right)-\sum_{i,j}T_{ij}\bar{x}^h_{ij}\left(T_{ij}\right)\\
 = & \sum_{j}T_{j}\left(v^h_{i^*j}+\tilde{x}^h_{i^*j}\left(T_{i^*j}\right)-v^h_{i^*j}\right)-\sum_{i,j}T_{ij}		\left(v^h_{ij}+\bar{x}^h_{ij}\left(T_{ij}\right)-v^h_{ij}\right)\\
 =: & \sum_{j}\sum_{m=1}^{T_{j}}v^h_{i^*j}-\sum_{i,j}\sum_{m=1}^{T_{ij}}v^h_{ij}+X^h_{i^*}\left(T\right)\\
 = & \sum_{t=1}^{T}v^h_{i^*j(t)}-\sum_{t=1}^{T}v^h_{i(t)j(t)}+X^h_{i^*}\left(T\right),\\
 = & S^{M_h}_{i^*}\left(T\right)+X_{i^*}\left(T\right)
\end{align*}
where 
\[ X^h_{i^*}\left(T\right) := \sum_{j}T_{j}\left(\tilde{x}^h_{i^*j}\left(T_{i^*j}\right)-v^h_{i^*j}\right)-\sum_{i,j}T_{ij}\left(\bar{s}^h_{ij}\left(T_{ij}\right)-v^h_{ij}\right). \]
In particular, we can bind the~regret $R^h(T)$ as
\[
R^{M_h} = \max_{i^*} S^{M_h}_{i^*}(T) =\max_{i^*} \left( S^h_{i^*}(T) - X^h_{i^*}(T) \right) \leq R^h(T) + \max_{i^*} \left|X^h_{i^*}(T)\right|. \]
Clearly $X^h_{i^*}\left(T\right)$ satisfies 
\begin{equation*}
\left|\frac{X^h_{i^*}\left(T\right)}{T}\right| \leq \sum_{j}\frac{T_{j}}{T}\left|\tilde{x}^h_{i^*j}\left(T_{i^*j}\right)-v^h_{i^*j}\right|+\sum_{i,j}\frac{T_{ij}}{T}\left|\bar{x}^h_{ij}\left(T_{ij}\right)-v^h_{ij}\right|.
\end{equation*}
Using the~$\epsilon$-UPO property and the~assumption that $\limsup_{n}\,\left|\bar{x}^h_{ij}\left(n\right)-v^h_{ij}\right|\leq c\epsilon$
holds a.s. for each $i,j$, we get 
\begin{equation*} \begin{split}
\limsup_{T\rightarrow\infty}\left|\frac{X^h_{i^*}\left(T\right)}{T}\right|
\leq \ & \limsup_{T\rightarrow\infty}\sum_{j}\frac{T_{j}}{T}\left|\tilde{x}^h_{i^*j}\left(T_{i^*j}\right)-v^h_{i^*j}\right|+\\
 & \ + \ \sum_{i,j}\frac{T_{ij}}{T}\left|\bar{x}^h_{ij}\left(T_{ij}\right)-v^h_{ij}\right|\\
 \leq \ & \left(c+1\right)\epsilon\limsup_{T\rightarrow\infty}\sum_{j}\frac{T_{j}}{T}+\\
 & \ + c\epsilon\limsup_{T\rightarrow\infty}\sum_{i,j}\frac{T_{ij}}{T}=\left(2c+1\right)\epsilon.
\end{split} \end{equation*}
 Consequently, this implies that 
\begin{equation*} \begin{split}
\limsup_{T\rightarrow\infty} R^{M_h}\left(T\right)/T \ \leq \ & \limsup_{T\rightarrow\infty} R^h\left(T\right)/T+ \limsup_{T\rightarrow\infty}\max_{i^*}X^h_{i^*}\left(T\right)/T \\
\leq \ & \epsilon+\left(2c+1\right)\epsilon=2\left(c+1\right)\epsilon
\end{split} \end{equation*}
holds almost surely, which is what we wanted to prove.
\end{proof}

\section{The proof of Proposition \ref{prop:expl_removal}}

\explRemoval*

\begin{proof}[Sketch]
By Lemma~\ref{lem:emp_and_avg}, it suffices to prove the result for the empirical frequencies $\hat \sigma(t)$ and $\hat \mu(t)$.
By definition of exploitability, the first case follows from \eqref{eq: HC and NE 1} in Lemma~\ref{L: HC-and-NE} (using $\epsilon=0$).
Since the proof of the second case is similar to that of \eqref{eq: HC and NE 1} in Lemma~\ref{L: HC-and-NE}, we only give the main ideas. We focus on bounding the second player's exploitability.

The first ingredient is that $u^M_1(br,\hat\mu_2(t))$ is equal to $g^{M,A}_{\max}(t)$, the analogy of $g^M_{\max}(t)$ where we only consider the~iterations where the second player didn't explore.
Secondly, using the $\gamma$-Hannan-consistency of player 1's $A^\gamma$, we bound the~asymptotic difference between $g^M_{\max}(t)$ and $g^M(t)$.
Since the second player explores independently of first player's algorithm, the~same bound holds for $g^{M,A}_{\max}(t)$ and $g^{M,A}(t)$.
Thirdly, the Hannan-consistency of player 2's $A$ implies that $g^{M,A}(t) \geq v^M $ holds in the limit.

Finally, we combine these observations and show that if the exploitability of $\hat \mu_2(t)$ was higher than $\gamma$ in the limit -- that is, if $u^M_1(br,\hat\mu_2(t))$ was higher than $v^M + \gamma$ -- then the regret $g^{M,A}_{\max}(t) - g^{M,A}(t) = u^M_1(br,\hat\mu_2(t)) - g^{M,A}(t)$ of player 1 would be higher than $\gamma$.
\end{proof}

\end{document}